\documentclass[%
 reprint,
 superscriptaddress,
 amsmath,amssymb,
 aps,
]{revtex4-2}
\usepackage{graphicx}% Include figure files
\usepackage{bm}% bold math

\usepackage{wasysym}
\usepackage{enumerate}
\usepackage{tikz}
\usetikzlibrary{quantikz2}
\usepackage{amsfonts,amsthm}
\usepackage{physics}
\usepackage{bbm}
\usepackage[unicode=true, colorlinks=true, linkcolor=blue, citecolor=red]{hyperref}
\usepackage[capitalize,sort&compress]{cleveref}
\usepackage{appendix}
\usepackage{mathtools}
\usepackage{titlesec}

\usepackage{floatrow}
\usepackage[caption=false]{subfig}
\makeatletter
\newtheorem*{rep@theorem}{\rep@title}
\newcommand{\newreptheorem}[2]{%
\newenvironment{rep#1}[1]{%
 \def\rep@title{#2 \ref{##1}}%
 \begin{rep@theorem}}%
 {\end{rep@theorem}}}
\makeatother

\newtheorem{theorem}{Theorem}
\newreptheorem{theorem}{Theorem}

\newtheorem{lemma}{Lemma}
\newreptheorem{lemma}{Lemma}

\newtheorem{proposition}{Proposition}
\newreptheorem{proposition}{Proposition}

\newtheorem{corollary}{Corollary}
\newreptheorem{corollary}{Corollary}

\theoremstyle{definition}
\newtheorem{definition}{Definition}
\newreptheorem{definition}{Definition}

\let\oldsection\section
\renewcommand{\section}[1]{\textit{#1.---}}

\newcounter{suppeqnstart}
\newcounter{suppfigstart}
\newcounter{supptabstart}
\newcounter{suppthmstart}
\newcounter{suppdefstart}
\newcounter{supppropstart}
\newcounter{supplemstart}
\newcounter{suppcorstart}

\makeatletter
\let\latex@title\title
\newcommand{\savedtitle}{}
\renewcommand{\title}[1]{%
  \gdef\savedtitle{#1}%
  \latex@title{#1}%
}
\makeatother

\newcommand{\beginsupplement}{
    \clearpage
    \widetext
    \renewcommand{\section}{\oldsection}
    \begin{center}
    \textbf{\large Supplemental Material: \savedtitle}
    \end{center}
    \setcounter{section}{0}
    \renewcommand\thesection{\arabic{section}}
    \setcounter{page}{1}
    
    \renewcommand{\bibnumfmt}[1]{[S##1]}
    \renewcommand{\citenumfont}[1]{S##1}

    \setcounter{suppeqnstart}{\value{equation}}%
    \addtocounter{suppeqnstart}{1000}%
    \setcounter{equation}{\value{suppeqnstart}}%
    \renewcommand{\theequation}{S\number\numexpr\value{equation}-\value{suppeqnstart}\relax}%

    \setcounter{suppfigstart}{\value{figure}}%
    \addtocounter{suppfigstart}{1000}%
    \setcounter{figure}{\value{suppfigstart}}%
    \renewcommand{\thefigure}{S\number\numexpr\value{figure}-\value{suppfigstart}\relax}%

    \setcounter{supptabstart}{\value{table}}%
    \addtocounter{supptabstart}{1000}%
    \setcounter{table}{\value{supptabstart}}%
    \renewcommand{\thetable}{S\number\numexpr\value{table}-\value{supptabstart}\relax}%

    \setcounter{suppthmstart}{\value{theorem}}%
    \addtocounter{suppthmstart}{1000}%
    \setcounter{theorem}{\value{suppthmstart}}%
    \renewcommand{\thetheorem}{S\number\numexpr\value{theorem}-\value{suppthmstart}\relax}%

    \setcounter{supplemstart}{\value{lemma}}%
    \addtocounter{supplemstart}{1000}%
    \setcounter{lemma}{\value{supplemstart}}%
    \renewcommand{\thelemma}{S\number\numexpr\value{lemma}-\value{supplemstart}\relax}%

    \setcounter{suppdefstart}{\value{definition}}%
    \addtocounter{suppdefstart}{1000}%
    \setcounter{definition}{\value{suppdefstart}}%
    \renewcommand{\thedefinition}{S\number\numexpr\value{definition}-\value{suppdefstart}\relax}%

    \setcounter{suppcorstart}{\value{corollary}}%
    \addtocounter{suppcorstart}{1000}%
    \setcounter{corollary}{\value{suppcorstart}}%
    \renewcommand{\thecorollary}{S\number\numexpr\value{corollary}-\value{suppcorstart}\relax}%

    \setcounter{supppropstart}{\value{proposition}}%
    \addtocounter{supppropstart}{1000}%
    \setcounter{proposition}{\value{supppropstart}}%
    \renewcommand{\theproposition}{S\number\numexpr\value{proposition}-\value{supppropstart}\relax}%

    \crefalias{section}{suppsection}%
    \crefformat{suppsection}{Section~##2##1##3}%
    \Crefformat{suppsection}{Section~##2##1##3}%
}

\crefformat{suppsection}{the~#2SM#3}
\Crefformat{suppsection}{the~#2SM#3}

\newcommand{\CNOT}{\mathrm{CNOT}}
\newcommand{\SWAP}{\mathrm{SWAP}}
\newcommand{\SUM}{\mathrm{SUM}}

\newcommand{\GLkF}{\mathsf{GL}_k(\mathbb{F}_2)}
\newcommand{\GLkFq}{\mathsf{GL}_k(\mathbb{F}_q)}
\newcommand{\PSLkFq}{\mathsf{PSL}_k(\mathbb{F}_q)}
\newcommand{\SLkFq}{\mathsf{SL}_k(\mathbb{F}_q)}
\newcommand{\GLmF}{\mathsf{GL}_m(\mathbb{F}_2)}
\newcommand{\GLx}{\mathsf{GL}_4(\mathbb{F}_2)}
\newcommand{\PG}{\mathrm{PG}(3, 2)}

\newcommand{\ee}{\mathrm{e}}
\newcommand{\ii}{\mathrm{i}}
\newcommand{\overbar}[1]{\mkern 1.5mu\overline{\mkern-1.5mu#1\mkern-1.5mu}\mkern 1.5mu}

\newcommand{\G}{G}
\newcommand{\N}{N}
\newcommand{\GSUB}{\mathcal{G}}
\newcommand{\NSUB}{\mathcal{N}}
\newcommand{\Gq}{G_q}
\newcommand{\Nq}{N_q}

\DeclareMathOperator{\RM}{RM}
\DeclareMathOperator{\QRM}{QRM}
\DeclareMathOperator{\Span}{Span}
\DeclareMathOperator{\Aut}{Aut}
\DeclareMathOperator{\PAut}{PAut}
\DeclareMathOperator{\Eval}{Eval}
\DeclareMathOperator{\End}{End}
\DeclareMathOperator{\Stab}{Stab}
\DeclareMathOperator{\diag}{diag}

\usepackage{xcolor}
\usepackage{pdfrender}
\usepackage[outline]{contour}
\usetikzlibrary{calc}
\usetikzlibrary{arrows.meta}

\colorlet{XFaceColor}{red!85!black}
\colorlet{XTextColor}{red!90!black}
\colorlet{ZFaceColor}{blue!65!black}
\colorlet{ZTextColor}{blue!75!black}
\contourlength{0.9pt}

\begin{document}

\title{Constraints on phantom codes from automorphism group bounds}

\author{Arthur S. Morris}
\email{armo@math.ku.dk}
\affiliation{Department of Mathematical Sciences, University of Copenhagen, 2100 Copenhagen, Denmark}
\author{Daniel Malz}
\affiliation{Department of Mathematical Sciences, University of Copenhagen, 2100 Copenhagen, Denmark}
\affiliation{Department of Physics, University of Basel, 4056 Basel, Switzerland}

\date{\today}

\begin{abstract}
Executing a logical quantum circuit fault-tolerantly incurs a large spacetime overhead. Recent work has proposed and investigated \emph{phantom codes}, defined by the property that every in-block logical $\CNOT$ circuit can be implemented with a physical permutation, a property that has the potential to greatly reduce the depth of compiled circuits. Here we show that phantomness comes at the cost of low encoding rate. Specifically, we prove that any binary phantom code encoding $k$ logical qubits into $n$ physical qubits with distance $d\geq 2$ obeys the bound $k\leq \log_2(n+1)$ for all $k\neq 4$. For $k=4$ we explicitly construct a nonstabiliser $(\!(8, 2^4, 2)\!)$ phantom code that violates the bound and has a transversal non-Clifford gate. We further show that, within the class of nontrivial CSS phantom codes with $k\neq 4$, there is a unique family of codes saturating this bound. In addition, we prove that this logarithmic ceiling cannot be circumvented by permitting additional local unitary gates, or by making use of subsystem codes: any subspace or subsystem code admitting a $\SWAP$-transversal implementation of every logical $\CNOT$ circuit is constrained to satisfy the same bound. These bounds follow from a general theorem relating the length of a quantum code to the structure of its automorphism group, a result which may find applications beyond phantom codes.
\end{abstract}

\maketitle

\section{Introduction}Quantum computers are highly susceptible to environmental noise. To enable the fault-tolerant execution of deep quantum circuits, quantum error-correcting (QEC) codes must be used to protect logical quantum information from decoherence. The construction of useful QEC codes is, therefore, a problem of significant practical relevance to the development of both near-term and future large-scale devices \cite{Terhal_2015}. 

QEC incurs overhead in both space (number of required physical qubits) and time, collectively referred to as spacetime overhead, which depends sensitively on the code used.
For a code encoding $k$ logical qubits into $n$ physical qubits, the overhead depends on the \emph{rate} $R=k/n$ of the code, as well as the time for syndrome extraction and fault-tolerant implementations of logical gates.
Recent advances in the theory of quantum low-density parity-check (qLDPC) codes~\cite{breuckmann_quantum_2021} have produced a number of promising constructions, including high-rate planar~\cite{l4mx-l3xx, liang2025planarquantumlowdensityparitycheck, Bravyi_2024} and asymptotically good (constant rate and relative distance) qLDPC codes~\cite{panteleev_asymptotically_2022, breuckmann_balanced_2021, leverrier_quantum_2022, dinur_good_2022, lin_good_2022}, that offer efficient encoding of logical states while maintaining low-overhead syndrome extraction.
More recently, progress has been made on constructing high-rate codes with transversal non-Clifford gates~\cite{he_asymptotically_2025, golowich_asymptotically_2024}, including in qLDPC codes~\cite{golowich_quantum_2024, lin_transversal_2024}.
While sparse high-rate codes reduce the overhead associated with memory, their complex structure can make it difficult to efficiently apply targeted logical gates; indeed, such gates are subject to fundamental limitations~\cite{Guyot2025}.
This has sparked efforts to find general constructions to fault-tolerantly implement addressable logical gates~\cite{Cross2024,Swaroop2024,He2025,Liu2026b,Chang2026,Williamson2026}.

\begin{figure*}
\centering

\sidesubfloat[]{\begin{quantikz}[row sep = {0.6cm,between origins}]
     & \gate[5]{V} &  \permute{2, 1, 5, 3, 4} &  \\
     & & & \\
     & & & \\
     \setwiretype{n} & & \setwiretype{q} & \\
     \setwiretype{n} & & \setwiretype{q} & \\
\end{quantikz}  \raisebox{.5em}{$=$} \begin{quantikz}[row sep = {0.6cm,between origins}]
     &\targ{1}  & \ctrl{2} & \gate[5]{V} &  \\
     & \ctrl{-1} & &  & \\
     & & \targ{-2} & &\\
     \setwiretype{n} & & & & \setwiretype{q} \\
     \setwiretype{n} & & & & \setwiretype{q} \\
\end{quantikz}}
\hspace{1cm}
\sidesubfloat[]{\begin{quantikz}[row sep = {0.6cm,between origins}]
     & \gate[5]{V} & \permute{2, 4, 1, 5, 3} & \gate[1,disable auto
height][.65cm][.4cm]{U_1} & \\
     & & & \gate[1,disable auto
height][.65cm][.4cm]{\phantom{U_2}} & \\
     & & & \gate[1,disable auto
height][.65cm][.4cm]{\phantom{U_2}\hspace{-1.1em}\raisebox{.5em}{\vdots}} & \\
     \setwiretype{n} & & \setwiretype{q} & \gate[1,disable auto
height][.65cm][.4cm]{\phantom{U_2}} & \\
     \setwiretype{n} & & \setwiretype{q} & \gate[1,disable auto
height][.65cm][.4cm]{U_n} & \\
\end{quantikz} \raisebox{.5em}{$=$} \begin{quantikz}[row sep = {0.6cm,between origins}]
     &  & \targ{1} & \gate[5]{V} &  \\
     & \ctrl{1} & \ctrl{-1}&  & \\
     & \targ{2}& & &\\
     \setwiretype{n} & & & & \setwiretype{q} \\
     \setwiretype{n} & & & & \setwiretype{q} \\
\end{quantikz}}
\caption{\textbf{Phantom codes.} When commuted through the encoding isometry $V$, physical permutations may induce an operation on the encoded logical qubits. \textbf{(a)} Phantom codes of $k$ logical qubits are defined by the property that every logical $\CNOT$ circuit can be implemented by a permutation of the $n$ physical qubits. \textbf{(b)} In phantom-LU codes, every logical $\CNOT$ circuit can be implemented by a local unitary (LU) gate and/or permutation of the physical qubits. One of our key results is that both phantom and phantom-LU codes are subject to the bound $k\leq\log_2(n+1)$ whenever $k\geq 2$, $k\neq 4$.}
    \label{fig:Phantom}
\end{figure*}
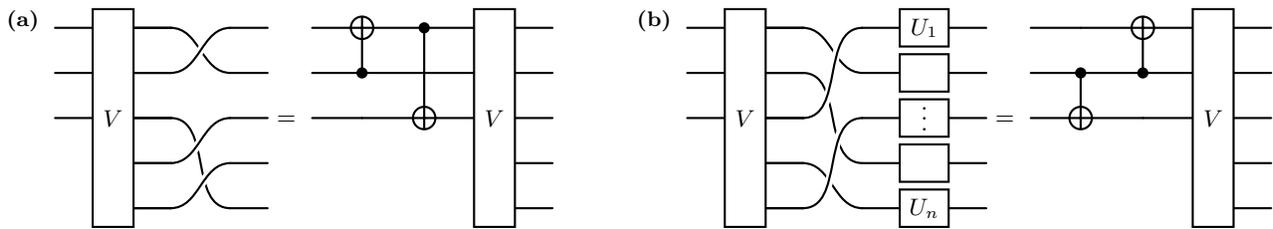

One promising way in which certain logical gates can be implemented cheaply is via code automorphisms, which are permutations of the physical qubits, possibly followed by a local unitary, that leave the code space invariant~\cite{Calderbank1996a,Grassl2013,vf7v-cpq9}. Since permutations can be realised by relabelling the qubits, such circuits are, in principle, almost free to implement, with only the unitary component requiring a physical operation. In practice, however, relabelling alters the connectivity of the device and thus automorphism gates may be most applicable to platforms with long-range connectivity or shuttling, such as trapped ions or neutral atoms~\cite{Bruzewicz_2019,Henriet_2020,Evered_2023,Majidy_Wilson_Laflamme_2024}. Automorphism gates can also be employed in high-rate codes~\cite{Xu2025b,Malcolm2025,yang2026spacetimeefficienthardwarecompatiblecomplexquantum} and have already been demonstrated experimentally~\cite{Hong2024,Reichardt2024}.
This motivates the search for families of codes with large automorphism groups~\cite{berthusen2025simplelogicalquantumcomputation,macaree2026exhaustiveoptimisationautomorphismgroups}.

In this vein, Ref.~\cite{koh2026entanglinglogicalqubitsphysical} has introduced \textit{phantom codes}, in which every in-block logical $\CNOT$ circuit can be implemented using permutations of physical qubits. These codes offer a powerful means of reducing the depths of compiled logical circuits, allowing for dramatic improvements in logical qubit fidelity. However, they are not without drawbacks: Ref.~\cite{Guyot2025} shows that a CSS phantom code of distance $d\geq 3$ cannot have constant rate $R=\Theta(1)$, and more generally all CSS phantom codes are limited by the distance-dependent bound $k\leq d\log_2 n +\text{const.}$~\cite{koh2026entanglinglogicalqubitsphysical}. While these results alone do not rule out the existence of phantom codes with high (but sub-constant) encoding rates, empirical evidence from exhaustive searches and infinite code families suggests a stronger constraint, namely that the logical dimension $k$ of phantom codes appears to be at most logarithmic in the block length for any distance~\cite{koh2026entanglinglogicalqubitsphysical}. Until now, it has been unclear whether this behaviour is an artefact of known constructions, or a fundamental obstruction.

Here, we resolve this open question by proving that a (not necessarily stabiliser) phantom code of length $n$ can encode at most logarithmically many qubits, $k\leq \log_2(n+1)$, whenever $k\neq 4$. This bound is tight, as evidenced by the existence of a family of $[\![2^k-1, k, 2]\!]$ codes, known as punctured hypercube codes, that we show are unique among CSS codes with distance $d>1$. On the other hand, for $k=4$ we explicitly construct an $(\!(8, 2^4, 2)\!)$ nonstabiliser phantom code~\cite{Note50} that violates the bound and additionally has a transversal non-Clifford gate. Furthermore, we show that the limitation imposed by the bound cannot be evaded by permitting local unitary (LU) gates in addition to physical permutations: any subspace or subsystem code admitting an implementation of every $\CNOT$ circuit via $\SWAP$-transversal operations must obey the same bound.

These results are based on a general structural theorem relating the length of a code to the constitution of its automorphism group. This bound accompanies a variety of other broad no-go theorems providing constraints on the set of shallow logical gates implementable in a given QEC code~\cite{zeng2007transversalityversusuniversalityadditive,Eastin_2009,Bravyi_2010,Bravyi_2011, bravyi_classification_2013,Pastawski_2015,Flammia_2017,chakraborty2026nogotheoremfaulttolerant, tansuwannont2026constructionlogicalcliffordgroup}. In particular, the Eastin--Knill theorem states that the group of logical operations implemented by transversal gates must be discrete, hence cannot be universal~\cite{Eastin_2009}, and it was recently shown that to implement the whole logical Clifford group in an $[\![n,k]\!]$ code one needs at least $k$-fold transversal Clifford gates~\cite{chakraborty2026nogotheoremfaulttolerant, tansuwannont2026constructionlogicalcliffordgroup}. Due to the wealth of existing codes~\cite{QEC-zoo} and the rate at which new ones are discovered, such results are of key importance for establishing viable research directions in quantum coding theory. 

We first present our theorem on automorphism groups before discussing its application to phantom codes.

\section{Phantom codes}\label{Sec:PhantomCodes}We now define binary phantom codes and introduce their LU and subsystem variants. We begin by recalling the definition of a logical operator of a QEC code. Let $\mathcal{H}\cong (\mathbb{C}^2)^{\otimes n}$ be the Hilbert space of $n$ physical qubits. A QEC code of length $n$ encoding $k$ qubits is a subspace $Q=\Im V\leq \mathcal{H}$ defined by an isometry $V:\mathcal{A}\to\mathcal{H}$, where $\mathcal{A}\cong (\mathbb{C}^2)^{\otimes k}$ is the logical space. A logical operator on $Q$ is a unitary $L\in\mathsf{U}(\mathcal{H})$ that preserves the code space, meaning it satisfies $(\mathbbm{1}-P)LP = 0$, where $P = V V^\dagger$ is the projector onto $Q$. Equivalently, $L$ is a logical operator if and only if $V^\dagger L V \in\mathsf{U}(\mathcal{A})$. The logical operators of a code form a compact Lie group denoted $\mathcal{L}(Q)$~\cite{Eastin_2009}.

\begin{definition}[Phantom codes~\cite{koh2026entanglinglogicalqubitsphysical}]\label{Def:Phantom}
    A QEC code encoding $k$ logical qubits into $n$ physical qubits is \emph{phantom} if, for some logical basis, the logical $\overline{\CNOT}_{ab}$ gate can be implemented by physical qubit permutations for every ordered pair of logical qubits $(a, b)\in [k]^2, a\neq b$.
\end{definition}

Note that this is more general than Definition 1 of Ref.~\cite{koh2026entanglinglogicalqubitsphysical}, as we do not require the codes to be CSS.

A natural way to generalise phantom codes while ensuring small overhead is to allow single-qubit gates in addition to $\SWAP$s. Such operations, known as $\SWAP$-transversal gates, also do not spread errors between qubits. The group of $\SWAP$-transversal gates, denoted $M_n$, is a semidirect product of the group of single-qubit unitaries and the symmetric group, $M_n = \mathsf{U}(2)^n\rtimes S_n$~\cite{Note30}. $M_n$ is a compact Lie group~\cite{knapp1996}.

\begin{definition}[Phantom-LU codes]\label{Def:TransversalPhantom}
    A QEC code encoding $k$ logical qubits into $n$ physical qubits is \emph{phantom-LU} if, for some logical basis, the logical $\overline{\CNOT}_{ab}$ gate can be implemented via a combination of qubit permutations and single-qubit gates for every ordered pair of logical qubits $(a, b)\in[k]^2$, $a\neq b$.
\end{definition}

\footnotetext[30]{$M_n$ is defined as the set of unitary operators that can be written as a product $m=U \sigma$ of a permutation $\sigma\in S_n$ and a collection of single-qubit gates $U = \bigotimes_{i=1}^n u_i\in \mathsf{U}(2)^n$ (this representation is unique). The subgroup of single-qubit gates is normal in $M_n$, since conjugation by a permutation simply reorders the tensor factors: $\sigma U \sigma^{-1} = \bigotimes_{i=1}^n u_{\sigma^{-1}(i)} \in \mathsf{U}(2)^n$. By contrast, $S_n$ is not normal in $M_n$ since, for example, $X_1\sigma_{12}X_1 = X_1 X_2 \sigma_{12}\notin S_n$. Hence, $M_n$ is a semidirect product $M_n = \mathsf{U}(2)^n\rtimes S_n$.}

We can further extend the notion of phantom codes to subsystem codes. A subsystem QEC code is a subspace $Q = \Im V\leq \mathcal{H}$ defined by an isometry $V:\mathcal{A}\otimes\mathcal{B}\to\mathcal{H}$, where  $\mathcal{A}\cong (\mathbb{C}^2)^{\otimes k}$ is the logical subsystem and $\mathcal{B}\cong (\mathbb{C}^2)^{\otimes r}$ the gauge subsystem~\cite{kribs2006operatorquantumerrorcorrection, PhysRevA.73.012340, aly2006subsystemcodes, PhysRevLett.94.180501, PhysRevLett.95.230504}. A $k$-qubit logical state $\rho$ is represented in $\mathcal{A}\otimes \mathcal{B}$ as $\rho\otimes\sigma$, where $\sigma$ is any $r$-qubit state in $\mathcal{B}$, and encoded into the physical space via $V$. When $r=0$, a subsystem code reduces to an ordinary QEC code, also called a subspace code. A logical gate for a subsystem code is a unitary $L\in\mathsf{U}(\mathcal{H})$ satisfying $V^\dagger L V = U\otimes w$ for some $U\in\mathsf{U}(\mathcal{A})$ and $w\in\mathsf{U}(\mathcal{B})$. In other words, $L$ is a logical operator if it preserves both $Q$ and the tensor product structure $\mathcal{A}\otimes\mathcal{B}$. The unitary $w$ is known as a gauge operator, since it changes the logical state representative $\rho\otimes\sigma$ by transforming $\sigma$. The group of logical operators of a subsystem code is likewise denoted $\mathcal{L}(Q)$. Since it is the preimage of the compact group $\mathsf{U}(\mathcal{A})\otimes \mathsf{U}(\mathcal{B})$ under the continuous map $L\mapsto V^\dagger L V$, $\mathcal{L}(Q)$ is closed in $\mathsf{U}(\mathcal{H})$, and hence a compact Lie group.

\begin{definition}[Phantom-LU subsystem (PLUS) codes]\label{Def:STP}
    A subsystem QEC code $Q$ defined by an isometry $V: \mathcal{A}\otimes \mathcal{B}\to\mathcal{H}$ encoding $k$ logical and $r$ gauge qubits into $n$ physical qubits is \emph{phantom-LU} if, for some logical basis in $\mathcal{A}$, there exists a $\SWAP$-transversal operator $m_{ab}\in M_n$ for each ordered pair of logical qubits $(a, b)\in [k]^2, a\neq b$, such that $m_{ab} V = V(\overline{\CNOT}_{ab}\otimes w_{ab})$, where $w_{ab}\in\mathsf{U}(\mathcal{B})$. Similarly, $Q$ is said to be \textit{phantom} if it is possible to choose all $m_{ab}$ to be permutations.
\end{definition}

\cref{Def:Phantom,Def:TransversalPhantom,Def:STP} all refer to qubit QEC codes; in \cref{Sec:Qudit} we discuss an extension to qudit codes.

\section{Automorphisms of quantum codes}The key idea behind phantom codes is to use qubit permutations to implement $\CNOT$ gates. Such permutations preserve the code space and are therefore code automorphisms. More generally, the set of all $\SWAP$-transversal gates that preserve the code space $Q$ forms the automorphism group $\Aut Q$~\cite{Calderbank1996a,Grassl2013,vf7v-cpq9,Note20}.
In these terms, a QEC code is phantom-LU (respectively, phantom) precisely when its automorphism group (respectively, permutation automorphism group) contains a set of elements implementing every logical $\CNOT$ circuit. Thus, by studying the structure of automorphism groups, we can derive necessary conditions that any phantom code must satisfy.
Below, we show that subgroups of $\Aut Q$ can in particular impose non-trivial lower bounds on the length of $Q$.

\footnotetext[20]{This definition differs from the terminology of classical coding theory, where `automorphism' usually refers only to a permutation.}

\begin{definition}[Automorphism group]\label{Def:Aut}
    The automorphism group of a subsystem code $Q = \Im V\cong \mathcal{A}\otimes \mathcal{B}$ of length $n$ is $\Aut Q \coloneqq M_n\cap\mathcal{L}(Q)$, the group  of $\SWAP$-transversal logical operators of $Q$. Equivalently, $\Aut Q$ is the set of elements $m\in M_n$ such that $V^\dagger m V = U\otimes w$ for some $U\in\mathsf{U}(\mathcal{A})$ and $w\in\mathsf{U}(\mathcal{B})$. The permutation automorphism group is $\PAut Q \coloneqq S_n \cap \Aut Q$, while the group of transversal gates is $\mathcal{T}(Q)=\mathsf{U}(2)^n\cap\Aut Q$. 
\end{definition}

For any subsystem or subspace QEC code $Q$, the automorphism group $\Aut Q$ is the intersection of two compact Lie groups, namely $M_n$ and $\mathcal{L}(Q)$, and is therefore itself a compact Lie group. In the case $r=0$, when $Q$ is a subspace code, \cref{Def:Aut} reduces to $\Aut Q = \{m\in M_n\,|\, mPm^\dagger=P\}$. 

Let $p:M_n\to S_n$ be the projection homomorphism defined by $p(U\sigma)=\sigma$, and let $G\leq\Aut Q$ be a compact subgroup. The image $p(G)$ records the permutation part of the elements of $G$, while the kernel $K\coloneqq\ker(p|_G) = G\cap \mathsf{U}(2)^n$ consists of those elements of $G$ whose permutation component is trivial. Since $G/K\cong p(G)\leq S_n$, the quotient $G/K$ is naturally realised via a permutation action on the $n$ physical qubits. Because this action is faithful, the number of qubits must be at least the smallest number of points on which $G/K$ admits a faithful action. For a finite group $H$, the \textit{minimal permutation degree} of $H$, denoted $\mu(H)$, is defined as the minimal size of a set $X$ such that $H$ can act faithfully on $X$, or equivalently as the minimum value of $m$ for which an injective embedding of $H$ into $S_m$ is possible:
\begin{equation}
    \mu(H) \coloneqq \min\,\{m\in\mathbb{N} \,|\,H\hookrightarrow S_m\}.
\end{equation}
By Cayley's theorem, $\mu(H)$ is finite for all finite $H$. In the present setting, this gives $n\geq \mu(G/K)$.

The next step is to relate $\mu(G/K)$ to a more useful quotient of $G$. Suppose $N\triangleleft G$ and $S\coloneqq G/N$ is finite, non-Abelian, and simple. If one can show that $\mu(S)\leq\mu(G/K)$, then it immediately follows that $n\geq\mu(S)$. In order to establish this bound, we require a technical result that enables the direct comparison of $\mu(G/K)$ and $\mu(G/N)$. 
Given a proper normal subgroup $B$ of a finite group $H$, it is possible for $\mu(H/B)$ to be greater than $\mu(H)$, in which case $H$ is called \emph{exceptional}~\cite{easdown-1988}. Nevertheless, this possibility can be ruled out if $H/B$ is simple (has no non-trivial proper normal subgroups) and non-Abelian:

\begin{proposition}[Theorem 1,~\cite{kovacs-2000}]\label{Thm:muGN}
     If $H/B$ has no non-trivial Abelian normal subgroup, then $\mu(H/B)\leq \mu(H)$.
     In particular, if $H/B$ is simple and non-Abelian, it has no non-trivial proper normal subgroups, so $\mu(H/B)\leq \mu(H)$.
\end{proposition}

We are now ready to state and prove our first main result: a general constraint on the length of a QEC code in terms of the structure of its automorphism group. We will later utilise \cref{Thm:GeneralCase} to derive bounds on the encoding rates of phantom and phantom-LU codes. 

\begin{theorem}[Permutation-degree lower bound]\label{Thm:GeneralCase}
    Let $Q$ be a binary subspace or subsystem QEC code of length $n$. If there exist compact subgroups $N,G$ of $\Aut Q$, such that (\textit{i}) $N\triangleleft\, G$, (\textit{ii}) $S\coloneqq G/N$ is finite, non-Abelian, and simple, and (\textit{iii}) $S\neq A_5$, then $n\geq \mu(S)$.
\end{theorem}

The exception $S\neq A_5$, which is not relevant in the applications considered here, arises from the following Lemma, the proof of which is given in the End Matter.

\begin{lemma}[$A_5$ exception]\label{Lem:U2A5}
    Let $K\leq \mathsf{U}(2)^n$ be compact, and let $H\trianglelefteq K$ be a closed normal subgroup of $K$ such that $S\coloneqq K/H$ is finite, non-Abelian, and simple. Then $S\cong A_5$.
\end{lemma}

\begin{proof}[Proof of \cref{Thm:GeneralCase}.]
    Both $K$ and $N$ are normal subgroups of $G$, so \cref{prop:HK3} implies that $KN/N\trianglelefteq G/N = S$. Because $S$ is simple, it follows that $KN/N$ is either trivial or all of $S$. Equivalently, either (\textit{i}) $KN = G$, or (\textit{ii}) $K\leq N$. Suppose first that (\textit{i}) holds, so that $KN = G$. Then \cref{prop:HK3} gives
    \begin{equation}\label{Test2}
        K/(K\cap N) \cong KN/N \cong G/N = S.
    \end{equation}
    Thus $S$ is isomorphic to a quotient of $K$. Since $K$ is a closed normal subgroup of $G$ it is compact, so $K$ is isomorphic to a compact subgroup of $\mathsf{U}(2)^n$. By \cref{Lem:U2A5}, this is possible only if $S=A_5$, contrary to assumption. We can therefore assume (\textit{ii}) holds and $K\leq N$, which implies $K\trianglelefteq N\triangleleft G$. Then \cref{Prop:3rdIso} yields $(G/K)/(N/K) \cong G/N \cong S$. Hence, $(G/K)/(N/K)$ is simple and non-Abelian, allowing us to apply \cref{Thm:muGN} and conclude that 
    \begin{equation}
        n\geq \mu(G/K) \geq \mu[(G/K)/(N/K)] = \mu(S)
    \end{equation}
    as claimed.
\end{proof}

\section{Encoding rates of phantom codes}We now identify the conditions under which a code has the phantom property, and apply \cref{Thm:GeneralCase} to obtain a bound on the encoding rate of any phantom code.

To begin, observe that the set of logical $\CNOT$ circuits forms a faithful unitary representation $g\mapsto U_g$ of $\GLkF$, the group of invertible linear transformations of $\mathbb{F}_2^k$. This can be deduced, for instance, by employing the symplectic representation of the Clifford group~\cite{bataille2020quantumcircuitscnotgates, koh2026entanglinglogicalqubitsphysical}. Now consider the set of permutations of the physical qubits that act as a logical $\CNOT$ circuit, up to a global phase: 
\begin{equation}\label{eq:G}
    \G \coloneqq\left\{
   \sigma\in S_n \;\middle|\; \begin{aligned}& \sigma V = \ee^{\ii\theta}VU_g \text{ for some} \\
   & g\in\GLkF, \theta\in\mathbb{R}
  \end{aligned}
\right\}.
\end{equation}
$\G$ forms a group under the composition of permutations, and moreover carries a homomorphism $\pi:\G\to\GLkF$ defined by sending $\sigma\in \G$ to the unique element $g\in\GLkF$ satisfying $\sigma V = \ee^{\ii\theta}V U_g$ for some $\theta\in\mathbb{R}$. A code is phantom if and only if $\pi$ is surjective, meaning every logical $\CNOT$ circuit $U_g$ has a representative physical permutation $\sigma_g$. 

Importantly, while the logical $\CNOT$ circuits provide a representation of $\GLkF$ and satisfy $U_g U_{g'} = U_{gg'}$, this need not be true of the permutations: in general, $\sigma_g \sigma_{g'}$ is equivalent to $\sigma_{gg'}$ only up to multiplication by a permutation that acts trivially on the logical subspace. The set of such permutations, denoted 
\begin{equation}\label{eq:TrivialLogical}
    \N\coloneqq\ker\pi = \{\nu\in S_n \,|\,\nu V = \ee^{\ii\theta}V \text{ for some }\theta\in\mathbb{R}\},
\end{equation}
is a normal subgroup of $\G$, and the surjectivity of $\pi$ implies that $\G/\N \cong \GLkF$. Equivalently, $\G$ is an extension of $\GLkF$ by $\N$. Since $\GLkF\cong\mathsf{PSL}_k(\mathbb{F}_2)$ is simple and non-Abelian for $k\geq 3$~\cite{ConradSimple,Note3}, we obtain the following result.

\footnotetext[3]{Table 1 in Ref.~\cite{cooperstein-1978} gives the result for $\SLkFq$; note that $\mathsf{SL}_k(\mathbb{F}_2) = \GLkF$ as all invertible binary matrices have unit determinant. Similarly, $\mathsf{SL}_k(\mathbb{F}_2) = \mathsf{PSL}_k(\mathbb{F}_2)$ as the centre of $\mathsf{SL}_k(\mathbb{F}_2)$ is trivial, so the result of Ref.~\cite{ConradSimple} likewise applies to $\GLkF$.}

\begin{theorem}[Bound on phantom codes]\label{Thm:PhantomBound}
    Any phantom code of distance $d\geq 2$ encoding $k\geq 2$ logical qubits, where $k\neq 4$, requires at least $n\geq 2^k-1$ physical qubits.
\end{theorem}

\begin{proof}
    The groups $\G$ and $\N$ defined in \cref{eq:G,eq:TrivialLogical} satisfy $\N\triangleleft \G\leq \Aut Q$ and are compact subgroups of $\Aut Q$. Moreover, for $k\geq3$, $S=\G/\N\cong \GLkF$ is simple, non-Abelian, finite, and not equal to $A_5$, so \cref{Thm:GeneralCase} applies and $n\geq\mu(\GLkF)$. Cooperstein~\cite{cooperstein-1978, Note3} showed that for $k\geq 5$ or $k=3$, the minimal permutation degree of $\GLkF$ is given by $\mu\left(\GLkF\right) = 2^k - 1$, hence $n\geq 2^k-1$ whenever $k\geq 5$ or $k=3$. Finally, if $k=2$ then the quantum singleton bound~\cite{gottesman-2024} $n-k\geq 2(d-1)$ implies that any code with distance $d\geq 2$ must have $n\geq 4 > 2^k-1$, so the bound holds in this case too.
\end{proof}

\cref{Thm:PhantomBound} naturally generalises to subsystem and subspace phantom codes over Galois qudits, which we discuss in \cref{Sec:Qudit}.

The condition $k\neq 4$ is due to the exceptional isomorphism $\GLx\cong A_8$, which implies that $n\geq\mu(\GLx) = \mu(A_8) = 8$ in this case. This isomorphism is associated with an isomorphism of group actions: the action of $\GLx$ on the lines within the finite projective geometry $\PG = \mathbb{F}_2^4\backslash\{0\}$ is identical to the action of $A_8$ on the set $W\coloneqq\{b\in\mathbb{F}_2^8\,|\,|b|=4\}/\{b\sim\bar{b}\}$ of weight-4 binary strings $b$, modulo bitwise complementation. In \cref{Sec:PG} we utilise this correspondence to construct a phantom code with $k=4$ that saturates the bound $n\geq 8$, thereby proving that it is tight.

\begin{theorem}[$\PG$ code]
    There exists an explicit $(\!(8, 2^4, 2)\!)$ phantom code $Q$ with a transversal non-Clifford gate and maximal permutation automorphism group $\PAut Q = S_8$. This code is necessarily nonstabiliser, in the sense that no Pauli stabiliser $[\![8, 4]\!]$ phantom code exists.
\end{theorem}

In fact, Theorem 6 of Ref.~\cite{koh2026entanglinglogicalqubitsphysical} shows that the bound from \cref{Thm:PhantomBound} is also tight for all $k\geq 5$ and $k=3$, since it is saturated by the family of $[\![2^k-1, k, 2]\!]$ punctured hypercube codes, described in \cref{Sec:Hypercube}. Up to CSS isomorphism, this family is unique:

\begin{theorem}[Punctured hypercube codes are unique]\label{Thm:Minimal}
    For $k\geq 5$ and $k=3$, the only binary CSS phantom codes of type $[\![2^k-1, k, d]\!]$ with $d>1$ are the punctured hypercube codes (up to CSS isomorphism).
\end{theorem}

The proof of \cref{Thm:Minimal} is given in \cref{Sec:MinimalProof}. Punctured hypercube codes have large check weights scaling as $w=\Theta(n)$ and hence are not qLDPC. \Cref{Thm:Minimal} thus provides a partial answer to one of the questions raised in Ref.~\cite{koh2026entanglinglogicalqubitsphysical}, namely whether phantom codes can be qLDPC. Note, however, that this result does not preclude the existence of $[\![n>2^k-1, k]\!]$ qLDPC CSS phantom codes, or even of non-CSS qLDPC phantom codes of this type, the possibility of which remains an interesting open question.

One might expect that the additional degrees of freedom afforded by the single-qubit gates in phantom-LU codes could reduce the overhead required to implement logical $\CNOT$ circuits, or similarly that the gauge freedom in subsystem phantom codes might help in this regard. In fact, this is not the case, and the same exponential lower bound on $n$ applies to both classes of codes. We show this by proving that the bound holds for PLUS codes. By restricting to $r=0$, or demanding that each $\SWAP$-transversal operator have trivial local unitary part, it follows that the same bound holds for subspace phantom-LU and subsystem phantom codes, respectively.

The analogue of \cref{eq:G} for a PLUS code $Q\cong\mathcal{A}\otimes \mathcal{B}$ is defined as the group of $\SWAP$-transversal operations that act as logical $\CNOT$ circuits, up to a possible gauge operator in the tensor factor $\mathcal{B}$:
\begin{equation}\label{Eq:GSUB}
    \GSUB \coloneqq \left\{m\in M_n\;\middle|\; \begin{aligned}& m V = V(U_g\otimes w)\text{ for some}\\ & g\in\GLkF, w\in\mathsf{U}(\mathcal{B}) \end{aligned}\right\}.
\end{equation}
Defining a map $\Pi:\GSUB\to\GLkF$ by sending $m\in\GSUB$ to the unique $g\in\GLkF$ such that $mV=V(U_g\otimes w)$ for some $w\in\mathsf{U}(\mathcal{B})$, we set
\begin{equation}\label{Eq:NSUB}
    \NSUB\coloneqq \ker\Pi = \left\{m\in M_n\;\middle|\; \begin{aligned}&m V = V(\mathbbm{1}\otimes w)\\ & \text{for some } w\in\mathsf{U}(\mathcal{B}) \end{aligned}\right\}.
\end{equation}
The map $\Pi$ is well defined, for if $U_g\otimes w = U_{g'}\otimes w'$, then $U_g^\dagger U_{g'}\otimes \mathbbm{1} = \mathbbm{1}\otimes w (w')^\dagger.$ The left-hand side lies in $\End(\mathcal{A})\otimes \mathbbm{1}$, and the right-hand side in $\mathbbm{1}\otimes\End(\mathcal{B})$, so both must be scalar. Hence $U_{g'} = \lambda U_g$ for some $\lambda\in\mathsf{U}(1)$, which for $\CNOT$ circuits is only possible if $g=g'$. A code is phantom-LU if and only if $\Pi$ is surjective, in which case $\GSUB/\NSUB\cong \GLkF$.

In order to be able to apply \cref{Thm:GeneralCase} and obtain a bound on PLUS codes, we require a technical lemma concerning the structure of $\GSUB$ and $\NSUB$, the proof of which is given in the End Matter.

\begin{lemma}[Compactness of subgroups]\label{Lemma:GCompact}
    For any PLUS code $Q$, $\GSUB$ is a compact Lie group, and $\NSUB$ is a closed normal subgroup of $\GSUB$.
\end{lemma}

\cref{Lemma:GCompact} immediately provides the desired bound. 

\begin{theorem}[Bound on PLUS codes]\label{Thm:STPbound}
    Any subspace or subsystem phantom-LU code of distance $d\geq 2$ encoding $k\geq 2$ logical qubits, where $k\neq 4$, requires at least $n\geq 2^k-1$ physical qubits.
\end{theorem}

The proof of \cref{Thm:STPbound} is identical to the proof of \cref{Thm:PhantomBound}, only with $\G$ and $\N$ replaced by $\GSUB$ and $\NSUB$.

\section{Discussion}We have shown that the principal advantage of phantom codes, namely the ability to execute every logical $\CNOT$ circuit with code automorphisms, comes with a strong and unavoidable limitation in the binary setting: the full logical $\CNOT$ group can be implemented by permutations only at the cost of an exponentially large physical code block, which makes phantomness fundamentally incompatible with the high-rate regime targeted by modern qLDPC constructions. That this restriction also applies to phantom-LU codes and their subsystem variants indicates that this constraint is not an artefact of requiring pure permutations. Our work thus establishes a general trade-off between the achievability of phantom entangling gates and a high encoding rate, makes progress on the characterisation of phantom codes, and provides a new perspective on the extent to which large automorphism groups can reduce the spacetime overhead of fault-tolerant logical computation.

Within the class of binary CSS codes of distance $d>1$, the exponential obstruction is even more rigid. Punctured hypercube codes, the unique family of such codes saturating the bound, have check weights scaling as $\Theta(n)$, and therefore represent an extremal endpoint of the phantom-code paradigm: they maximise the number of qubits compatible with permutation-realised logical $\CNOT$ gates, but do so only by giving up sparse stabiliser structure. The important question of whether a qLDPC code could display the phantom property, or if the high-weight checks observed in all known phantom codes are necessary, is yet unresolved.

In addition to this issue, several directions remain open.
Firstly, given that both phantom and phantom-LU codes are subject to an identical bound on the encoding rate, it is natural to ask whether these two classes of codes are genuinely distinct. In \cref{Sec:Equivalence} we formulate a necessary and sufficient condition for a phantom-LU code to be locally unitary equivalent to a phantom code. Determining whether this condition holds for every phantom-LU code, or if there exist obstructions to its satisfiability, would clarify the precise relation between these two classes.
Secondly, it would be interesting to explore to what degree the trade-off between the size of the automorphism group of a code, and its encoding rate, extends to automorphisms implementing logical operations other than $\CNOT$ gates. This tension arises both in phantom codes, and codes in which only the $\CNOT$ circuits on $l<k$ qubits can be implemented with permutations; in the latter case, the number of physical qubits must be at least $n\geq 2^l-1$ whenever $l\neq 4$.
Finally, it remains to understand whether, despite their limitations, phantom codes can be practically useful in concrete specialised applications. For example, by code switching between a high-rate code and a low-rate phantom code it may be possible to utilise the unique advantages of both codes simultaneously. 

\section{Acknowledgments}We are grateful for fruitful discussions with Michael Gullans and Shayan Majidy, who suggested considering the extension to phantom-LU codes.
DM acknowledges financial support from VILLUM FONDEN via the QMATH Centre of Excellence (Grant No.10059) and from the Novo Nordisk Foundation under grant numbers NNF22OC0071934 and NNF20OC0059939. ASM acknowledges funding from a Postdoctoral fellowship of the Danish Quantum Algorithm Academy (DQAA) under the Danish e-infrastructure Consortium (DeiC), with grant number 5260-00005B.

\bibliography{apssamp}

\onecolumngrid
\vspace{.5em}
\begin{center}
    \textbf{End Matter}
\end{center}
\vspace{.3em}
\twocolumngrid

\begin{appendices}

\section{Group theory}
We use this appendix to recall some results from group theory and present technical lemmas that are utilised in the main text.

\begin{definition}
    For two subgroups $A, B\leq H$, the subset product is defined as $AB = \{a b \,|\, a\in A, b\in B\}$. 
\end{definition}

\begin{proposition}\label{prop:HK2}
    If $A$ is a normal subgroup of $H$ and $B\leq H$ then $A\trianglelefteq AB\leq H$.
\end{proposition}

\begin{proof}
    If $A\trianglelefteq H$ then $hA = Ah$ for all $h\in H$. For any subgroup $B\leq H$, this implies $bA = Ab$ for all $b\in B$, meaning $AB = BA$. Hence $(AB)(AB)=A(AB)B = AB$ and $(AB)^{-1}=B^{-1}A^{-1}=BA=AB$, so $AB\leq H$. Since $A\leq AB$ and $A$ is normal in $H$, it follows automatically that $A\trianglelefteq AB$.
\end{proof}

\begin{proposition}\label{prop:HK3}
    If both $A$ and $B$ are normal subgroups of $H$, then $AB\trianglelefteq H$ and $AB/B \trianglelefteq H/B$. Moreover, $AB/B\cong A/(A\cap B)$.
\end{proposition}

\begin{proof}
    Suppose both $A, B\trianglelefteq H$. Then for all  $h\in H$ we have $hAB = AhB = ABh$ so $AB$ is normal in $H$. Let $q:H\to H/B$ be the quotient map. Then $AB/B = q(AB) = q(A)q(B) = q(A)$ as $q(B)=1$. But $q(A)\leq H/B$ as $q$ is a homomorphism, so $AB/B\leq H/B$. To see that $AB/B$ is normal in $H/B$, note that for any $hB\in H/B$ and any $xB\in AB/B = \{x B\,|\, x\in AB\}$ we have $(hB)(xB)(h^{-1}B) = (hxh^{-1})B = x'B$. But $x'=hxh^{-1}\in AB$ as $AB\trianglelefteq H$, so $x'B\in AB/B$, meaning $AB/B\trianglelefteq H/B$. Finally, the restriction of the quotient map to $A$ has kernel $A\cap B$ and image $q(A) = AB/B$, so $A/(A\cap B)\cong AB/B$.
\end{proof}

\begin{proposition}\label{Prop:3rdIso}
    Let $G$ be a group and suppose $K\trianglelefteq N\trianglelefteq G$. Then $(G/K)/(N/K)\cong G/N$.
\end{proposition}

\begin{proof}
    Define $\phi:G/K\to G/N$ by $gK\mapsto gN$. This map is well-defined, as if $gK=g'K$ then $g^{-1}g'\in K\trianglelefteq N$, so $gN=g'N$. The kernel of $\phi$ consists of those cosets $gK\in G/K$ for which $gN = N$, which occurs if and only if $g\in N$. Hence $\ker\phi = \{nK\,|\,n\in N\}\cong N/K$. Moreover, $\phi$ is surjective as $gK\subseteq gN$ for all $gN\subseteq G$. 
\end{proof}

\begin{proposition}\label{prop:Central}
    Let $K$ be a group, let $H$ be a normal subgroup of $K$, and let $Z$ be a subgroup of the centre $\mathcal{Z}(K)$. Then $ZH/H\leq \mathcal{Z}(K/H)$, i.e. the centrality of $Z$ is preserved by taking quotients.
\end{proposition}

\begin{proof}
    Let $z\in \mathcal{Z}(K)$ be a central element, take $k\in K$ and let $q:K\to K/H$ be the quotient map. Then $q(z)q(k)=q(zk)=q(kz)=q(k)q(z)$, so $q(z)\in \mathcal{Z}(K/H)$ and $q(Z)\leq \mathcal{Z}(K/H)$ whenever $Z\leq \mathcal{Z}(K)$. Moreover, as $q(H) = 1$ we have $ZH/H = q(ZH) = q(Z)q(H) = q(Z) \leq \mathcal{Z}(K/H)$, meaning $ZH/H$ is central in $K/H$. 
\end{proof}

\begin{proposition}\label{prop:Product}
    Let $A$ and $B$ be topological Hausdorff groups, let $K\leq A\times B$ be compact, and let $H\trianglelefteq K$ be closed and such that $K/H\cong S$ is simple. Then $S$ is isomorphic either to a quotient $N/M$ where $N\leq A$ is compact and $M\triangleleft N$ is closed, or to such a quotient with $N\leq B$.
\end{proposition}

\begin{proof}
    Let $\pi_A:A\times B\to A$ and $\pi_B:A\times B\to B$ be the canonical projection maps, and let $N_A\coloneqq\ker(\pi_A|_K)\trianglelefteq K$. Since $\pi_A|_K$ is continuous and $A$ is Hausdorff, $N_A$ is closed in $K$, and hence compact because $K$ is compact. Let $q:K\to K/H \cong S$ be the quotient map.  Then $q(N_A)\trianglelefteq S$, because the image of a normal subgroup under a homomorphism is normal in the image. Since $S$ is simple, it follows that either (\textit{i}) $q(N_A) = 1$ or (\textit{ii}) $q(N_A) = S$. 
    
    In case (\textit{i}) we have $N_A\trianglelefteq H \trianglelefteq K$, so $S\cong K/H\cong (K/N_A)/(H/N_A)$ is a quotient of the subgroup $K/N_A \cong \pi_A(K)\leq A$. Since $K$ is compact and $\pi_A$ is continuous, the image $\pi_A(K)$ is a compact subgroup of $A$. Also, as $H$ is closed and hence compact in $K$, the image $H/N_A$ of $H$ under the quotient map $K\to K/N_A$ is compact in $K/N_A$. But $K/N_A$ is Hausdorff, so compact subsets are closed, hence $H/N_A$ is a closed normal subgroup of $K/N_A$. In this case $S$ is a quotient $N/M$ with $N = \pi_A(K)\leq A$ compact and $M = \pi_A(H) \trianglelefteq N$ closed.    
    
    On the other hand, in case (\textit{ii}) where $q(N_A) = S$, we have $S\cong N_A/(N_A\cap H)$. But $N_A\leq \{1\}\times B$, so $N_A\cong\pi_B(N_A)$ is topologically isomorphic to a subgroup of $B$. Hence $S$ is a quotient of a subgroup of $B$. We know that $N_A$ is compact, and moreover $N_A\cap H$ is a closed normal subgroup of $N_A$, because $H$ is closed and normal in $K$. Hence $S$ is a quotient $N/M$, where $N = \pi_B(N_A)\leq B$ is compact and $M= \pi_B(N_A\cap H) \trianglelefteq N$ is closed.
\end{proof}

\section{Proof of Lemma\texorpdfstring{~\ref{Lem:U2A5}}{ 1}}\label{Sec:Lemma1}
We now prove \cref{Lem:U2A5} by induction on $n$. We begin with the base case $n=1$, which is a slight generalisation of Theorem 26.1 from Ref.~\cite{Dornhoff}.

\begin{lemma}\label{prop:U2}
    Let $K\leq \mathsf{U}(2)$ be compact, and let $H\trianglelefteq K$ be a closed normal subgroup of $K$ such that $S\coloneqq K/H$ is finite, non-Abelian, and simple. Then $S\cong A_5$.
\end{lemma}

\begin{proof}
    Define $Z \coloneqq K\cap \mathsf{U}(1)$ to be the intersection of $K$ with the centre $\mathsf{U}(1)=\mathcal{Z}[\mathsf{U}(2)]$ of $\mathsf{U}(2)$. As $H$ is normal in $K$, it follows from \cref{prop:HK2} that $H\trianglelefteq ZH \leq K$. Since $Z$ is a subgroup of the centre of $K$, \cref{prop:Central} implies that $ZH/H\leq\mathcal{Z}(K/H) = \mathcal{Z}(S)$. But a non-Abelian simple group has trivial centre, so $\mathcal{Z}(S)=1$ and $ZH/H = 1$, which implies $Z\leq H$. As $Z$ is normal in $K$, we have $Z\trianglelefteq H\trianglelefteq K$, which implies that the quotient factors as $S = K/H \cong \overbar{K}/\overbar{H}$, where we have defined $\overbar{K} \coloneqq K/Z$ and $\overbar{H} \coloneqq H/Z$. Thus, $S$ is a quotient of $\overbar{K}$.

    Now, let $q:\mathsf{U}(2)\to \mathsf{U}(2)/\mathsf{U}(1) \cong \mathsf{SO}(3)$ be the quotient map. The restriction $\tilde{q}\coloneqq q|_K$ has kernel $Z$, so induces an isomorphism of topological groups $\overbar{K}  \equiv K/Z \cong \tilde{q}(K)$. Moreover, as $q$ is continuous, $\overbar{K}$ is compact. Similarly, $H$ is compact as it is closed in $K$, so the image $\overbar{H} \cong \tilde{q}(H)$ is a compact subgroup of $\overbar{K}$. As $\overbar{H}$ is closed in the compact Hausdorff group $\overbar{K}$, it follows that the quotient topology on the finite group $S\cong\overbar{K}/\overbar{H}$ is discrete.

    At this point we use the classification of the closed subgroups of $\mathsf{SO}(3)$: every closed subgroup is conjugate to one of $\mathbb{Z}_m$, $D_{2m}$, $T\cong A_4$, $O\cong S_4$, $I\cong A_5$, $\mathsf{SO}(2)$, $\mathsf{O}(2)$, or $\mathsf{SO}(3)$ \cite{Olive_2019, golubitsky-1988}. We now check the possibilities one-by-one.

    Firstly, suppose $\overbar{K}$ is one of the infinite closed subgroups of $\mathsf{SO}(3)$, let $\overbar{K}^\circ$ be the identity component of $\overbar{K}$, and let $p:\overbar{K}\to\overbar{K}/\overbar{H}\cong S$ be the quotient map. By definition, the quotient map $p$ is continuous, so $p(\overbar{K}^\circ)$ is a connected subgroup of $S$. However, $S$ is discrete so $p(\overbar{K}^\circ) = 1$, meaning $\overbar{K}^\circ \leq \ker p = \overbar{H}$. It follows that $S\cong \overbar{K}/\overbar{H} \cong (\overbar{K}/\overbar{K}^\circ)/(\overbar{H}/\overbar{K}^\circ)$, i.e. $S$ is a quotient of $\overbar{K}/\overbar{K}^\circ$. If $\overbar{K} \cong \mathsf{SO}(2)$ or $\mathsf{SO}(3)$ is connected, then $\overbar{K}^\circ\cong\overbar{K}$ and $\overbar{K}/\overbar{K}^\circ \cong 1$. On the other hand, if $\overbar{K}\cong \mathsf{O}(2)$ then $\overbar{K}/\overbar{K}^\circ\cong\mathbb{Z}_2$. Neither of these options is possible as $S$ is non-Abelian, therefore $\overbar{K}$ must be finite. 
    
    Of the finite closed subgroups of $\mathsf{SO}(3)$, the groups $\mathbb{Z}_m, D_{2m}, A_4$, and $S_4$ are all solvable, meaning they cannot have a non-Abelian simple quotient. The only remaining option is $A_5$, which is non-Abelian simple, meaning its only non-trivial quotient is $A_5$. Hence $S\cong\overbar{K}/\overbar{H}$ can only be $A_5$.
\end{proof}

From here, \cref{Lem:U2A5} follows straightforwardly from an inductive application of \cref{prop:Product}.

\begin{replemma}{Lem:U2A5}[$A_5$ exception]
    Let $K\leq \mathsf{U}(2)^n$ be compact, and let $H\trianglelefteq K$ be a closed normal subgroup of $K$ such that $S\coloneqq K/H$ is finite, non-Abelian, and simple. Then $S\cong A_5$.
\end{replemma}

\begin{proof}
    For $n>1$, the proof is by induction on $n$. By writing $K\leq \mathsf{U}(2)^n = A\times B$ with $A=\mathsf{U}(2)$, $B=\mathsf{U}(2)^{n-1}$ and applying \cref{prop:Product}, we have that $S\cong N/M$ is a quotient of a compact subgroup $N$ by a closed normal subgroup $M\trianglelefteq N$, where either $N\leq\mathsf{U}(2)$ or $N\leq \mathsf{U}(2)^{n-1}$. If it is the former then we are done; otherwise we have decreased $n$ by 1, so by induction we can reduce to $n=1$. Hence, by \cref{prop:U2} it follows that $S\cong A_5$.
\end{proof}

\section{Proof of \texorpdfstring{\cref{Lemma:GCompact}}{Lemma 2}}Finally, we prove that $\mathcal{G}$ and $\mathcal{N}$ have the topological properties required for the application of \cref{Thm:GeneralCase}.

\begin{replemma}{Lemma:GCompact}
    For any PLUS code $Q$, $\GSUB$ is a compact Lie group, and $\NSUB$ is a closed normal subgroup of $\GSUB$.
\end{replemma}

\begin{proof}
    Let $Q\cong \mathcal{A}\otimes \mathcal{B}$ be a PLUS code, and define a continuous map $\rho:\Aut Q\to \mathsf{PU}(\mathcal{A})$ as $\rho(m) = [U]$, where $U\in\mathsf{U}(\mathcal{A})$ is any unitary satisfying $V^\dagger m V = U\otimes w$ for some $w\in\mathsf{U}(\mathcal{B})$, and $[U]$ denotes the equivalence class of unitary matrices in $\mathsf{U}(\mathcal{A})$ differing from $U$ by a multiplicative phase. This map is well defined, for if $U\otimes w = U'\otimes w'$, then $U^\dagger U'\otimes \mathbbm{1} = \mathbbm{1}\otimes w (w')^\dagger.$ The left-hand side lies in $\End(\mathcal{A})\otimes \mathbbm{1}$, and the right-hand side in $\mathbbm{1}\otimes\End(\mathcal{B})$, so both must be scalar. Thus $U' = \lambda U$ for some $\lambda\in\mathsf{U}(1)$, meaning $[U'] = [U]$ in $\mathsf{PU}(\mathcal{A})$. We can then identify $\NSUB = \ker \rho$, meaning that $\NSUB$ is closed. Moreover, defining $\Gamma \coloneqq \{[U_g]\,|\,g\in\GLkF\}\subset \mathsf{PU}(\mathcal{A})$, we have $\GSUB = \rho^{-1}(\Gamma)$. As $\Gamma$ is finite, it is closed in $\mathsf{PU}(\mathcal{A})$; hence $\GSUB$ is closed in the compact group $\Aut Q$, and is therefore compact.
\end{proof}

\end{appendices}

\beginsupplement

\section{\texorpdfstring{$\GLmF$}{GL(m, 2)}-invariant classical codes}\label{Sec:GLMFInv}
The aim of this section is to explain the following classification, which follows from the results of Ref.~\cite{bardoe}.

\begin{proposition}[Ref.~\cite{bardoe}]\label{Lemma:Uniqueness}
    A binary linear code of length $n=2^m-1$, where $m\geq 5$ or $m=3$, has an automorphism group containing $\GLmF$ as a subgroup if and only if it is a punctured Reed-Muller code $\RM^*(r, m)$ or a shortened Reed-Muller code $\RM_*(r, m)$, where $0\leq r \leq m-1$. These codes satisfy
    \begin{equation}\label{Test}
        \Aut \RM^*(r, m) = \Aut \RM_*(r, m) = \begin{cases}
            \GLmF & 1\leq r\leq m-2, \\
            S_n & r = 0\text{ or } m-1.
        \end{cases}
    \end{equation}
\end{proposition}

The values of $r$ for which the automorphism groups are the full permutation group $S_n$ correspond to the cases in which $\RM^*(r, m)$ and $\RM_*(r, m)$ reduce to `elementary' codes: $\RM^*(m-1, m) = \mathbb{F}_2^n$ is the full space; $\RM_*(0, m) = \{0\}$ is the trivial code; $\RM^*(0, m)= C^{\text{rep}}_n$ is the repetition code of length $n$; and $\RM_*(m-1, m)=\RM^*(0, m)^\perp = C^{\text{even}}_n$ is the code consisting of all even-weight vectors of length $n$.  

An outline of this section is as follows. In \cref{Sec:GLKFAction} we discuss the faithful permutation representations of $\GLmF$ on a set of size $n=2^m-1$. In \cref{Sec:RMCodes} we review classical Reed-Muller codes and their punctured and shortened variants, in particular showing that these latter codes have length $n$ and automorphism groups containing $\GLmF$. Finally, in \cref{Sec:Uniqueness} we explain how a special case of the results from Ref.~\cite{bardoe} implies that these codes are unique in this respect. Throughout this section, we write $\Omega = \mathbb{F}_2^m$ for the $m$-dimensional vector space over $\mathbb{F}_2$, define $P = \Omega\backslash\{0\}$ to be the set of all non-zero vectors in $\Omega$, and denote the total number of points in $P$ by $n=2^m-1$.

\subsection{\texorpdfstring{$\GLmF$}{GL(m, 2)} actions}\label{Sec:GLKFAction}
Recall that $\GLmF$ is defined as the group of invertible $m\times m$ matrices over $\mathbb{F}_2$. Every matrix $A\in\GLmF$ has a natural action on $P$, defined by sending $x\in P$ to $Ax$. As the zero vector in $\Omega$ is fixed by all $A$, this gives a bijection from $P$ to itself and hence defines a faithful permutation representation $s:\GLmF\to S_n$ of degree $n=|P|=2^m-1$. Since $\mu(\GLmF) = 2^m-1$ for $m\geq 5$ and $m=3$, this action is minimal in either case~\cite{cooperstein-1978}. 

We can obtain a second minimal permutation representation of $\GLmF$ by taking the dual (contragredient) of $s$, and writing $\bar{s}(A) \coloneqq s(A^{-\intercal})$, where $A^{-\intercal} = (A^{-1})^\intercal$. This representation is naturally interpreted as the action of $\GLmF$ on the $(m-1)$-dimensional hyperplanes of $P$, given the natural interpretation of $P$ as a projective space. The following result shows that this classification of minimal faithful permutation representations of $\GLmF$ is exhaustive.

\begin{proposition}[\cite{volta-2000}, Sec. 3.3]
    The only subgroups of $\mathsf{PSL}_m(\mathbb{F}_q)$ of index $(q^m-1)/(q-1)$ are parabolic, i.e. they are the point and hyperplane stabiliser groups. Hence, for $m\geq 5$ and $m=3$ the only minimal faithful permutation representations of $\GLmF\cong\mathsf{PSL}_m(\mathbb{F}_2)$ are $s$ and $\bar{s}$.
\end{proposition}

It is straightforward to see that a classical code $C$ is invariant under $\bar{s}$ if and only if it is invariant under $s$. Writing $s(A)\cdot C$ to denote the permutation action of $s(A)$ on the codewords of $C$, invariance under $s$ means that $s(A)\cdot C = C$  for all $A\in\GLmF$. As $A\mapsto A^{-\intercal}$ is a bijection on $\GLmF$, this condition is equivalent to $s(A^{-\intercal})\cdot C = \bar{s}(A)\cdot C = C$ for all $A\in \GLmF$, so $C$ is invariant under $\bar{s}$. It follows that a classical code $C$ of length $n=2^m-1$ has $\GLmF\leq\Aut C$ if and only if it is invariant under the standard action of $\GLmF$. Hence, to classify codes of length $n$ with $\GLmF\leq\Aut C$ and $m\geq 5$ or $m=3$, it is sufficient to classify those invariant under $s$. We now describe the punctured and shortened Reed-Muller codes --- both of which have this property --- before explaining how the result of \cite{bardoe} implies that they are unique.

\subsection{Classical Reed-Muller codes}\label{Sec:RMCodes}
The classical Reed-Muller codes are a family of codes based on the evaluation of polynomials over $\mathbb{F}_2$. They are notable for their large automorphism groups, and efficient encoding and decoding algorithms~\cite{macwilliams_theory_1977}. We now briefly review the construction of these codes and explain how the punctured and shortened variants are built from them.

\subsubsection{Reed-Muller codes}
Let $\mathbb{F}_2^\Omega \coloneqq \{f\,|\, f:\Omega\to \mathbb{F}_2\}$ be the set of all functions from $\Omega$ to $\mathbb{F}_2$. A function $f:\Omega\to \mathbb{F}_2$ is completely defined by its evaluation vector $\Eval(f)$, which lists the values attained by $f$ on each element of $\Omega$; given some ordering $p_1, p_2, \ldots, p_{2^m}$ of the elements of $\Omega$, we write
\begin{equation}
    \Eval(f) \coloneqq \big(f(p_1), f(p_2), \ldots, f(p_{2^m})\big).
\end{equation}
For constants $a, b\in\mathbb{F}_2$ and functions $f, g\in \mathbb{F}_2^\Omega$, we have $\Eval(af+bg)=a\Eval(f)+b\Eval(g)$, so $\mathbb{F}_2^\Omega$ is a $2^m$-dimensional vector space over $\mathbb{F}_2$. A basis of this space is given by the set of squarefree monomials $x_1^{a_1}\cdots x_m^{a_m}$, where $x_i$ denotes the $i^{\text{th}}$ coordinate of $x=(x_1, \ldots, x_m)\in\Omega$, and $a_i\in \{0, 1\}$. (Note that $x^2_i = x_i$ for all $x_i\in \mathbb{F}_2$, so monomials containing higher-order factors play no role.) Hence, any function $f:\Omega\to \mathbb{F}_2$ can be uniquely expressed as a polynomial
\begin{equation}\label{Eq:Poly}
    f(x) = \sum_{a\in \{0, 1\}^m} g(a) x_1^{a_1}\cdots x_m^{a_m}
\end{equation}
for some coefficients $g(a)\in\mathbb{F}_2$. Equivalently, $\mathbb{F}_2^\Omega$ is naturally isomorphic to the Boolean coordinate ring in $m$ variables, denoted $B_m \coloneqq \mathbb{F}_2[x_1, \ldots, x_m]/\langle x_1^2-x_1, \ldots, x_m^2 - x_m\rangle$. As a vector space, $B_m$ decomposes as the direct sum $B_m = \bigoplus_{s=0}^m (B_m)_s$ of the $\dim (B_m)_s = \binom{m}{s}$-dimensional linear subspaces
\begin{equation}
    (B_m)_s = \Span_{\mathbb{F}_2}\{x_{i_1}\cdots x_{i_s}\,|\,1\leq i_1<\cdots < i_s\leq m\} \subseteq \mathbb{F}_2^\Omega
\end{equation}
spanned by the squarefree monomials of degree $s$.

\begin{definition}
    The classical Reed-Muller code of type $(r, m)$, where $0\leq r\leq m$, is the span of the evaluation vectors of all polynomials of degree at most $r$:
    \begin{equation}
        \RM(r, m) \coloneqq \big\{\Eval(f)\,\big|\, f:\Omega\to \mathbb{F}_2, \deg(f)\leq r\big\} = \bigoplus_{s=0}^r (B_m)_s  \subseteq \mathbb{F}_2^\Omega.
    \end{equation}
    The Reed-Muller code for $r=-1$ is defined to be the trivial subspace, $\RM(-1, m)\coloneqq \{\mathsf{0}\}$, where $\mathsf{0}\in\mathbb{F}_2^\Omega$ is the zero function.
\end{definition}

Since $\RM(r, m)$ is the direct sum of all $(B_m)_s$ with $s\leq r$, it has dimension $k(r, m) = \sum_{s=0}^r\binom{m}{s}$. A logical message $\mu\in \mathbb{F}_2^k$ is used to define the coefficients of a polynomial $f_\mu\in\mathbb{F}_2^\Omega$, which is canonically identified with its evaluation vector to yield an encoded message of length $2^m$. The distance $d(r, m)$ of this code is determined by the polynomials of highest degree, and it may be shown that $d(r, m) = 2^{m-r}$ \cite{macwilliams_theory_1977}. In summary, $\RM(r, m)$ has type $[2^m, \sum_{s=0}^r\binom{m}{s}, 2^{m-r}]$. 

From the definition it follows that, for fixed $m$, the Reed-Muller codes $\RM(r, m)$ form a family ordered by strict inclusion 
\begin{equation}
    \RM(-1, m) = \{\mathsf{0}\} \subset \RM(0, m) = C^{\text{rep}}_{2^m} \subset \RM(1, m)\subset \cdots \subset \RM(m-1, m) = C^{\text{even}}_{2^m} \subset \RM(m, m) = \mathbb{F}_2^\Omega,
\end{equation}
where $C^{\text{rep}}_{2^m}$ is the repetition code and $C^{\text{even}}_{2^m}$ the even-weight code. Another important property of Reed-Muller codes is that $\RM(r, m)^\perp = \RM(m-r-1, m)$, meaning that the duality operation acts to reorder the codes in the Reed-Muller family (note that the definition of $\RM(-1, m)$ ensures that this relation holds for all $-1\leq r\leq m$).

For any affine linear transformation $T:\mathbb{F}_2^m\to\mathbb{F}_2^m$ sending $x\mapsto Ax+b$ for some $A\in \mathsf{GL}_m(\mathbb{F}_2)$ and $b\in \mathbb{F}_2^m$, the Boolean function $(Tf)(x) = f(T^{-1}x)$ is a polynomial of the same degree as $f$. It follows that $T$ permutes the codewords of $\RM(r, m)$, that is, it is an automorphism of this code. In fact, it can be shown that, for $1\leq r\leq m-2$, this is the only class of transformations under which $\RM(r, m)$ is invariant, meaning $\Aut \RM(r, m) = \mathsf{AGL}_m(\mathbb{F}_2)$ is the affine general linear group over $\mathbb{F}_2$ \cite{macwilliams_theory_1977,berger-1993}. On the other hand, when $r=-1, 0, m-1$ or $m$, the codes have the maximal automorphism group $\Aut \RM(r, m)= S_{2^m}$.

\subsubsection{Punctured Reed-Muller codes}
An important variant of $\RM(r, m)$ can be obtained by `puncturing' each codeword by removing the coordinate corresponding to the zero vector in $\Omega$, or equivalently by replacing $\Omega$ with the projective space $P$ in the definition of $\RM(r, m)$. Like the elements of $\mathbb{F}_2^\Omega$, functions $f:P\to \mathbb{F}_2$ are completely characterised by their evaluation vectors, which are strings giving the value of $f$ on each of the $n=2^m-1$ elements of $P$. These functions can similarly be expanded in the basis of monomials $x_1^{a_1}\cdots x_m^{a_m}$. However, as $P$ contains only $n$ points whereas there are a total of $2^m$ such monomials, this basis is overcomplete, meaning that one term can be omitted. This can also be seen by noting that the function $\prod_{i=1}^m (1+x_i) = 0$ vanishes identically for all $x\in P$, giving a non-trivial linear relationship between the monomials. By convention, the top-degree monomial $x_1\cdots x_m$ is excluded from the basis set and a function $f:P\to \mathbb{F}_2$ is uniquely expressed as
\begin{equation}
    f(x) = \sum_{\substack{a\in \{0, 1\}^m \\ a\neq (1, \ldots, 1)}} g(a) x_1^{a_1}\cdots x_m^{a_m}.
\end{equation}
This decomposition allows us to identify $\mathbb{F}_2^P \cong B_m^* \coloneqq B_m/\langle \prod_{i=1}^m (1+x_i) \rangle = \bigoplus_{s=0}^{m-1} (B_m^*)_s$, where $(B_m^*)_s$ is the linear span of all degree-$s$ polynomials in $\mathbb{F}_2^P$.

\begin{definition}
    The punctured Reed-Muller code of type $(r, m)$, where $0\leq r\leq m-1$, is obtained from $\RM(r, m)$ by deleting the coordinate corresponding to evaluation at the point $0\in \Omega$, that is,
    \begin{equation}
    \RM^*(r, m) = \mathrm{Punct}_{0}\RM(r, m) = \big\{\Eval(f)\,\big|\, f:P\to \mathbb{F}_2, \deg(f)\leq r\big\} = \bigoplus_{s=0}^r (B_m^*)_s  \subseteq \mathbb{F}_2^P.
\end{equation}
\end{definition}

Punctured Reed-Muller codes have type $[2^m-1, \sum_{s=0}^r \binom{m}{s}, 2^{m-r}-1]$ and satisfy the inclusion relations
\begin{equation}\label{Eq:PuncturedChain}
    C^{\text{rep}}_n = \RM^*(0, m)\subset \RM^*(1, m)\subset \cdots \subset \RM^*(m-2, m) \subset \RM^*(m-1, m) = \mathbb{F}_2^P.
\end{equation}
Since $\RM^*(r, m)$ is obtained by puncturing $\RM(r, m)$ at the coordinate indexed by $0\in\mathbb{F}_2^m$, an affine permutation $T:x\mapsto Ax+b$ descends to a permutation of the punctured coordinates if and only if it fixes the origin, meaning $b=0$. Correspondingly, it may be shown that $\Aut \RM^*(r, m) = \mathsf{GL}_m(\mathbb{F}_2)$ for $1\leq r\leq m-2$, while $\Aut \RM^*(0, m) = \Aut \RM^*(m-1, m) = S_n$.

\subsubsection{Shortened Reed-Muller codes}
Each punctured Reed-Muller code contains the repetition code $\RM^*(0, m) = C^{\text{rep}}_n = \mathbb{F}_2\mathsf{1}$, where $\mathsf{1}\in\mathbb{F}_2^P$ is the constant function. The orthogonal complement of $C^{\text{rep}}_n$ in $\RM^*(r, m)$ is known as the shortened Reed-Muller code $\RM_*(r, m)$,
\begin{equation}
    \RM^*(r, m) = \mathbb{F}_2\mathsf{1}\oplus\RM_*(r, m).
\end{equation}
Shortened Reed-Muller codes can also be obtained directly from the code $\RM(r, m)$ by first restricting to those functions which vanish at zero, then puncturing on this coordinate:
\begin{equation}
    \RM_*(r, m) = \mathrm{Punct}_{0}\big\{\Eval(f)\,\big|\, f:\Omega\to \mathbb{F}_2, \deg(f)\leq r, f(0)=0\big\} = \bigoplus_{s=1}^r (B_m^*)_s 
\end{equation}
From either of these definitions it can be shown that the codes $\RM_*(r, m)$ are linear codes of type $[2^m-1, \sum_{s=1}^r\binom{m}{s}, 2^{m-r}]$, and that they fit into the chain of inclusions
\begin{equation}\label{Eq:ShortenedChain}
    \{\mathsf{0}\} = \RM_*(0, m)\subset \RM_*(1, m)\subset \cdots \subset \RM_*(m-2, m) \subset \RM_*(m-1, m) = C^{\text{even}}_n,
\end{equation}
where $C^{\text{even}}_n = [C^{\text{rep}}_n]^\perp$ is the even-weight code.

As a consequence of a general relationship between the punctured and shortened variants $C^*$ and $C_*$ of any classical code $C$, namely that $(C_*)^\perp = (C^\perp)^*$ and $(C^*)^\perp = (C^\perp)_*$ \cite{Huffman_Pless_2003}, the punctured and shortened Reed-Muller codes are related by taking duals:
\begin{subequations}\label{Eq:RMDuals}
\begin{align}
    \RM^*(r, m)^\perp ={}& \RM_*(m-1-r, m),\\
    \RM_*(r, m)^\perp ={}& \RM^*(m-1-r, m).
\end{align}
\end{subequations}
Since the inner product is invariant under permutations, the automorphism group of a code is unchanged by taking the dual, that is, $\Aut C^\perp = \Aut C$. It follows that $\Aut\RM_*(r, m) = \Aut\RM^*(m-r-1, m)^\perp = \mathsf{GL}_m(\mathbb{F}_2)$ for $1\leq r\leq m-2$, while $\Aut\RM_*(r, m) = S_n$ if $r=0$ or $m-1$.

\subsection{Uniqueness of Reed-Muller codes}\label{Sec:Uniqueness}
Let $\mathbb{F}$ be a field and $G$ be a group. A vector space $\mathbb{F}^n$ is an $\mathbb{F}G$-(permutation) module when it carries a linear permutation action of $G$, and an $\mathbb{F}G$-submodule of $\mathbb{F}^n$ is a linear subspace $C\leq\mathbb{F}^n$ that is invariant with respect to this action. In other words, $C$ is a linear code over $\mathbb{F}$ such that $G\leq \Aut C$. In particular, classifying the $\mathbb{F}G$ submodules of $\mathbb{F}^n$ is exactly equivalent to classifying the $G$-invariant codes of length $n$. Ref.~\cite{bardoe} classifies the $\mathbb{K}\GLmF$-submodules of $\mathbb{K}^n$ under the standard action of $\mathsf{GL}_m(\mathbb{F}_q)$, where $n=(q^m-1)/(q-1)$ for some prime power $q=p^t$ and $\mathbb{K}$ is the algebraic closure of $\mathbb{F}_q$. Specialising to the binary case where $q=2$, the main theorem of Ref.~\cite{bardoe} has the following statement.

\begin{proposition}[\cite{bardoe}, Theorem A]\label{Theorem:Bardoe}
    Let $\mathbb{K}\equiv \overline{\mathbb{F}}_2 = \bigcup_{t=1}^\infty \mathbb{F}_{2^t}$ be the algebraic closure of $\mathbb{F}_2$, let $G=\GLmF$, and consider the natural action $s$ of $G$ on $P \coloneqq \mathbb{F}_2^m\backslash\{0\}$. The space $\mathbb{K}^P = \{f\,|\,f:P\to \mathbb{K}\}$ is a $\mathbb{K}G$-module under this action, and decomposes as $\mathbb{K}^P = \mathbb{K}\mathsf{1}\oplus Y_P$, where $\mathbb{K}\mathsf{1}$ is the space of constant functions on $P$ and $Y_P$ is the kernel of the map $\mathbb{K}^P\to \mathbb{K}$, $f\mapsto n^{-1}\sum_{p\in P}f(p)$. The $\mathbb{K}G$-submodule lattice of $Y_P$ is isomorphic to the set $\mathcal{H} = \{0, 1, 2, \ldots, m-1\}$ with the natural order $0< 1 < \cdots < m-1$. In other words, each submodule $\mathcal{Y}_r$ of $Y_P$ is labelled by an element $r\in\mathcal{H}$, and the $\mathcal{Y}_r$ are ordered by strict inclusion as
    \begin{equation}
        0 = \mathcal{Y}_0 \subset \mathcal{Y}_1 \subset \cdots \subset \mathcal{Y}_{m-1} = Y_P.
    \end{equation}
    Moreover, the $\mathcal{Y}_r$ are given explicitly by $\mathcal{Y}_r = Y_P\cap \mathcal{F}_r$, where 
    \begin{equation}
         \mathcal{F}_r = \bigoplus_{s=0}^r\Span_{\mathbb{K}}\{x_{i_1}\cdots x_{i_s}\,|\,1\leq i_1<\cdots < i_s\leq m\} \subseteq \mathbb{K}^P
    \end{equation}
    and $\mathbb{K}^P$ has been identified with $\mathbb{K}[x_1, \ldots, x_m]/\langle \prod_i(1+x_i), x_1^2-x_1, \ldots, x_m^2-x_m\rangle$.
\end{proposition}

\cref{Lemma:Uniqueness} follows as a consequence of \cref{Theorem:Bardoe} via the following corollary. 

\begin{corollary}
    The only binary linear codes of length $n=2^m-1$ with automorphism group containing $\GLmF$ are the punctured Reed-Muller codes $\RM^*(r, m)$, and the shortened Reed-Muller codes $\RM_*(r, m)$, where $r=0, 1,\ldots, m-1$.
\end{corollary}

\begin{proof}
    Extension of scalars along $\mathbb{F}_2\hookrightarrow\mathbb{K}$ identifies $\mathbb{K}\otimes_{\mathbb{F}_2}\mathbb{F}_2^P$ with $\mathbb{K}^P$, and the map $v\mapsto 1\otimes_{\mathbb{F}_2}v$ injectively embeds $\mathbb{F}_2^P$ into $\mathbb{K}^P$. Moreover, for a submodule $C\leq\mathbb{F}_2^P$, the map $\phi:C\mapsto \mathbb{K}\otimes_{\mathbb{F}_2}C \subseteq \mathbb{K}^P$ is injective, since $C = \phi(C)\cap \mathbb{F}_2^P$. The submodules $\mathcal{Y}_r$ of $Y_P$ are canonically identified with the scalar extensions of the shortened Reed-Muller codes $\RM_*(r, m)$, that is
    \begin{equation}
        \mathcal{Y}_r \cong \mathbb{K}\otimes_{\mathbb{F}_2}\RM_*(r, m),
    \end{equation} 
    and similarly
    \begin{equation}
        \mathbb{K}\mathsf{1}\oplus\mathcal{Y}_r \cong \mathbb{K}\otimes_{\mathbb{F}_2} \RM^*(r, m).
    \end{equation}
    By \cref{Theorem:Bardoe}, the $\mathbb{K}\GLmF$-submodules of $\mathbb{K}^P$ are precisely the spaces $\mathcal{Y}_r$ and $\mathbb{K}\mathsf{1}\oplus\mathcal{Y}_r$ for $0\leq r\leq m-1$. Since the map $\phi$ is injective, this implies that the shortened and punctured Reed-Muller codes are the unique $\mathbb{F}_2\GLmF$-submodules of $\mathbb{F}_2^P$; this is equivalent to the statement of the corollary.
\end{proof}

The results of Ref.~\cite{bardoe} are significantly more general than the special case outlined above: it describes the full submodule lattice decomposition of $\overline{\mathbb{F}}_q^{\raisebox{-.3em}{\scalebox{.7}{$P$}}}$, where $P=(\mathbb{F}_q^k\backslash\{0\})/\mathbb{F}_q^*$, under the natural action of $\mathsf{GL}_k(\mathbb{F}_q)$ for any integer $k$ and any prime power $q=p^t$. In \cref{Sec:Qudit}, we discuss this result in the context of qudit phantom codes.

\section{Proof of Theorem\texorpdfstring{~\ref{Thm:Minimal}}{ 4}}\label{Sec:MinimalProof}
In this section, we prove the uniqueness of the punctured hypercube codes.

\begin{reptheorem}{Thm:Minimal}
    For $k\geq 5$ and $k=3$, the only binary CSS phantom codes of type $[\![2^k-1, k, d]\!]$ with $d>1$ are the punctured hypercube codes (up to CSS isomorphism).
\end{reptheorem}

To prove \cref{Thm:Minimal}, we split up the quantum CSS code $Q$ into its classical constituent codes $C_X$ and $C_Z$, and analyse each component individually. Because the permutation automorphism group of $Q$ contains $G$ as a subgroup, where $G$ is the $\GLkF$ extension defined in \cref{eq:G}, it follows that both $C_X$ and $C_Z$ must have (classical) automorphism groups containing $G$. The following lemma establishes that, when the bound $n\geq 2^k-1$ is saturated, no non-trivial group extension is possible, and the classical automorphism groups must both contain $\GLkF$ as a subgroup.

\begin{lemma}\label{Lemma:Nis1}
    For a phantom code encoding $k\geq 5$ or $k=3$ logical qubits into $n=2^k-1$ physical qubits, the group $N$ defined in \cref{eq:TrivialLogical} is trivial, meaning $G/N=G\cong\GLkF$.
\end{lemma}

\begin{proof}
    Let $\Omega$ be a set of size $n=2^k-1$, and let the $N$-orbits on $\Omega$ be $\Delta_1, \ldots, \Delta_t\subseteq\Omega$ with $\Omega = \coprod_i\Delta_i$. Because $N\triangleleft G$, the group $G$ permutes the set of $N$-orbits (for $\alpha\in \Omega$, we have $g\cdot(N\cdot\alpha) = N\cdot(g\cdot\alpha)$). Let $K$ be the kernel of this action, that is, 
    \begin{equation}
        K \coloneqq \{g\in G\,|\, g\cdot\Delta_i = \Delta_i\text{ for all }i=1,\ldots, t\}.
    \end{equation}
    Then $N\trianglelefteq K\trianglelefteq G$, so $G/K\cong (G/N)/(K/N)$ is a quotient of $G/N\cong\GLkF$. The simplicity of $\GLkF$ then implies that either (\textit{i}) $G/K = \GLkF$ or (\textit{ii}) $G/K = 1$. Case (\textit{i}) can be immediately ruled out as $G/K\cong\GLkF$ would imply the existence of a faithful permutation representation of degree $t<n$, which is contradicted by the minimality of $n$. Hence $G/K = 1$, meaning every $N$-orbit is actually $G$-invariant.

\footnotetext[5]{Only the group theoretical part of \cref{prop:Product} is relevant here; topological considerations are unimportant in this context.}

    We now wish to show that the action of $N$ is transitive. Suppose for contradiction that $t\geq 2$, and pick an $N$-orbit, say $\Delta\equiv\Delta_1$, with $|\Delta|<n$. Let $G_\Delta$ and $N_\Delta$ be the subgroups of $S_{|\Delta|}$ induced by the actions of $G$ and $N$ on $\Delta$. Then $N_\Delta\trianglelefteq G_\Delta$, and $G_\Delta/N_\Delta$ is again a quotient of $\GLkF$, hence either trivial or isomorphic to $\GLkF$. By \cref{Thm:muGN}, the latter would imply 
    \begin{equation}
        n = \mu(\GLkF) = \mu(G_\Delta/N_\Delta) \leq \mu(G_\Delta) \leq |\Delta| < n,
    \end{equation}
    a contradiction, so $G_\Delta/N_\Delta=1$. Since this holds for all $N$-orbits $\Delta_i$, we get a homomorphism $\rho:G\hookrightarrow G_{\Delta_1}\times\cdots\times G_{\Delta_t}$, $g\mapsto (g|_{\Delta_1}, \ldots, g|_{\Delta_t})$ that is injective since each $\Delta_i$ is invariant under $G$. Moreover, $\rho(N)\triangleleft \rho(G)$ and, by the injectivity of $\rho$, we have $\rho(G)/\rho(N)\cong G/N\cong \GLkF$, so $\GLkF$ is isomorphic to a quotient of the direct product $G_{\Delta_1}\times\cdots\times G_{\Delta_t}$. By induction using \cref{prop:Product}, it follows that $\GLkF$ is isomorphic to a quotient of some $G_{\Delta_i}\leq S_{|\Delta_i|}$ with $|\Delta_i|<n$, contradicting the minimality of $n$ \cite{Note5}. Hence $t=1$ and $N$ acts transitively on $\Omega$.

    Now let $\alpha\in\Omega$, and let $G_\alpha \coloneqq \Stab_G(\alpha)\leq S_{n-1}$ be the point stabiliser of $\alpha$. Since $N$ is transitive, there exists some $n\in N$ such that $g\cdot \alpha = n\cdot\alpha$. This implies $(n^{-1}g)\cdot \alpha = \alpha$, meaning $n^{-1}g\in G_\alpha$. Hence $g = n(n^{-1} g) \in N G_\alpha$ for every $g\in G$, which implies that $G=N G_\alpha$. Therefore, by \cref{prop:HK3} we have
    \begin{equation}
        \GLkF\cong G/N\cong N G_\alpha/N \cong G_\alpha/(G_\alpha\cap N).
    \end{equation}
    By \cref{Thm:muGN}, this implies that
    \begin{equation}
        n = \mu(\GLkF) = \mu[G_\alpha/(G_\alpha\cap N)]\leq\mu(G_\alpha) \leq n-1
    \end{equation}
    a contradiction. Thus $N=1$.
\end{proof}

\begin{proof}[Proof of \cref{Thm:Minimal}]
    A permutation automorphism of a CSS code preserves all $X$- and $Z$-type stabilisers, meaning that both constituent classical codes $C_X$ and $C_Z$ must have automorphism groups that contain $G$ as a subgroup. By \cref{Lemma:Nis1}, when $n=2^k-1$ we must have $G=\GLkF$, so both $C_X$ and $C_Z$ must contain $\GLkF$ as a subgroup. It turns out that for $k\geq 5$ and $k=3$, the only classical linear codes of length $2^k-1$ with automorphism group containing $\GLkF$ are the punctured and shortened Reed-Muller codes $\RM^*(r, k)$ and $\RM_*(r, k)$, where $r=0, 1,\ldots, k-1$. This fact follows from the results of Ref.~\cite{bardoe}, as explained in detail in \cref{Sec:GLMFInv}. Hence, both $C_X$ and $C_Z$ must be one of these codes. We now analyse each possible case.
    
    \textit{Case (i)} $C_X = \RM_*(r, k)$, $C_Z = \RM_*(s, k)$. The CSS condition $C_X^\perp\leq C_Z$ is
    \begin{equation}\label{Eq:CSS1}
        \RM^*(\bar{r}, k) = C^{\text{rep}}_n \oplus \RM_*(\bar{r}, k) \leq \RM_*(s, k)
    \end{equation}
    where $C^{\text{rep}}_n$ is the repetition code of length $n$, and $\bar{r} = k-1-r$. \cref{Eq:CSS1} cannot be satisfied as $C^{\text{rep}}_n$ is not contained in any shortened Reed-Muller code, so this pair does not give a valid CSS code.

    \textit{Case (ii)} $C_X = \RM^*(r, k)$, $C_Z = \RM^*(s, k)$. Without loss of generality we assume $r\geq s$, in which case the CSS condition $\RM^*(r, k)\geq\RM_*(\bar{s}, k)$ requires $r\geq \bar{s}$, where $\bar{s} = k-1-s$. Writing $r=\bar{s}+a$ for some $a\geq 0$, the dimension of a CSS code $Q$ formed from $C_X$ and $C_Z$ is $\dim Q = \dim C_X - \dim C_Z^\perp = 1+\sum_{p=\bar{s}+1}^{\bar{s}+a}\binom{k}{p}$. 
    Since for $k\geq 3$ the right-hand side is either 1 or larger than $k$, there is no solution over $r, s$ that yields $\dim Q = k$, hence no code of this form can have the desired maximal logical dimension.

    \textit{Case (iii)} $C_X = \RM_*(r, k)$, $C_Z = \RM^*(s, k)$. The CSS condition reads $\RM^*(\bar{r}, k)\geq \RM_*(\bar{s}, k)$, which holds whenever $r\geq \bar{s}$. Writing $r=\bar{s}+a$ for some $a\geq 0$, the dimension of the resultant CSS code, which we denote $\QRM^*(r, s)$, is then $\dim \QRM^*(r, s) = \sum_{p=\bar{s}+1}^{\bar{s}+a}\binom{k}{p}$. For $k\geq 3$, this is solved only by $(\bar{s}, a) = (0, 1)$ or $(k-2, 1)$, corresponding to $(r, s) = (1, k-1)$ or $(k-1, 1)$. The code $\QRM^*(1, k-1)$ has trivial $Z$-distance $d_Z=1$, since the fact that $C_X^\perp = \RM^*(k-2, k)$ has distance $3$ implies that $C_Z\backslash C_X^\perp = \mathbb{F}_2^n\backslash\RM^*(k-2, k)$ contains an element of weight $1$. Hence $\QRM^*(k-1, 1)$, which is the punctured hypercube code, is the unique binary CSS code of type $[\![2^k-1, k, d>1]\!]$.
\end{proof}

\section{Punctured hypercube codes}\label{Sec:Hypercube}

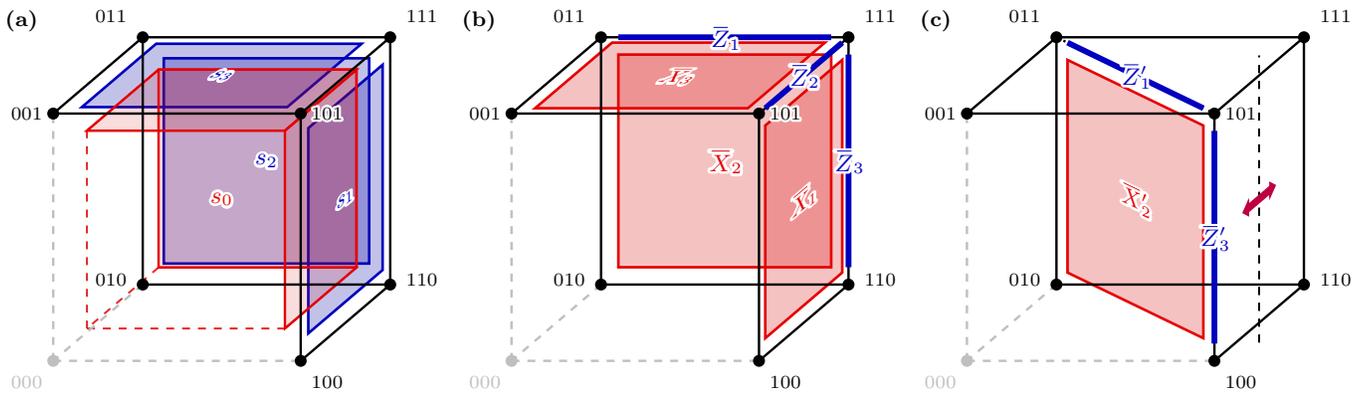
\begin{figure*}
\centering
\newcommand{\CubeScale}{0.7} 

\pgfmathsetlengthmacro{\CubeX}{\CubeScale*4.7cm}
\pgfmathsetlengthmacro{\CubeYx}{\CubeScale*1.7cm}
\pgfmathsetlengthmacro{\CubeYy}{\CubeScale*1.45cm}
\pgfmathsetlengthmacro{\CubeZ}{\CubeScale*4.7cm}
\sidesubfloat[]{\hspace{-1.5em}\begin{tikzpicture}[
    x={(\CubeX,0cm)},
    y={(\CubeYx,\CubeYy)},
    z={(0cm,\CubeZ)},
    vertex/.style={circle, fill=black, inner sep=1.6pt},
    cubeedge/.style={black, line width=0.9pt},
    zface/.style={draw=ZTextColor, fill=ZFaceColor, fill opacity=0.24, line width=1.0pt},
    xvolface/.style={draw=XTextColor, fill=XFaceColor, fill opacity=0.14, line width=0.9pt},
    xvolfaceLight/.style={draw=XTextColor, fill=XFaceColor!55!white, fill opacity=0.14, line width=0.9pt},
    xvoledge/.style={draw=XTextColor, line width=0.75pt, opacity=0.85},
    qlabel/.style={font=\scriptsize, fill=white, inner sep=0.8pt, outer sep=0pt},
]

\colorlet{specialgray}{black!25}

% White silhouette of math (good for glyph outlines)
\newcommand{\whitemathsilhouette}[1]{%
  \textpdfrender{
    TextRenderingMode=FillStroke,
    LineWidth=2.2pt,
    LineJoinStyle=Round,
    StrokeColor=white,
    FillColor=white
  }{$#1$}%
}

% Plain operator label with white outline
\newcommand{\plainoplabelcontent}[3]{%
  \begin{tikzpicture}[baseline=(full.base)]
    \node[
      inner sep=0pt,
      outer sep=0pt,
      font=\small\bfseries,
      text opacity=0,
      draw opacity=0
    ] (full) {$#2_{#3}$};

    \node[
      inner sep=0pt,
      outer sep=0pt,
      font=\small\bfseries
    ] at (full.center)
      {\whitemathsilhouette{#2_{#3}}};

    \node[
      inner sep=0pt,
      outer sep=0pt,
      font=\small\bfseries,
      text=#1
    ] at (full.center)
      {$#2_{#3}$};
  \end{tikzpicture}%
}

% --- Outer cube vertices ---
\coordinate (v000) at (0,0,0);
\coordinate (v100) at (1,0,0);
\coordinate (v010) at (0,1,0);
\coordinate (v110) at (1,1,0);
\coordinate (v001) at (0,0,1);
\coordinate (v101) at (1,0,1);
\coordinate (v011) at (0,1,1);
\coordinate (v111) at (1,1,1);

% --- Face centers ---
\coordinate (cx0) at ($(v000)!0.5!(v011)$); % x=0
\coordinate (cx1) at ($(v100)!0.5!(v111)$); % x=1
\coordinate (cy0) at ($(v000)!0.5!(v101)$); % y=0
\coordinate (cy1) at ($(v010)!0.5!(v111)$); % y=1
\coordinate (cz0) at ($(v000)!0.5!(v110)$); % z=0
\coordinate (cz1) at ($(v001)!0.5!(v111)$); % z=1
\coordinate (cc)  at ($(v000)!0.5!(v111)$); % cube center

% --- All six shrunken Z-face stabilizers ---
\def\faceshrink{0.17}

% x = 1
\filldraw[zface]
  ($(v100)!\faceshrink!(cx1)$) --
  ($(v101)!\faceshrink!(cx1)$) --
  ($(v111)!\faceshrink!(cx1)$) --
  ($(v110)!\faceshrink!(cx1)$) -- cycle;

% y = 1
\filldraw[zface]
  ($(v010)!\faceshrink!(cy1)$) --
  ($(v011)!\faceshrink!(cy1)$) --
  ($(v111)!\faceshrink!(cy1)$) --
  ($(v110)!\faceshrink!(cy1)$) -- cycle;

% z = 1
\filldraw[zface]
  ($(v001)!\faceshrink!(cz1)$) --
  ($(v101)!\faceshrink!(cz1)$) --
  ($(v111)!\faceshrink!(cz1)$) --
  ($(v011)!\faceshrink!(cz1)$) -- cycle;

% --- Inner shrunken cube for the global X stabilizer ---
\def\volshrink{0.2}

\coordinate (i000) at ($(v000)!\volshrink!(cc)$);
\coordinate (i100) at ($(v100)!\volshrink!(cc)$);
\coordinate (i010) at ($(v010)!\volshrink!(cc)$);
\coordinate (i110) at ($(v110)!\volshrink!(cc)$);
\coordinate (i001) at ($(v001)!\volshrink!(cc)$);
\coordinate (i101) at ($(v101)!\volshrink!(cc)$);
\coordinate (i011) at ($(v011)!\volshrink!(cc)$);
\coordinate (i111) at ($(v111)!\volshrink!(cc)$);

% Visible faces of the inner cube
% \fill[xvolfaceLight] (i000)--(i100)--(i101)--(i001)--cycle; % y=0
\filldraw[xvolface] (i010)--(i011)--(i111)--(i110)--cycle; % x=1
\filldraw[xvolface] (i100)--(i110)--(i111)--(i101)--cycle; % x=1
\filldraw[xvolface] (i001)--(i101)--(i111)--(i011)--cycle; % z=1

% Full inner-cube wireframe
\draw[xvoledge, dashed] (i000)--(i100);
\draw[xvoledge]         (i100)--(i110)--(i010);
\draw[xvoledge, dashed] (i010)--(i000);
\draw[xvoledge]         (i001)--(i101)--(i111)--(i011)--cycle;
\draw[xvoledge, dashed] (i000)--(i001);
\draw[xvoledge]         (i100)--(i101);
\draw[xvoledge]         (i010)--(i011);
\draw[xvoledge]         (i110)--(i111);

% --- Outer cube edges ---
\draw[specialgray, dashed, line width=0.9pt] (v000)--(v100);
\draw[cubeedge] (v100)--(v110)--(v010);
\draw[specialgray, dashed, line width=0.9pt] (v010)--(v000);
\draw[cubeedge] (v001)--(v101)--(v111)--(v011)--cycle;
\draw[specialgray, dashed, line width=0.9pt] (v000)--(v001);
\draw[cubeedge] (v100)--(v101);
\draw[cubeedge] (v010)--(v011);
\draw[cubeedge] (v110)--(v111);

% --- Vertices ---
\node[circle, fill=specialgray, inner sep=1.6pt] at (v000) {};
\foreach \v in {v100,v010,v110,v001,v101,v011,v111}
  \node[vertex] at (\v) {};

% --- Qubit labels ---
\node[qlabel, text=specialgray, shift={(-10pt,-8pt)}] at (v000) {000};
\node[qlabel, shift={( 10pt,-8pt)}] at (v100) {100};
\node[font=\scriptsize, inner sep=0pt, outer sep=0pt, shift={(-12pt,2pt)}] at (v010)
  {\contour{white}{010}};
\node[qlabel, shift={( 12pt, 2pt)}] at (v110) {110};
\node[qlabel, shift={(-10pt, 0pt)}] at (v001) {001};
\node[font=\scriptsize, inner sep=0pt, outer sep=0pt, shift={(10pt,0pt)}] at (v101)
  {\contour{white}{101}};
\node[qlabel, shift={(-12pt, 8pt)}] at (v011) {011};
\node[qlabel, shift={( 12pt, 8pt)}] at (v111) {111};

% --- New stabiliser labels ---
\begin{scope}[shift={(cx1)}, cm={0.760,0.649,0,1,(0,0)}]
  \node[inner sep=0pt, outer sep=0pt, transform shape] at (0,0)
    {\plainoplabelcontent{ZTextColor}{s}{1}};
\end{scope}

\begin{scope}[shift={(cy1)}, cm={1,0,0,1,(0,0)}]
  \node[inner sep=0pt, outer sep=0pt, transform shape] at (0,0)
    {\plainoplabelcontent{ZTextColor}{s}{2}};
\end{scope}

\begin{scope}[shift={(cz1)}, cm={1,0,0.760,0.649,(0,0)}]
  \node[inner sep=0pt, outer sep=0pt, transform shape] at (0,0)
    {\plainoplabelcontent{ZTextColor}{s}{3}};
\end{scope}

\node[inner sep=0pt, outer sep=0pt] at (cc)
  {\plainoplabelcontent{XTextColor}{s}{0}};

\end{tikzpicture}}
\hfill
\sidesubfloat[]{\hspace{-1.5em}\begin{tikzpicture}[
    x={(\CubeX,0cm)},
    y={(\CubeYx,\CubeYy)},
    z={(0cm,\CubeZ)},
    vertex/.style={circle, fill=black, inner sep=1.6pt},
    cubeedge/.style={black, line width=0.9pt},
    xplane/.style={draw=XTextColor, fill=XFaceColor, fill opacity=0.24, line width=1.0pt},
    zline/.style={draw=ZTextColor, line width=2.2pt},
    qlabel/.style={font=\scriptsize, fill=white, inner sep=0.8pt, outer sep=0pt},
]

\colorlet{specialgray}{black!25}

% White silhouette of math (good for glyph outlines)
\newcommand{\whitemathsilhouette}[1]{%
  \textpdfrender{
    TextRenderingMode=FillStroke,
    LineWidth=2.2pt,
    LineJoinStyle=Round,
    StrokeColor=white,
    FillColor=white
  }{$#1$}%
}

% Colored math on top
\newcommand{\fillmath}[2]{%
  \textcolor{#1}{$#2$}%
}

% Full operator label:
% 1) measure full label and overbarred symbol
% 2) draw a white bar underneath, aligned to the TeX overbar
% 3) draw white silhouette of the whole expression
% 4) draw colored expression on top
\newcommand{\oplabelcontent}[3]{%
  \begin{tikzpicture}[baseline=(full.base)]
    % Invisible measuring node for full expression
    \node[
      inner sep=0pt,
      outer sep=0pt,
      font=\small\bfseries,
      text opacity=0,
      draw opacity=0
    ] (full) {$\overbar{#2}_{#3}$};

    % Invisible measuring node for the overbarred symbol only
    \node[
      inner sep=0pt,
      outer sep=0pt,
      font=\small\bfseries,
      text opacity=0,
      draw opacity=0,
      anchor=base west
    ] (sym) at (full.base west) {$\overbar{#2}$};

    % White halo for the overbar itself, behind everything else
    \draw[line cap=round, white, line width=2pt]
      ($(sym.north west)+(0pt,-0.55pt)$) -- ($(sym.north east)+(0pt,-0.55pt)$);

    % White silhouette for the whole expression
    \node[
      inner sep=0pt,
      outer sep=0pt,
      font=\small\bfseries
    ] at (full.center)
      {\whitemathsilhouette{\overbar{#2}_{#3}}};

    % Colored expression on top
    \node[
      inner sep=0pt,
      outer sep=0pt,
      font=\small\bfseries,
      text=#1
    ] at (full.center)
      {$\overbar{#2}_{#3}$};
  \end{tikzpicture}%
}

% --- Cube vertices ---
\coordinate (v000) at (0,0,0);
\coordinate (v100) at (1,0,0);
\coordinate (v010) at (0,1,0);
\coordinate (v110) at (1,1,0);
\coordinate (v001) at (0,0,1);
\coordinate (v101) at (1,0,1);
\coordinate (v011) at (0,1,1);
\coordinate (v111) at (1,1,1);

% --- Face centers ---
\coordinate (cx) at ($(v100)!0.5!(v111)$);
\coordinate (cy) at ($(v010)!0.5!(v111)$);
\coordinate (cz) at ($(v001)!0.5!(v111)$);

% --- Edge centers ---
\coordinate (lz1) at ($(v011)!0.5!(v111)$);
\coordinate (lz2) at ($(v101)!0.5!(v111)$);
\coordinate (lz3) at ($(v110)!0.5!(v111)$);

% --- Shrunken X-logical faces ---
\def\faceshrink{0.14}

\filldraw[xplane]
  ($(v100)!\faceshrink!(cx)$) --
  ($(v101)!\faceshrink!(cx)$) --
  ($(v111)!\faceshrink!(cx)$) --
  ($(v110)!\faceshrink!(cx)$) -- cycle;

\filldraw[xplane]
  ($(v010)!\faceshrink!(cy)$) --
  ($(v011)!\faceshrink!(cy)$) --
  ($(v111)!\faceshrink!(cy)$) --
  ($(v110)!\faceshrink!(cy)$) -- cycle;

\filldraw[xplane]
  ($(v001)!\faceshrink!(cz)$) --
  ($(v101)!\faceshrink!(cz)$) --
  ($(v111)!\faceshrink!(cz)$) --
  ($(v011)!\faceshrink!(cz)$) -- cycle;

% --- Cube edges ---
\draw[specialgray, dashed, line width=0.9pt] (v000)--(v100);
\draw[cubeedge] (v100)--(v110)--(v010);
\draw[specialgray, dashed, line width=0.9pt] (v010)--(v000);
\draw[cubeedge] (v001)--(v101)--(v111)--(v011)--cycle;
\draw[specialgray, dashed, line width=0.9pt] (v000)--(v001);
\draw[cubeedge] (v100)--(v101);
\draw[cubeedge] (v010)--(v011);
\draw[cubeedge] (v110)--(v111);

% --- Slightly shortened Z-logical edges ---
\def\zedgeshrink{0.07}
\draw[zline] ($(v011)!\zedgeshrink!(v111)$) -- ($(v111)!\zedgeshrink!(v011)$);
\draw[zline] ($(v101)!\zedgeshrink!(v111)$) -- ($(v111)!\zedgeshrink!(v101)$);
\draw[zline] ($(v110)!\zedgeshrink!(v111)$) -- ($(v111)!\zedgeshrink!(v110)$);

% --- Vertices ---
\node[circle, fill=specialgray, inner sep=1.6pt] at (v000) {};
\foreach \v in {v100,v010,v110,v001,v101,v011,v111}
  \node[vertex] at (\v) {};

% --- Qubit labels ---
\node[qlabel, text=specialgray, shift={(-10pt,-8pt)}] at (v000) {000};
\node[qlabel, shift={( 10pt,-8pt)}] at (v100) {100};
\node[qlabel, shift={(-12pt, 2pt)}] at (v010) {010};
\node[qlabel, shift={( 12pt, 2pt)}] at (v110) {110};
\node[qlabel, shift={(-10pt, 0pt)}] at (v001) {001};
\node[font=\scriptsize, inner sep=0pt, outer sep=0pt, shift={(10pt,0pt)}] at (v101)
  {\contour{white}{101}};
\node[qlabel, shift={(-12pt, 8pt)}] at (v011) {011};
\node[qlabel, shift={( 12pt, 8pt)}] at (v111) {111};

% --- X-logical labels, face-aligned/sheared ---
\begin{scope}[shift={(cx)}, cm={0.760,0.649,0,1,(0,0)}]
  \node[inner sep=0pt, outer sep=0pt, transform shape] at (0,0)
    {\oplabelcontent{XTextColor}{X}{1}};
\end{scope}

\begin{scope}[shift={(cy)}, cm={1,0,0,1,(0,0)}]
  \node[inner sep=0pt, outer sep=0pt, transform shape] at (0,0)
    {\oplabelcontent{XTextColor}{X}{2}};
\end{scope}

\begin{scope}[shift={(cz)}, cm={1,0,0.760,0.649,(0,0)}]
  \node[inner sep=0pt, outer sep=0pt, transform shape] at (0,0)
    {\oplabelcontent{XTextColor}{X}{3}};
\end{scope}

% --- Z-logical labels ---
\node[inner sep=0pt, outer sep=0pt] at (lz1)
  {\oplabelcontent{ZTextColor}{Z}{1}};
\node[inner sep=0pt, outer sep=0pt] at (lz2)
  {\oplabelcontent{ZTextColor}{Z}{2}};
\node[inner sep=0pt, outer sep=0pt] at (lz3)
  {\oplabelcontent{ZTextColor}{Z}{3}};

\end{tikzpicture}}
\hfill
\sidesubfloat[]{\hspace{-1.5em}\begin{tikzpicture}[
    x={(\CubeX,0cm)},
    y={(\CubeYx,\CubeYy)},
    z={(0cm,\CubeZ)},
    vertex/.style={circle, fill=black, inner sep=1.6pt},
    cubeedge/.style={black, line width=0.9pt},
    xplane/.style={draw=XTextColor, fill=XFaceColor, fill opacity=0.24, line width=1.0pt},
    zline/.style={draw=ZTextColor, line width=2.2pt},
    qlabel/.style={font=\scriptsize, fill=white, inner sep=0.8pt, outer sep=0pt},
]

\colorlet{specialgray}{black!25}

% White silhouette of math (good for glyph outlines)
\newcommand{\whitemathsilhouette}[1]{%
  \textpdfrender{
    TextRenderingMode=FillStroke,
    LineWidth=2.2pt,
    LineJoinStyle=Round,
    StrokeColor=white,
    FillColor=white
  }{$#1$}%
}

% Colored math on top
\newcommand{\fillmath}[2]{%
  \textcolor{#1}{$#2$}%
}

% Unprimed operator label
\newcommand{\oplabelcontent}[3]{%
  \begin{tikzpicture}[baseline=(full.base)]
    \node[
      inner sep=0pt,
      outer sep=0pt,
      font=\small\bfseries,
      text opacity=0,
      draw opacity=0
    ] (full) {$\overbar{#2}_{#3}$};

    \node[
      inner sep=0pt,
      outer sep=0pt,
      font=\small\bfseries,
      text opacity=0,
      draw opacity=0,
      anchor=base west
    ] (sym) at (full.base west) {$\overbar{#2}$};

    \draw[line cap=round, white, line width=2pt]
      ($(sym.north west)+(0pt,-0.55pt)$) -- ($(sym.north east)+(0pt,-0.55pt)$);

    \node[
      inner sep=0pt,
      outer sep=0pt,
      font=\small\bfseries
    ] at (full.center)
      {\whitemathsilhouette{\overbar{#2}_{#3}}};

    \node[
      inner sep=0pt,
      outer sep=0pt,
      font=\small\bfseries,
      text=#1
    ] at (full.center)
      {$\overbar{#2}_{#3}$};
  \end{tikzpicture}%
}

% Primed operator label: renders as \bar{X}_2' etc.
\newcommand{\oplabelcontentprime}[3]{%
  \begin{tikzpicture}[baseline=(full.base)]
    \node[
      inner sep=0pt,
      outer sep=0pt,
      font=\small\bfseries,
      text opacity=0,
      draw opacity=0
    ] (full) {$\overbar{#2}_{#3}'$};

    \node[
      inner sep=0pt,
      outer sep=0pt,
      font=\small\bfseries,
      text opacity=0,
      draw opacity=0,
      anchor=base west
    ] (sym) at (full.base west) {$\overbar{#2}$};

    \draw[line cap=round, white, line width=2pt]
      ($(sym.north west)+(0pt,-0.55pt)$) -- ($(sym.north east)+(0pt,-0.55pt)$);

    \node[
      inner sep=0pt,
      outer sep=0pt,
      font=\small\bfseries
    ] at (full.center)
      {\whitemathsilhouette{\overbar{#2}_{#3}'}};

    \node[
      inner sep=0pt,
      outer sep=0pt,
      font=\small\bfseries,
      text=#1
    ] at (full.center)
      {$\overbar{#2}_{#3}'$};
  \end{tikzpicture}%
}

% --- Cube vertices ---
\coordinate (v000) at (0,0,0);
\coordinate (v100) at (1,0,0);
\coordinate (v010) at (0,1,0);
\coordinate (v110) at (1,1,0);
\coordinate (v001) at (0,0,1);
\coordinate (v101) at (1,0,1);
\coordinate (v011) at (0,1,1);
\coordinate (v111) at (1,1,1);

% --- Face center for diagonal X_2' face ---
\coordinate (cdiag) at ($(v011)!0.5!(v100)$);

% --- Edge centers ---
\coordinate (lz1) at ($(v011)!0.5!(v101)$);
\coordinate (lz3) at ($(v100)!0.5!(v101)$);

% --- Shrunken X-logical face ---
\def\faceshrink{0.14}

\filldraw[xplane]
  ($(v011)!\faceshrink!(cdiag)$) --
  ($(v101)!\faceshrink!(cdiag)$) --
  ($(v100)!\faceshrink!(cdiag)$) --
  ($(v010)!\faceshrink!(cdiag)$) -- cycle;

% --- Midpoints / center in the 101-111-110-100 plane ---
\coordinate (mrTop)    at ($(v101)!0.5!(v111)$);
\coordinate (mrBottom) at ($(v100)!0.5!(v110)$);
\coordinate (mrCenter) at ($(mrBottom)!0.5!(mrTop)$);

% --- Cube edges ---
\draw[specialgray, dashed, line width=0.9pt] (v000)--(v100);
\draw[cubeedge] (v100)--(v110)--(v010);
\draw[specialgray, dashed, line width=0.9pt] (v010)--(v000);
\draw[cubeedge] (v001)--(v101)--(v111)--(v011)--cycle;
\draw[specialgray, dashed, line width=0.9pt] (v000)--(v001);
\draw[cubeedge] (v100)--(v101);
\draw[cubeedge] (v010)--(v011);
\draw[cubeedge] (v110)--(v111);

% --- Slightly shortened Z-logical lines ---
\def\zedgeshrink{0.07}
\draw[black, dotted, line width=0.9pt] (v011) -- (v101);
\draw[zline]
  ($(v011)!\zedgeshrink!(v101)$) -- ($(v101)!\zedgeshrink!(v011)$);
\draw[zline]
  ($(v100)!\zedgeshrink!(v101)$) -- ($(v101)!\zedgeshrink!(v100)$);

% --- Vertices ---
\node[circle, fill=specialgray, inner sep=1.6pt] at (v000) {};
\foreach \v in {v100,v010,v110,v001,v101,v011,v111}
  \node[vertex] at (\v) {};

% --- Qubit labels ---
\node[qlabel, text=specialgray, shift={(-10pt,-8pt)}] at (v000) {000};
\node[qlabel, shift={( 10pt,-8pt)}] at (v100) {100};
\node[qlabel, shift={(-12pt, 2pt)}] at (v010) {010};
\node[qlabel, shift={( 12pt, 2pt)}] at (v110) {110};
\node[qlabel, shift={(-10pt, 0pt)}] at (v001) {001};
\node[font=\scriptsize, inner sep=0pt, outer sep=0pt, shift={(10pt,0pt)}] at (v101)
  {\contour{white}{101}};
\node[qlabel, shift={(-12pt, 8pt)}] at (v011) {011};
\node[qlabel, shift={( 12pt, 8pt)}] at (v111) {111};

% --- X-logical label, aligned with the 011--101 diagonal ---
\begin{scope}[shift={(cdiag)}, cm={0.900,-0.435,0,1,(0,0)}]
  \node[inner sep=0pt, outer sep=0pt, transform shape] at (0,0)
    {\oplabelcontentprime{XTextColor}{X}{2}};
\end{scope}

% --- Z-logical labels ---
\node[inner sep=0pt, outer sep=0pt] at (lz1)
  {\oplabelcontentprime{ZTextColor}{Z}{1}};
\node[inner sep=0pt, outer sep=0pt] at (lz3)
  {\oplabelcontentprime{ZTextColor}{Z}{3}};

% --- Dashed connector and small in-plane double arrow ---
\draw[black, dashed, line width=0.7pt]
  ($(mrBottom)!-0.08!(mrTop)$) -- ($(mrTop)!-0.08!(mrBottom)$);

% --- Dashed connector and small filled in-plane double arrow ---
\draw[black, dashed, line width=0.7pt]
  ($(mrBottom)!-0.08!(mrTop)$) -- ($(mrTop)!-0.08!(mrBottom)$);

% Filled double-headed arrow, entirely in the 101-111-110-100 plane
\def\arrhalflen{0.18}
\def\arrheadlen{0.060}
\def\arrhalfthick{0.015}
\def\arrheadhalfwidth{0.03}

\filldraw[purple, line width=0.5pt]
  % lower tip and lower head
  ($(mrCenter)+(0,-\arrhalflen,0)$) --
  ($(mrCenter)+(0,-\arrhalflen+\arrheadlen,\arrheadhalfwidth)$) --
  ($(mrCenter)+(0,-\arrhalflen+\arrheadlen,\arrhalfthick)$) --
  % shaft upper side
  ($(mrCenter)+(0,\arrhalflen-\arrheadlen,\arrhalfthick)$) --
  % upper head
  ($(mrCenter)+(0,\arrhalflen-\arrheadlen,\arrheadhalfwidth)$) --
  ($(mrCenter)+(0,\arrhalflen,0)$) --
  ($(mrCenter)+(0,\arrhalflen-\arrheadlen,-\arrheadhalfwidth)$) --
  ($(mrCenter)+(0,\arrhalflen-\arrheadlen,-\arrhalfthick)$) --
  % shaft lower side
  ($(mrCenter)+(0,-\arrhalflen+\arrheadlen,-\arrhalfthick)$) --
  % lower head
  ($(mrCenter)+(0,-\arrhalflen+\arrheadlen,-\arrheadhalfwidth)$) --
  cycle;

\end{tikzpicture}}
\caption{\textbf{Punctured hypercube codes.} \textbf{(a)} Stabilisers and \textbf{(b)} logical operators of the $[\![7, 3, 2]\!]$ punctured (hyper)cube code. \textbf{(c)} The logical $\CNOT_{21}$ gate can be performed by reflecting the square face $100$-$101$-$111$-$110$ about a line of symmetry: this maps $\overbar{Z}_1\mapsto\overbar{Z}_1\overbar{Z}_2$ and $\overbar{X}_2\mapsto \overbar{X}_1\overbar{X}_2$, while preserving all other logical operators up to a stabiliser.}
    \label{fig:Hypercube}
\end{figure*}

The punctured hypercube code of dimension $D$ is a $[\![2^D-1, D, 2]\!]$ CSS code defined by $C_X = \RM_*(D-1, D)$ and $C_Z = \RM^*(1, D)$, where $\RM^*(r, m)$ and $\RM_*(r, m)$ are respectively the punctured and shortened classical Reed-Muller codes described in \cref{Sec:RMCodes}. \cref{Thm:Minimal} shows that these codes are the unique family of phantom CSS codes with non-trivial distance that saturate the bound $n\geq 2^k-1$. We now explain how the polynomial evaluation scheme used to define the constituent classical Reed-Muller codes provides a convenient geometric description of the punctured hypercube codes, and we illustrate this with the example of the $[\![7, 3, 2]\!]$ code.

\subsection{Geometric interpretation of Reed-Muller codes}\label{Sec:RMGeom}
Recall that the Reed-Muller code $\RM(r, m)$ is the set of evaluation vectors of binary polynomials in $m$ variables of degree at most $r$. For a polynomial $f:\mathbb{F}_2^m\to\mathbb{F}_2$, the evaluation vector $\Eval(f)$ lists the values of $f(x)$ for all $x\in\Omega=\mathbb{F}_2^m$. By identifying $\Omega$ as an $m$-dimensional hypercube, we can associate each such polynomial with a geometric region of $\Omega$, namely those points $x\in \Omega$ where $f(x)=1$. For example, the monomial $f(x) = x_i$ corresponds to the half-cube where $x_i=1$, while $f(x) = 1+x_i$ indicates the half-cube $x_i=0$. More generally, for a set $S\subseteq[m]$ of size $|S| = r$ and a binary vector $b\in\mathbb{F}_2^S$, define
\begin{equation}
    \chi_{S, b}(x) \coloneqq \prod_{i\in S} (1+b_i+x_i),
\end{equation}
as the indicator for the region $\left\{x\in\Omega\;\middle|\; x_i = b_i \text{ for all } i\in S\right\}$ defined by $S$ and $b$. Geometrically, this is a coordinate-aligned affine subspace of codimension $r$.

Every polynomial of degree $\leq r$ can be expressed as a linear combination of the indicator functions $\{\chi_{S, b}\,|\,S\subseteq[m], |S| = r, b\in\mathbb{F}_2^S\}$. To see this, note that all degree-$r$ monomials are indicator functions with all $b_i=1$. Moreover, any indicator function for a set $T\subset[m]$ of size $|T| = r-1$ can be expressed as a linear combination of codimension-$r$ indicators by inserting factors of $(1+x_i) +x_i = 1$, for example $x_2(1+x_3) = (1+x_1)x_2(1+x_3) + x_1 x_2(1+x_3)$. Hence, by induction the codimension-$r$ indicator functions form an overcomplete basis of the Reed-Muller code $\RM(r, m)$. This provides a geometric picture of $\RM(r, m)$: it is generated by the coordinate-aligned affine subspaces of codimension $r$ in $\Omega$.

This picture naturally generalises to the punctured and shortened Reed-Muller codes. The space $P=\Omega\backslash\{0\}$ used in the definition of these codes is constructed by removing the zero vertex from the $m$-dimensional hypercube. Codewords of $\RM^*(r, m)$ are obtained simply by removing the entry corresponding to this vertex from all codewords of $\RM(r, m)$, which implies $\RM^*(r, m)$ is generated by the punctured evaluation vectors corresponding to the codimension-$r$ coordinate subspaces of $\Omega$. If such a subspace contains $0$, its punctured image is that subspace with the origin removed. On the other hand, a polynomial $f(x)\in\RM(r, m)$ is a codeword of the shortened Reed-Muller code $\RM_*(r, m)$ only if $f(0)=0$. In particular, the generators $\chi_{S, b}$ satisfy $\chi_{S, b}(0) = 0$ if and only if $b\neq 0$. Geometrically, this means that an affine subspace of $\RM(r, m)$ is, after puncturing, an admissible generator of $\RM_*(r, m)$ if and only if it does not contain the origin. In other words, $\RM_*(r, m)$ is generated by the set of codimension-$r$ affine subspaces that do not contain the origin.

\subsection{Stabilisers and logical operators}
We now employ the geometric description of classical Reed-Muller codes to determine the stabilisers and logical operators of the $[\![2^D-1, D, 2]\!]$ quantum punctured hypercube code family. We set $n=2^D-1$, $C_X = \RM_*(D-1, D)$, and $C_Z = \RM^*(1, D)$ throughout. 

We begin by computing the $X$- and $Z$-type stabiliser spaces, which are given by the duals of $C_X$ and $C_Z$. It follows from \cref{Eq:RMDuals} that the $X$ stabiliser space is $S_X = C_X^\perp = \RM^*(0, D)$, while the $Z$ stabiliser space is $S_Z = C_Z^\perp = \RM_*(D-2, D)$. It is straightforward to identify $S_X = C_n^{\text{rep}}$ as the repetition code of length $n$, meaning that there is a single $X$-type stabiliser equal to the product of $X$ over all sites, $s_0 = \prod_{v\in P} X_v$. By \cref{Sec:RMGeom}, any codeword of $S_Z = \RM_*(D-2, D)$ can be expressed as a binary linear combination of indicator functions of two-dimensional affine subspaces of $P$. In other words, the $Z$-stabilisers are generated by all square faces that do not contain the zero vertex.

We now turn to the logical operators. The space of logical $X$-type Pauli operators is given by the quotient
\begin{equation}
    \mathcal{L}_X = C_Z/C_X^\perp = \RM^*(1, D)/\RM^*(0, D) = \langle 1, x_1, \ldots, x_D\rangle / \langle 1\rangle \cong \langle x_1, \ldots, x_D\rangle.
\end{equation}
Hence, a representative set of independent $X$-type logical operators is given by the set of coordinate half-cubes defined by $x_i=1$ for $i=1, \ldots, D$. The space of $Z$-type logical operators is 
\begin{equation}
    \mathcal{L}_Z = C_X/C_Z^\perp = \RM_*(D-1, D)/\RM_*(D-2, D) \cong \left\langle x_2\cdots x_D,\ldots, x_1\cdots\hat{x}_i\cdots x_D, \ldots, x_1\cdots x_{D-1} \right\rangle, 
\end{equation}
where the circumflex on $\hat{x}_i$ indicates that this factor is omitted. The monomial $x_2\cdots x_D$ is non-zero only when $x_2=\cdots=x_D=1$, and therefore corresponds to an edge of length 1 parallel to the $x_1$-axis and connected to the all-ones vertex. Similar considerations apply to the other monomials in $\mathcal{L}_Z$.

\subsection{The \texorpdfstring{$[\![7, 3, 2]\!]$}{[[7, 3, 2]]} punctured hypercube code}
As a concrete example, we now describe the punctured hypercube code for $D=3$. This is a $[\![7, 3, 2]\!]$ code that may be obtained by removing a single qubit from the $[\![8, 3, 2]\!]$ colour code~\cite{Campbell_smallest_2016,PRXQuantum.6.020338,Wang_2024}.

The eight points of $\Omega = \mathbb{F}_2^3$ are the length-3 binary strings forming the vertices of a cube. By deleting the vertex $000$, which lies at the origin, we obtain the punctured cube $P=\Omega\backslash\{0\}$. A qubit is placed at each vertex of $P$. The code has a unique $X$-type stabiliser $s_0 = \prod_{v\in P} X_v$, along with three $Z$-type stabilisers corresponding to the square faces of $P$ not containing the $000$ vertex:
\begin{equation}
         s_1 = Z_{100}Z_{101}Z_{110}Z_{111}, \qquad s_2 = Z_{010}Z_{011}Z_{111}Z_{110}, \qquad \text{and} \qquad s_3 = Z_{001}Z_{011}Z_{111}Z_{101}.
\end{equation}
These arise from the indicator functions $x_1$, $x_2$, and  $x_3$ respectively: any degree-1 polynomial $f(x)$ with $f(0)=0$ can be expressed as a linear combination of these monomials. The $X$-type logical operators may be taken to be the same faces, only with the $Z$'s replaced with $X$'s, i.e.
\begin{subequations}
\begin{equation}
    \overbar{X}_1 = X_{100}X_{101}X_{110}X_{111}, \qquad \overbar{X}_2 = X_{010}X_{011}X_{111}X_{110}, \qquad \text{and} \qquad \overbar{X}_3 = X_{001}X_{011}X_{111}X_{101}.
\end{equation}
(Note that the $Z$ stabilisers are always squares, while the logical $X$ operators are generally codimension-1 subcubes; these shapes are the same only in dimension $D=3$.) The $Z$-type logical operators consist of three edges emanating from the $111$-vertex,
\begin{equation}
    \overbar{Z}_1 = Z_{011}Z_{111}, \qquad \overbar{Z}_2 = Z_{101}Z_{111}, \qquad \overbar{Z}_3 = Z_{110}Z_{111}.
\end{equation}
\end{subequations}

The logical $\overline{\CNOT}_{21}$ gate can be implemented by simultaneously exchanging qubits $101$ and $111$, and qubits $100$ and $110$. Geometrically, this corresponds to reflecting the square face with vertices $100$, $101$, $111$, and $110$ about its symmetry axis parallel to the line $111$-$110$. To see this, observe that this transformation changes the logical operators as
\begin{equation}
\begin{aligned}
    &\overbar{X}_1 \mapsto \overbar{X}_1, \qquad &&\overbar{X}_2 \mapsto X_{010}X_{011}X_{100}X_{101} = \overbar{X}_1\overbar{X}_2, \qquad  &&\overbar{X}_3 \mapsto \overbar{X}_3, \\
    &\overbar{Z}_1 \mapsto Z_{011}Z_{101} = \overbar{Z}_1\overbar{Z}_2, \qquad  &&\overbar{Z}_2 \mapsto\overbar{Z}_2, \qquad  &&\overbar{Z}_3 \mapsto Z_{100}Z_{101} = s_1\overbar{Z}_3,
\end{aligned}
\end{equation}
which, after accounting for equivalence modulo stabilisers, is exactly the action of the $\overline{\CNOT}_{21}$ gate. By symmetry, all other logical $\CNOT$ gates may be similarly implemented by reflecting square faces across the corresponding symmetry axes.

\section{An \texorpdfstring{$(\!(8, 2^4, 2)\!)$}{((8, 16, 2))} phantom code}\label{Sec:PG}
\footnotetext[50]{A QEC code $Q$ is of type $(\!(n, K, d)\!)$ when it encodes a logical Hilbert space $\mathcal{A}\cong \mathbb{C}^K$ into $n$ physical qubits, such that the Knill-Laflamme conditions hold for all errors of weight at most $d-1$. If $K=2^k$, then $Q$ encodes $k$ logical qubits. The round brackets indicate that $Q$ need not be a Pauli stabiliser code.}
The bound given in \cref{Thm:STPbound} holds for all $k\geq2$ with the exception of $k=4$. This isolated discrepancy stems from the exceptional isomorphism $\GLx\cong A_8$, which implies that $\mu(\GLx) = 8$ in contrast to the usual value $\mu(\GLkF) = 2^k-1$. In this section we discuss this isomorphism and its relation to the finite projective geometry $\PG = \mathbb{F}_2^4\backslash\{0\}$, before using it to construct a nontrivial phantom code of type $(\!(8, 2^4, 2)\!)$~\cite{Note50}.

We now give a brief outline of the organisation of this section. \cref{Sec:PG32} collects the group-theoretic and geometric background underlying the $(\!(8, 2^4, 2)\!)$ phantom code. Readers interested only in the explicit code construction may skip ahead to \cref{Sec:CodeConstruction} and consult \cref{Tab:Lines,Tab:xinL} for the required data. In \cref{Sec:PermutationInv}, we demonstrate that, in addition to the $A_8$-invariance of the code space required by the phantom property, odd permutations also preserve the code, meaning that it is invariant under the full physical permutation group $S_8$. In fact, odd permutations realise the natural duality of $\PG$, and implement a non-Clifford gate $U_{\text{c}}$ on the logical space. In \cref{Sec:NonClifford} we show that, in addition to this operation, the code also possesses a strongly transversal non-Clifford gate implemented by $T^{\otimes 8}$. By Schur-Weyl duality, the invariance of the code under $S_8$ guarantees that it decomposes into permutation modules. In \cref{Sec:PermutationModule} we analyse this decomposition and employ it to find a complete set of (non-Pauli) operators stabilising the code space. Finally, in \cref{Sec:NoPauli}, we prove that no Pauli stabiliser phantom code of type $[\![8, 4]\!]$ exists. This shows that our code is not locally unitary equivalent to any phantom stabiliser code, and explains its absence from the exhaustive enumeration of CSS phantom codes carried out in Ref.~\cite{koh2026entanglinglogicalqubitsphysical}.

\subsection{\texorpdfstring{$\PG$}{PG(3, 2)} and the exceptional isomorphism \texorpdfstring{$\GLx\cong A_8$}{GL(4, 2)=A8}}\label{Sec:PG32}
Rather than reproducing a complete proof of the abstract isomorphism $\GLx\cong A_8$, which may be found, for example, in Refs.~\cite{Conway1999, Murray, Taylor}, we provide here an explicit realisation that is particularly convenient for the constructions that follow. This construction is entirely concrete, although it is necessarily non-canonical, since it depends both on the choice of basis of $\mathbb{F}_2^4$ and the labelling of the eight letters permuted by $A_8$. To construct the isomorphism, we will define a map $\phi:\GLx\to A_8$ by specifying its action on a generating set of $\GLx$, before checking that it is well-defined and confirming that its image is all of $A_8$.

Let $E_{ij}$ be the $4\times 4$ matrix with entries $(E_{ij})_{kl} = \delta_{ik}\delta_{jl}$, and let $g_{ij} = \mathbbm{1}+E_{ij}\in \GLx$ be an elementary transvection. The six adjacent transvections
% \begin{equation}\label{Eq:transvections}
%     g_{12}, g_{23}, g_{34}, g_{21}, g_{32}, \text{ and } g_{43}
% \end{equation}
\begin{equation}\label{Eq:transvections}
    \begin{aligned}
        & g_{12} = \smqty(1 & 1 & 0 & 0 \\ 0 & 1 & 0 & 0 \\ 0 & 0 & 1 & 0 \\ 0 & 0 & 0 & 1), && g_{23} = \smqty(1 & 0 & 0 & 0 \\ 0 & 1 & 1 & 0 \\ 0 & 0 & 1 & 0 \\ 0 & 0 & 0 & 1), && g_{34} = \smqty(1 & 0 & 0 & 0 \\ 0 & 1 & 0 & 0 \\ 0 & 0 & 1 & 1 \\ 0 & 0 & 0 & 1), \\
        & g_{21} = \smqty(1 & 0 & 0 & 0 \\ 1 & 1 & 0 & 0 \\ 0 & 0 & 1 & 0 \\ 0 & 0 & 0 & 1), && g_{32} = \smqty(1 & 0 & 0 & 0 \\ 0 & 1 & 0 & 0 \\ 0 & 1 & 1 & 0 \\ 0 & 0 & 0 & 1), && g_{43} = \smqty(1 & 0 & 0 & 0 \\ 0 & 1 & 0 & 0 \\ 0 & 0 & 1 & 0 \\ 0 & 0 & 1 & 1).
    \end{aligned}
\end{equation}
generate $\GLx$~\cite{ConradSimple}. We specify their images in $A_8$ under $\phi$ to be
\begin{equation}\label{Eq:PhiDef}
\begin{aligned}
    &\phi(g_{12}) \coloneqq (1\; 2)(3\; 4)(5\; 6)(7\; 8), &&\phi(g_{23}) \coloneqq (1\; 5)(2\; 8)(3\; 7)(4\; 6), &&\phi(g_{34}) \coloneqq (1\; 2)(3\; 8)(4\; 7)(5\; 6),\\
    &\phi(g_{21}) \coloneqq (1\; 4)(2\; 7)(3\; 8)(5\; 6), &&\phi(g_{32}) \coloneqq (1\; 6)(2\; 5)(3\; 7)(4\; 8), &&\phi(g_{43}) \coloneqq (1\; 4)(2\; 3)(5\; 6)(7\; 8).
\end{aligned}
\end{equation}
As noted above, this choice is not canonical but simply one convenient labelling for which the ensuing calculations work out cleanly. To prove that this assignment extends to a homomorphism, we make use of the standard presentation of $\GLkF$ in terms of the elementary transvections.

\begin{proposition}[Theorem 2.3.8, Ref.~\cite{Hahn1989}]
    For $k\geq 3$, the group $\GLkF$ has a presentation with generators $e_{ij}$ for $i\neq j$, subject to the relations
    \begin{enumerate}
        \item $e_{ij}^2 = 1$,
        \item $[e_{ij}, e_{kl}] = 1$ if $j\neq k$ and $i\neq l$, and
        \item $[e_{ij}, e_{jk}] = e_{ik}$ if $i$, $j$, and $k$ are distinct,
    \end{enumerate}
    where $[x, y] \coloneqq x^{-1}y^{-1}xy$ is the group commutator.
\end{proposition}

These relations allow one to recover all elementary transvections from the six adjacent ones. For example, $g_{13} = [g_{12}, g_{23}]$, $g_{24} = [g_{23}, g_{34}]$, $g_{14} = [g_{13}, g_{34}] = [g_{12}, g_{24}]$, and similarly for the remaining $g_{ij}$. The same formulas therefore determine the images of the remaining transvections under $\phi$. Concretely, one finds
\begin{subequations}\label{Eq:PhiExtra}
\begin{align}
    \phi(g_{13}) ={}& [\phi(g_{12}), \phi(g_{23})] = (1\;3)(2\;4)(5\;7)(6\;8)\\
    \phi(g_{24}) ={}& [\phi(g_{23}), \phi(g_{34})] = (1\;7)(2\;4)(3\;5)(6\;8)\\
    \phi(g_{14}) ={}& [\phi(g_{13}), \phi(g_{34})] = (1\;5)(2\;6)(3\;7)(4\;8)
\end{align}
\end{subequations}
(these particular elements will be needed later when we construct a certain stabiliser subgroup).

\begin{lemma}
    The rule $\phi$ extends to a well-defined isomorphism $\phi:\GLx\overset{\sim}{\to}A_8$.
\end{lemma}

\begin{proof}
    To show that the assignment given in \cref{Eq:PhiDef} defines a homomorphism, it is sufficient to verify that the chosen images of the elementary generators satisfy the defining relations of $\GLx$. This is readily confirmed by direct calculation. Since each $\phi(g_{ij})$ is a product of four transpositions, it is an even permutation, so the image lies in $A_8$. To prove that $\phi$ is an isomorphism, it suffices to compare the orders of the groups $A_8$, $\GLx$, and $\Im\phi$. On the one hand, $|\GLx| = \prod_{l=0}^{4-1}(2^4-2^l) = 20160$, while $|A_8| = 8!/2 = 20160$. Finally, it can be verified - for example, by using a computer algebra system such as \texttt{Mathematica} - that the subgroup of $S_8$ generated by the six permutations $\phi(g_{12}), \phi(g_{23}), \phi(g_{34}), \phi(g_{21}), \phi(g_{32})$, and $\phi(g_{43})$ also has order 20160. Hence, $\phi$ is an injective homomorphism from $\GLx$, and its image is all of $A_8$, so $\phi$ is an isomorphism. 
\end{proof}

\begin{figure}
    \centering
\sidesubfloat[]{\hspace{-1.5em}\includegraphics[width=0.47\linewidth]{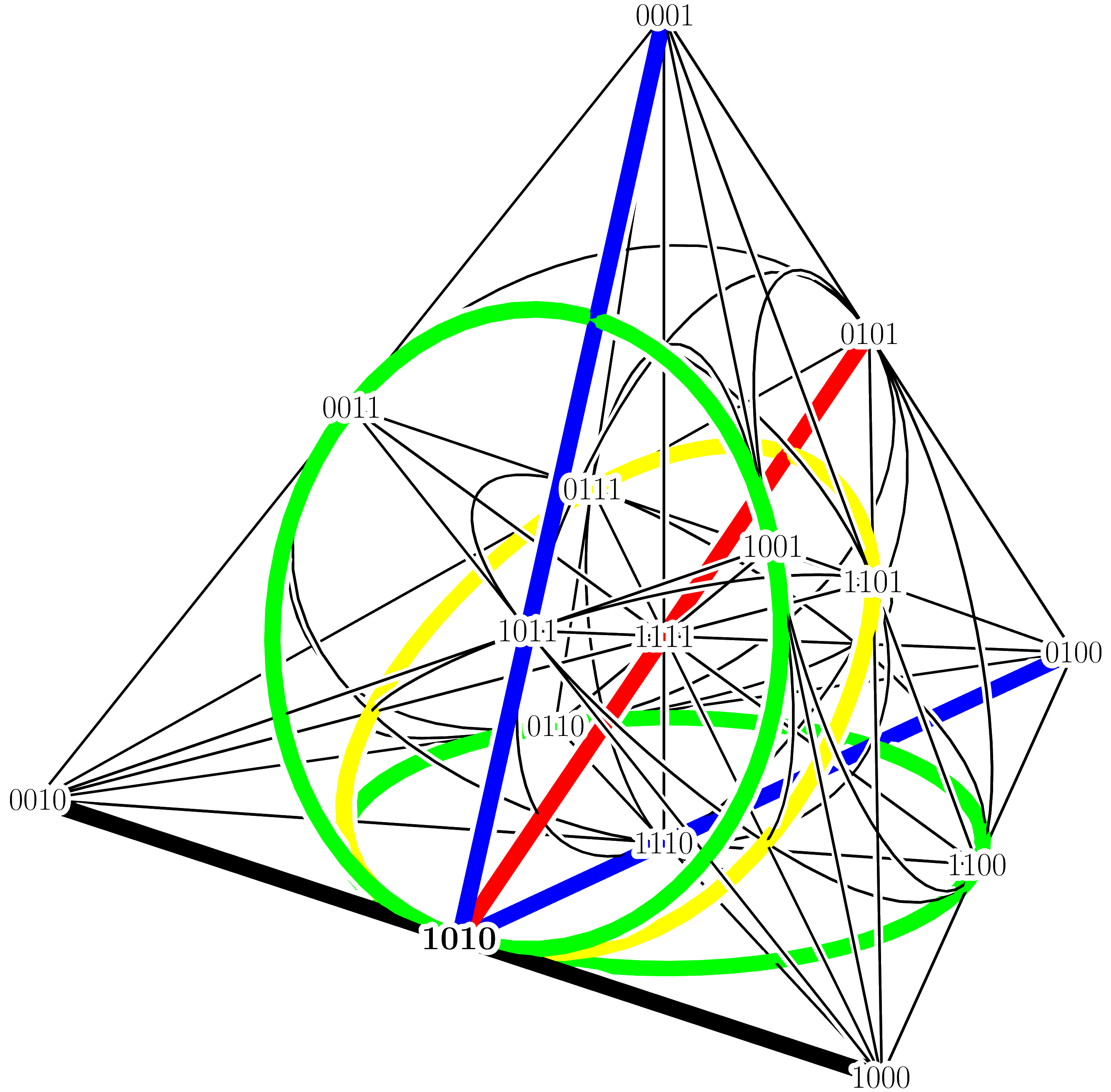}}
\hfill
\sidesubfloat[]{\hspace{-1.5em}\includegraphics[width=0.47\linewidth]{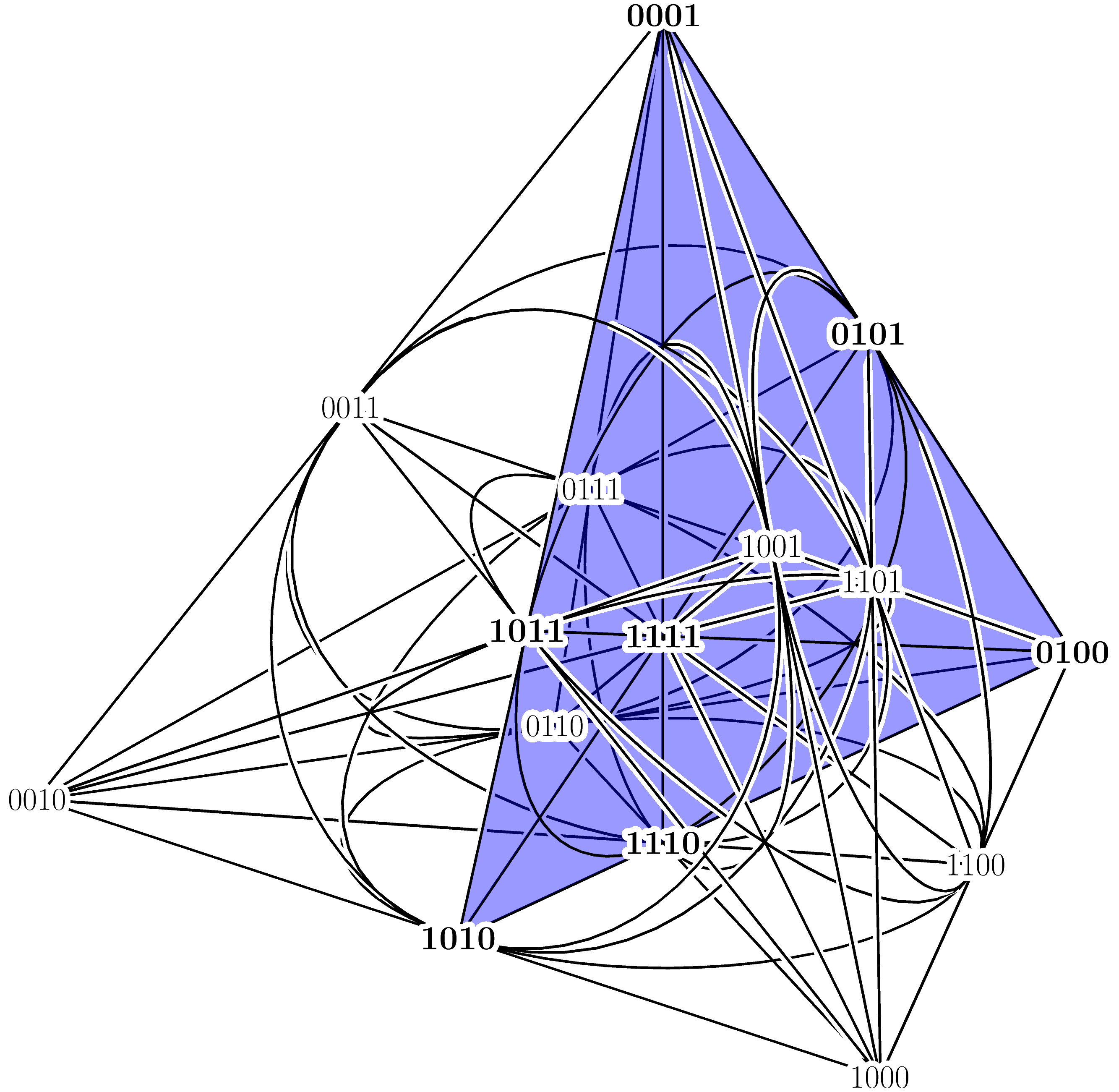}}
    \caption{\textbf{$\boldsymbol{\mathrm{PG}(3, 2)}$.} The finite projective geometry $\PG$ consists of the 15 non-zero elements of $\mathbb{F}_2^4$. \textbf{(a)} The 35 projective lines of $\PG$, shown here both as straight lines and loops within the tetrahedron, correspond to the 2-dimensional linear subspaces of $\mathbb{F}_2^4$. Every line contains 3 points, and every point is contained in 7 lines. (The colours of the highlighted lines were chosen for purely aesthetic reasons, and have no intrinsic meaning.) \textbf{(b)} The 15 projective planes of $\PG$ correspond to the 3-dimensional linear subspaces of $\mathbb{F}_2^4$. Every plane contains 7 lines and 7 points. The particular plane highlighted here is the orthogonal complement of the point $1010$.}
    \label{Fig:PG32}
\end{figure}

We now use $\phi$ to compare two natural, faithful, transitive actions of the isomorphic groups $\GLx$ and $A_8$. 

First, recall that $\GLx$ acts faithfully on the non-zero vectors of $\mathbb{F}_2^4$ via matrix multiplication (cf.~\cref{Sec:GLKFAction}). Since $\mathbb{F}_2^*=\{1\}$, each 1-dimensional linear subspace of $\mathbb{F}_2^4$ contains exactly one non-zero vector, so we can identify the points of the projective space $\PG\coloneqq \mathbb{F}_2^4\backslash\{0\}$ with the $2^4-1=15$ non-zero vectors themselves. 

Before turning to the induced group action of $\GLx$ on subspaces, it is useful to recount some facts about the finite geometry $\PG$. Every pair of distinct points $u, v\in\PG$ spans a unique line, namely $\langle u, v\rangle = \{u, v, u+v\}$, and the set of all such lines is denoted
\begin{equation}
    L \coloneqq \left\{\ell = \{u, v, u+v\}\;\middle|\; u, v\in\PG, u\neq v\right\}.
\end{equation}
These projective lines correspond to the 2-dimensional linear subspaces of $\mathbb{F}_2^4$, and their number is given by the Gaussian binomial coefficient $|L| = \genfrac{[}{]}{0pt}{}{4}{2}_2 = 35$. Dually, the planes of $\PG$ correspond to the 3-dimensional linear subspaces of $\mathbb{F}_2^4$. There are $\genfrac{[}{]}{0pt}{}{4}{3}_2 = \genfrac{[}{]}{0pt}{}{4}{1}_2 = 15$ such planes, and each such plane contains $2^3-1 = 7$ points. The incidence structure is highly regular: every point lies on exactly 7 lines, every line contains exactly 3 points, every plane contains exactly 7 points and 7 lines, and every line lies in exactly 3 planes. This remarkably symmetric structure is illustrated in \cref{Fig:PG32} with the `tetrahedral model' of $\PG$. \cref{Tab:Lines,Tab:xinL} moreover give the incidence data for all lines and points in $\PG$.

The action of $\GLx$ on the vectors of $\mathbb{F}_2^4$ preserves the dimensions of linear subspaces, and therefore induces actions on the points, lines and planes of $\PG$. In particular, the action of $\GLx$ on $L$ given by $g\cdot\{u, v, u+v\} = \{gu, gv, gu+gv\}$ is both faithful and transitive. On the $A_8$ side, there is a standard degree-35 action on the set $W$ of unordered partitions of an 8-element set into two subsets of size 4~\cite[Chapter 6]{Taylor}. It is convenient to encode such partitions by binary strings: writing $\bar{b}:=1\cdots 1+b$ for the bitwise complement of a binary string $b\in\mathbb{F}_2^8$, define
\begin{equation}
    W \coloneqq \{b\in \mathbb{F}_2^8 \;|\; |b| = 4\}/\{b\sim \bar{b}\},
\end{equation}
where $|b|$ denotes the Hamming weight of $b$. Note that $|W| = \binom{8}{4}/2 = 35 = |L|$. Since the natural permutation action of $A_8$ is 4-transitive~\cite{wolfram}, it induces a transitive action on $W$.

The key point for our purposes is that, perhaps surprisingly, the action of $\GLx$ on the 35 projective lines $L$ is isomorphic to the action of $A_8$ on the unordered 4-4 bipartitions $W$. In \cref{Sec:CodeConstruction} we will utilise this correspondence to construct a length-8 phantom code with four logical qubits. To prove this isomorphism, we will use the following elementary lemma, which provides a convenient criterion for recognising when two transitive $G$-sets are equivalent. Recall that if a group $G$ acts on a set $X$, the stabiliser of an element $x\in X$ is $\Stab_G(x) \coloneqq\{g\in G\,|\,g\cdot x=x\}$.

\begin{table}[]
\setlength{\tabcolsep}{.4em}
    \centering
    \begin{tabular}{ccclc||ccclc}
$\ell$ & $\{u, v, u+v\}$ & $b(\ell)$ & Generators & $\ell^\perp$  & $\ell$ & $\{u, v, u+v\}$ & $b(\ell)$ & Generators & $\ell^\perp$ \\
\hline
1 & 1000, 0100, 1100 & 10101010 & $\mathrm{Id}$ & 20 & 19 & 1100, 1011, 0111 & 00100111 & $g_{21}g_{43}g_{32}$ & 35 \\ 
2 & 1000, 0010, 1010 & 00110011 & $g_{23}g_{32}$ & 10 & 20 & 0010, 0001, 0011 & 10100110 & $g_{23}g_{12}g_{34}g_{32}g_{23}g_{21}g_{43}g_{32}$ & 1 \\ 
3 & 1000, 0110, 1110 & 01100110 & $g_{32}$ & 28 & 21 & 0010, 1001, 1011 & 01101010 & $g_{23}g_{43}g_{34}g_{32}g_{21}g_{43}g_{32}$ & 11 \\ 
4 & 1000, 0001, 1001 & 11110000 & $g_{34}g_{23}g_{43}g_{32}$ & 8 & 22 & 0010, 0101, 0111 & 01110100 & $g_{23}g_{12}g_{32}g_{23}g_{21}g_{43}g_{32}$ & 5 \\ 
5 & 1000, 0101, 1101 & 01011010 & $g_{32}g_{43}g_{32}$ & 22 & 23 & 0010, 1101, 1111 & 01000111 & $g_{43}g_{32}g_{23}g_{21}g_{43}g_{32}$ & 17 \\ 
6 & 1000, 0011, 1011 & 11000011 & $g_{23}g_{43}g_{32}$ & 12 & 24 & 1010, 0001, 1011 & 10101001 & $g_{34}g_{23}g_{34}g_{32}g_{21}g_{43}g_{32}$ & 9 \\ 
7 & 1000, 0111, 1111 & 01101001 & $g_{43}g_{32}$ & 30 & 25 & 1010, 1001, 0011 & 01100101 & $g_{23}g_{34}g_{32}g_{21}g_{43}g_{32}$ & 13 \\ 
8 & 0100, 0010, 0110 & 11010010 & $g_{12}g_{23}g_{21}g_{32}$ & 4 & 26 & 1010, 0101, 1111 & 11010001 & $g_{43}g_{32}g_{21}g_{43}g_{32}$ & 26 \\
9 & 0100, 1010, 1110 & 10000111 & $g_{12}g_{21}g_{32}$ & 24 & 27 & 1010, 1101, 0111 & 11100010 & $g_{34}g_{43}g_{32}g_{21}g_{43}g_{32}$ & 34 \\ 
10 & 0100, 0001, 0101 & 01100011 & $g_{12}g_{34}g_{23}g_{21}g_{43}g_{32}$ & 2 & 28 & 0110, 0001, 0111 & 00111010 & $g_{12}g_{34}g_{32}g_{23}g_{21}g_{43}g_{32}$ & 3 \\ 
11 & 0100, 1001, 1101 & 00110110 & $g_{12}g_{34}g_{21}g_{43}g_{32}$ & 21 & 29 & 0110, 1001, 1111 & 10100011 & $g_{43}g_{34}g_{32}g_{21}g_{43}g_{32}$ & 29 \\ 
12 & 0100, 0011, 0111 & 10110001 & $g_{12}g_{23}g_{21}g_{43}g_{32}$ & 6 & 30 & 0110, 0101, 0011 & 00010111 & $g_{12}g_{32}g_{23}g_{21}g_{43}g_{32}$ & 7 \\ 
13 & 0100, 1011, 1111 & 00011011 & $g_{12}g_{21}g_{43}g_{32}$ & 25 & 31 & 0110, 1101, 1011 & 01110001 & $g_{12}g_{32}g_{21}g_{43}g_{32}$ & 33 \\ 
14 & 1100, 0010, 1110 & 11100001 & $g_{23}g_{21}g_{32}$ & 16 & 32 & 1110, 0001, 1111 & 00110101 & $g_{34}g_{32}g_{23}g_{21}g_{43}g_{32}$ & 15 \\ 
15 & 1100, 1010, 0110 & 01001011 & $g_{21}g_{32}$ & 32 & 33 & 1110, 1001, 0111 & 01010011 & $g_{34}g_{32}g_{21}g_{43}g_{32}$ & 31 \\ 
16 & 1100, 0001, 1101 & 10010011 & $g_{34}g_{23}g_{21}g_{43}g_{32}$ & 14 & 34 & 1110, 0101, 1011 & 10110010 & $g_{32}g_{21}g_{43}g_{32}$ & 27 \\ 
17 & 1100, 1001, 0101 & 00111001 & $g_{34}g_{21}g_{43}g_{32}$ & 23 & 35 & 1110, 1101, 0011 & 00101011 & $g_{32}g_{23}g_{21}g_{43}g_{32}$ & 19 \\
18 & 1100, 0011, 1111 & 01110010 & $g_{23}g_{21}g_{43}g_{32}$ & 18 & 
    \end{tabular}
    \caption{The projective space $\PG$ contains exactly 35 projective lines $\ell = \langle u, v\rangle$. Since the action of $\GLx$ on these lines is transitive, every line $\ell$ may be written as $\ell = g_\ell\cdot \ell_0$ for some $g_\ell\in \GLx$, where $\ell_0 = \langle e_1, e_2\rangle$ is the fixed reference line. The column labelled `Generators' records one such element $g_{\ell}$, expressed as a word in the elementary transvections listed in \cref{Eq:transvections}.     Under the action isomorphism between $\GLx\curvearrowright L$ and $A_8\curvearrowright W$, each line corresponds to an unordered complementary pair $w(\ell) = \{b(\ell), \bar{b}(\ell)\}\in W$.     The column `$b(\ell)$' provides a representative string $b(\ell)$, chosen as $b(\ell) \coloneqq \phi(g_\ell)\cdot b_0$, where $b_0 = b(\ell_0) = 10101010$. Finally, the column `$\ell^\perp$' gives the index of the dual lines under the orthogonal complement operation defined in \cref{Eq:Dual}}
    \label{Tab:Lines}
\end{table}

\begin{lemma}\label{Lem:Transitive}
    Let $G$ act transitively on sets $X$ and $Y$, and fix points $x_0\in X$ and $y_0\in Y$. If $\Stab_G(x_0) = \Stab_G(y_0)$ then the rule defined by $F(x_0) = y_0$ and $F(g\cdot x_0) = g\cdot y_0$ is a well-defined $G$-equivariant bijection $F:X\to Y$.
\end{lemma}

\begin{proof} 
    Since the action of $G$ on $X$ is transitive, every $x\in X$ can be written in the form $x=g\cdot x_0$ for some $g\in G$. If two such expressions represent the same element, say $x = g\cdot x_0 = g'\cdot x_0$, then $ g^{-1}g'\in \Stab_G(x_0)$. By assumption this implies $g^{-1} g'\in \Stab_G(y_0)$, so $g\cdot y_0 = g'\cdot y_0$ and the formula $F$ is independent of the chosen representative $g$, i.e. it is well-defined. This map is $G$-equivariant by construction,
    \begin{equation}
        F(g\cdot x) = F(g\cdot (g'\cdot x_0)) = (gg')\cdot y_0 = g\cdot F(g'\cdot x_0) = g\cdot F(x),
    \end{equation}
    and it is surjective because the action of $G$ on $Y$ is transitive. Finally, it is injective because if $F(x) = F(x')$ with $x = g\cdot x_0$ and $x'=g'\cdot x_0$, then $g\cdot F(x_0)=g'\cdot F(x_0)$, so $g^{-1} g'\in \Stab_G(y_0) = \Stab_G(x_0)$, which implies $x = g\cdot x_0 = g'\cdot x_0 = x'$.
\end{proof}

We next compute the stabiliser of a convenient reference line. Let $\ell_0\coloneqq\langle e_1, e_2\rangle$, where $e_1, \ldots, e_4$ is the standard basis of $\mathbb{F}_2^4$.

\begin{lemma}\label{Lemma:StabG4}
    The stabiliser group of $\ell_0$ under $\GLx$ has order $576$ and is generated by $g_{12}, g_{21}, g_{34}, g_{43}, g_{13}, g_{14}, g_{23}$ and $g_{24}$.
\end{lemma}

\begin{proof}
    A matrix $g\in\GLx$ stabilises $\ell_0$ if and only if it preserves the 2-dimensional subspace spanned by $e_1$ and $e_2$. Relative to the decomposition $\mathbb{F}_2^4 = \langle e_1, e_2\rangle\oplus \langle e_3, e_4\rangle$, this means precisely that $g$ is block upper-triangular, so
    \begin{equation}
        \Stab_{\GLx}(\ell_0) = \left\{\mqty(A & B\\ 0 & D) \;\middle|\; A, D\in \mathsf{GL}_2(\mathbb{F}_2), B\in M_{2\times 2}(\mathbb{F}_2)\right\}.
    \end{equation}
    The order of this group is therefore $|\Stab_{\GLx}(\ell_0)| = |\mathsf{GL}_2(\mathbb{F}_2)|^2 \cdot |M_{2\times 2}(\mathbb{F}_2)| = 6^2\cdot 2^4 = 576$. To determine a set of generators of $\Stab_{\GLx}(\ell_0)$, observe that any matrix in this group factors as 
    \begin{equation}
        \mqty(A & B\\ 0 & D) = \mqty(\mathbbm{1} & BD^{-1}\\ 0 & \mathbbm{1})\mqty(A & 0\\ 0 & \mathbbm{1})\mqty(\mathbbm{1} & 0\\ 0 & D),
    \end{equation}
    so defining
    \begin{align}
        & M =  \left\{\mqty(\mathbbm{1} & B'\\ 0 & \mathbbm{1}) \;\middle|\; B'\in M_{2\times 2}(\mathbb{F}_2)\right\}, && K_1 = \left\{\mqty(A & 0\\ 0 & \mathbbm{1}) \;\middle|\; A\in \mathsf{GL}_2(\mathbb{F}_2)\right\}, && K_2 = \left\{\mqty(\mathbbm{1} & 0\\ 0 & D) \;\middle|\; D\in \mathsf{GL}_2(\mathbb{F}_2)\right\},
    \end{align} 
    we have that $\Stab_{\GLx}(\ell_0) = \langle M, K_1, K_2\rangle$. The subgroup $K_1\cong \mathsf{GL}_2(\mathbb{F}_2)$ is generated by $g_{12}$ and $g_{21}$, while $K_2\cong \mathsf{GL}_2(\mathbb{F}_2)$ is generated by $g_{34}$ and $g_{43}$. Finally, for any $B = \begin{psmallmatrix}b_{11} & b_{12} \\ b_{21} & b_{22}\end{psmallmatrix}\in M_{2\times 2}(\mathbb{F}_2)$ one has
    \begin{equation}
        \mqty(\mathbbm{1} & B\\ 0 & \mathbbm{1}) = g_{13}^{b_{11}} g_{14}^{b_{12}} g_{23}^{b_{21}} g_{24}^{b_{22}},
    \end{equation}
    so $M$ is generated by $g_{13}, g_{14}, g_{23},$  and $g_{24}$. Combining these three generating families proves the claim.
\end{proof}

\begin{table}[]
    \centering
    \setlength{\tabcolsep}{.7em}
    \begin{tabular}{cccc}
$x$ & $\ell\ni x$ & $y \in x^\perp$ & $\ell\subset x^\perp$ \\
\hline
1000 & 1, 2, 3, 4, 5, 6, 7 & 0100, 0010, 0110, 0001, 0101, 0011, 0111 & 8, 10, 12, 20, 22, 28, 30\\ 
0100 & 1, 8, 9, 10, 11, 12, 13 & 1000, 0010, 1010, 0001, 1001, 0011, 1011 & 2, 4, 6, 20, 21, 24, 25\\ 
1100 & 1, 14, 15, 16, 17, 18, 19 & 1100, 0010, 1110, 0001, 1101, 0011, 1111 & 14, 16, 18, 20, 23, 32, 35\\ 
0010 & 2, 8, 14, 20, 21, 22, 23 & 1000, 0100, 1100, 0001, 1001, 0101, 1101 & 1, 4, 5, 10, 11, 16, 17\\
1010 & 2, 9, 15, 24, 25, 26, 27 & 0100, 1010, 1110, 0001, 0101, 1011, 1111 & 9, 10, 13, 24, 26, 32, 34\\ 
0110 & 3, 8, 15, 28, 29, 30, 31 & 1000, 0110, 1110, 0001, 1001, 0111, 1111 & 3, 4, 7, 28, 29, 32, 33\\ 
1110 & 3, 9, 14, 32, 33, 34, 35 & 1100, 1010, 0110, 0001, 1101, 1011, 0111 & 15, 16, 19, 24, 27, 28, 31\\ 
0001 & 4, 10, 16, 20, 24, 28, 32 & 1000, 0100, 1100, 0010, 1010, 0110, 1110 & 1, 2, 3, 8, 9, 14, 15\\ 
1001 & 4, 11, 17, 21, 25, 29, 33 & 0100, 0010, 0110, 1001, 1101, 1011, 1111 & 8, 11, 13, 21, 23, 29, 31\\ 
0101 & 5, 10, 17, 22, 26, 30, 34 & 1000, 0010, 1010, 0101, 1101, 0111, 1111 & 2, 5, 7, 22, 23, 26, 27\\ 
1101 & 5, 11, 16, 23, 27, 31, 35 & 1100, 0010, 1110, 1001, 0101, 1011, 0111 & 14, 17, 19, 21, 22, 33, 34\\ 
0011 & 6, 12, 18, 20, 25, 30, 35 & 1000, 0100, 1100, 0011, 1011, 0111, 1111 & 1, 6, 7, 12, 13, 18, 19\\ 
1011 & 6, 13, 19, 21, 24, 31, 34 & 0100, 1010, 1110, 1001, 1101, 0011, 0111 & 9, 11, 12, 25, 27, 33, 35\\ 
0111 & 7, 12, 19, 22, 27, 28, 33 & 1000, 0110, 1110, 0101, 1101, 0011, 1011 & 3, 5, 6, 30, 31, 34, 35\\
1111 & 7, 13, 18, 23, 26, 29, 32 & 1100, 1010, 0110, 1001, 0101, 0011, 1111 & 15, 17, 18, 25, 26, 29, 30
    \end{tabular}
    \caption{Incidence data for the points, lines, and planes of $\PG$. Each of the 15 points in $\PG$ is contained in 7 lines, listed in the column labelled `$\ell\ni x$'. The lines are recorded according to the indices assigned in \cref{Tab:Lines}. Similarly, each plane contains exactly 7 points. The column `$y\in x^\perp$' records the 7 points $y$ contained in the orthogonal complement $x^\perp$ of the point $x$ under the duality operation defined in \cref{Eq:Dual}, while `$\ell\subset x^\perp$' records the 7 lines contained in each of these planes. Each line is contained in exactly 3 planes.}
    \label{Tab:xinL}
\end{table}

We may now compare the $\GLx$ and $A_8$ actions described above.

\begin{lemma}
    The action of $\GLx$ on $L$ is isomorphic to the action of $A_8$ on $W$.
\end{lemma}

\begin{proof}
    Since $\phi:\GLx\to A_8$ is an isomorphism, we may transport the natural action of $A_8$ on $W$ to an action of $\GLx$ on this same set by defining $g\star w \coloneqq \phi(g)\cdot w$. Here $w=\{b, \bar{b}\}\in W$ is an equivalence class of length-8 binary strings, and $\cdot$ denotes the standard action of $A_8$ on $W$. The action $\star$ is transitive because the original $A_8$ action is transitive.

    Having fixed the reference line $\ell_0 = \langle e_1, e_2\rangle$, we now choose a compatible base string in $W$. Defining $b_0 \coloneqq 10101010 \in \mathbb{F}_2^8$, it may be checked directly that for each generator $g\in\{g_{12}, g_{21}, g_{34}, g_{43}, g_{13}, g_{14}, g_{23}, g_{24}\}$ of $\Stab_{\GLx}(\ell_0)$ (cf. \cref{Lemma:StabG4}), the corresponding element $\phi(g)\in A_8$ stabilises the unordered complementary pair $w_0 \coloneqq \{b_0, \bar{b}_0\}\in W$ (the explicit forms of these permutations are given in \cref{Eq:PhiDef,Eq:PhiExtra}). Hence $\phi(\Stab_{\GLx}(\ell_0)) \subseteq \Stab_{A_8}(w_0)$, and because $\phi$ is an isomorphism we have $|\phi(\Stab_{\GLx}(\ell_0))| = 576$. It remains to compute the order of $\Stab_{A_8}(w_0)$. To preserve the unordered pair $w_0$, a permutation may either preserve both $b_0$ and $\bar{b}_0$, or swap them, so inside $S_8$ the stabiliser has order $2\cdot 4!\cdot 4! = 1152$. Exactly half of these permutations are even, so $|\Stab_{A_8}(w_0)| = 1152/2 = 576$. Thus we have an inclusion between two subgroups of the same finite order, which implies that they must be equal: $\phi(\Stab_{\GLx}(\ell_0)) = \Stab_{A_8}(w_0)$.

    We are now in precisely the setting of \cref{Lem:Transitive}. The actions of $\GLx$ on $L$ and $W$ are both transitive, and the chosen basepoints have the same stabiliser. It therefore follows that the function $w:L\to W$ defined as $w(\ell_0) \coloneqq w_0$ and $w(\ell) = w(g\cdot\ell_0) \coloneqq \phi(g)\cdot w(\ell_0)$ is a well-defined $\GLx$-equivariant bijection. This identifies the action of $\GLx$ on projective lines with the action of $A_8$ on unordered 4-4 bipartitions of an 8-element set.
\end{proof}

To record explicit representatives for the classes $w(\ell)\in W$, we fix for each $\ell\in L$ an element $g_\ell\in\GLx$ such that $\ell=g_\ell\cdot\ell_0$, and define an auxiliary representative $b:L\to \mathbb{F}_2^8$ by setting $b(\ell_0) \coloneqq b_0 = 10101010$ and $b(\ell) \coloneqq \phi(g_\ell)\cdot b(\ell_0)$. In terms of $b$, we then have $w(\ell) = \{b(\ell), \bar{b}(\ell)\}$. The values of the function $b$, along with the associated group elements $g_\ell$, are tabulated for each projective line $\ell\in L$ in \cref{Tab:Lines}. Note that, while $w(\ell) = \{b(\ell), \bar{b}(\ell)\}$ is equivariant with respect to $\GLx$, satisfying $w(g\cdot\ell) = \phi(g)\cdot w(\ell)$ for all $g\in\GLx$ and $\ell\in L$, the representative $b(\ell)$ does \textit{not} have this property: in general, either $b(g\cdot\ell) = \phi(g)\cdot b(\ell)$ or $b(g\cdot\ell) = \phi(g)\cdot\bar{b}(\ell)$. For example, $b(g_{12}\cdot\ell_0) = b(\ell_0) = 10101010$ as $\ell_0$ is stabilised by $g_{12}$, but $\phi(g_{12})\cdot b(\ell_0) = 01010101 = \bar{b}_0$ is the conjugate of the representative string.

Finally, we explain how this model interacts with the natural duality on $\PG$. Consider the standard symmetric bilinear form on $\PG$, defined as $\langle x, y\rangle = \sum_i x_i y_i$. Because the field has characteristic 2, this is not an inner product in the usual positive-definite sense, as there exist non-zero vectors $x$ with $\langle x, x\rangle = 0$, for example $x=1100$ (such vectors are called \textit{isotropic}). The form is nevertheless non-degenerate, which is the property required for the duality below.

Given a projective linear subspace $S$ of $\PG$, the orthogonal complement of $S$ is defined as 
\begin{equation}\label{Eq:Dual}
    S^\perp \coloneqq \{x\in \PG\;|\; \langle x, y\rangle = 0 \text{ for all }y\in S\}.
\end{equation}
Since $\langle\cdot, \cdot\rangle$ is non-degenerate, the operation $S\mapsto S^\perp$ interchanges 1-dimensional and 3-dimensional subspaces of $\mathbb{F}_2^4$, and sends 2-dimensional subspaces to 2-dimensional subspaces. Equivalently, within the projective space it exchanges points and planes, while inducing a bijection from the set of lines $L$ to itself. This duality is contragredient with respect to the natural action of $\GLx$: for every $g\in \GLx$ and every projective subspace $S\leq\PG$, we have $(g\cdot S)^\perp = g^{-\intercal}\cdot S^\perp$. Accordingly, the map $g\mapsto g^{-\intercal}$ defines an automorphism of $\GLx$, and under the explicit isomorphism $\phi:\GLx\to A_8$ this yields an automorphism of $A_8$. Since $\Aut A_8 = S_8$, this automorphism is given by conjugation by some permutation in $S_8$, and because it is not an inner automorphism of $\GLx$, the conjugating permutation must be odd. With the present choices, one may take
\begin{equation}\label{Eq:TauC}
    \tau_{\text{c}} \coloneqq (2\;4)(3\;7)(5\;6),
\end{equation}
for which a direct check on the generators shows that $\tau_{\text{c}}\phi(g)\tau_{\text{c}}^{-1} = \phi(g^{-\intercal})$ for all $g\in\GLx$. Consequently, the bijection $w:L\to W$ intertwines orthogonal complementation on $L$ with the action of $\tau_{\text{c}}$ on $W$, i.e.  $w(\ell^\perp) = \tau_{\text{c}}\cdot w(\ell)$ for all $\ell\in L$. In fact, this also holds for the specific representative bit strings listed in \cref{Tab:Lines}, that is
\begin{equation}
    b(\ell^\perp) = \tau_{\text{c}}\cdot b(\ell)
\end{equation}
for all $\ell\in L$.

\subsection{Construction of the \texorpdfstring{$(\!(8, 2^4, 2)\!)$}{((8, 16, 2))} phantom code}\label{Sec:CodeConstruction}
Having discussed the actions of $\GLx\cong A_8$ on the projective lines $L$ and unordered 4-4 bipartitions $W$, we now proceed to explicitly construct an $(\!(8, 2^4, 2)\!)$ phantom code. The idea is to index the states within the logical subspace by the points of $\PG$ (plus an extra point for the logical zero state), and to use the incidence structure of this geometry to translate the geometric action of $\GLx$ on this space into a permutation action of $A_8$ on a set of 8 physical qubits.

\begin{definition}
    For each projective line $\ell\in L$, the corresponding 8-qubit \textit{line state} is given by 
    \begin{equation}
        \ket{\ell} \coloneqq \frac{\ket{b(\ell)}+\big|\overline{b(\ell)}\big\rangle}{\sqrt{2}},
    \end{equation}
    where $b(\ell)\in\mathbb{F}_2^8$ is the representative weight-4 binary string listed in \cref{Tab:Lines}, and $\bar{b}(\ell) = 1\cdots1+b(\ell)$ its bitwise complement.
\end{definition}

For each $g\in\GLx$, define $\sigma_g$ to be the unitary gate acting as the permutation $\phi(g)$ on the physical qubits, that is $\sigma_g\ket{s}\coloneqq \ket{\phi(g)\cdot s}$ for all $s\in\mathbb{F}_2^8$. As $\phi$ is an isomorphism, the map $g\mapsto \sigma_g$ is a faithful unitary representation of $\GLx$, meaning $\sigma_g^\dagger = \sigma_{g^{-1}}$ and $\sigma_{g}\sigma_{g'} = \sigma_{g g'}$ for all $g, g'\in\GLx$. The transitive action of $\GLx$ on the projective lines of $\PG$ translates, via the action isomorphism between $\GLx\curvearrowright L$ and $A_8\curvearrowright W$, into a simple transformation rule for the line states under the unitary operators $\sigma_g$. This is described by the following lemma.

\begin{lemma}\label{Lemma:Lines}
    The line states $\ket{\ell}$, $\ell\in L$ are orthonormal, $\braket{\ell}{\ell'} = \delta_{\ell, \ell'}$, and satisfy $\sigma_g \ket{\ell} = \ket{g\cdot\ell}$ for all $g\in\GLx$.
\end{lemma}

\begin{proof}
To see that the line states are orthonormal, note that if $\ell\neq \ell'$ then the equivalence classes $w(\ell) = \{b(\ell), \bar{b}(\ell)\}$ and $w(\ell') = \{b(\ell'), \bar{b}(\ell')\}$ are disjoint, so $\braket{b(\ell)}{b(\ell')} = \braket{b(\ell)}{\bar{b}(\ell')} = \braket{\bar{b}(\ell)}{b(\ell')} = \braket{\bar{b}(\ell)}{\bar{b}(\ell')} = 0$, and $\braket{\ell}{\ell'} = 0$. On the other hand, the line states are normalised by construction since $\braket{\ell}{\ell} = (\braket{b(\ell)}{b(\ell)}+\braket{\bar{b}(\ell)}{\bar{b}(\ell)})/2 = 1$.

We now prove the transformation law for the line states. For all lines $\ell\in L$ and all $g\in\GLx$, the representative string $b(\ell)$ satisfies either $b(g\cdot\ell) = \phi(g)\cdot b(\ell)$ or $b(g\cdot\ell) = \phi(g)\cdot\bar{b}(\ell)$. In either case, the superposition $\ket{b(\ell)}+\ket{\bar{b}(\ell)}$ satisfies 
\begin{equation}
    \ket{b(g\cdot\ell)}+\ket{\bar{b}(g\cdot\ell)} = \ket{\phi(g)\cdot b(\ell)}+\ket{\phi(g)\cdot\bar{b}(\ell)}.
\end{equation}
By the definition of the permutation operators $\sigma_g$, we have $\ket{\phi(g)\cdot b(\ell)} = \sigma_g\ket{b(\ell)}$ and $\ket{\phi(g)\cdot \bar{b}(\ell)} = \sigma_g\ket{\bar{b}(\ell)}$, so 
\begin{equation}\label{Eq:EquivariantLine}
    \sigma_g\ket{\ell} = \frac{\sigma_g\ket{b(\ell)}+\sigma_g\ket{\bar{b}(\ell)}}{\sqrt{2}} = \frac{\ket{\phi(g)\cdot b(\ell)}+\ket{\phi(g)\cdot\bar{b}(\ell)}}{\sqrt{2}} = \frac{\ket{b(g\cdot\ell)}+\ket{\bar{b}(g\cdot\ell)}}{\sqrt{2}} = \ket{g\cdot\ell}
\end{equation}
as claimed.
\end{proof}

As described in \cref{Sec:PG32}, the incidence structure of $\PG$ is highly regular, and in particular, each point $x\in\PG$ lies in exactly 7 distinct projective lines $\ell\in L$. We define the \textit{point-star} state
\begin{equation}\label{Eq:PointStar}
    \ket{a_x} \coloneqq \sum_{\ell\ni x} \ket{\ell}
\end{equation}
as the linear superposition of all lines containing $x$. The data required to write down these states explicitly is provided in \cref{Tab:xinL}, which lists the indices of the lines containing each point $x$. From the definition of the point-star states and the orthonormality relation $\braket{\ell}{\ell'} = \delta_{\ell,\ell'}$, we have that
\begin{equation}\label{Eq:NormAx}
    \braket{a_x}{a_y} = \sum_{\ell\ni x}\sum_{\ell'\ni y}\braket{\ell}{\ell'} = |\{\ell \in L\,|\, x, y\in\ell\}| = \begin{cases}
        7 & x=y, \\ 1 & x\neq y.
    \end{cases}
\end{equation}
That is, $\braket{a_x}{a_y}$ counts the number of lines $\ell\in L$ such that both $x$ and $y$ are contained in $\ell$. Two distinct points are contained in a unique projective line, $\braket{a_x}{a_y} = 1$ if $y\neq x$, while if $y=x$ then the fact that every point is contained in 7 lines means that $\braket{a_x}{a_x} = 7$. Additionally, the transformation of the line states given in \cref{Lemma:Lines} ensures that the $\ket{a_x}$ transform under $\GLx$ as
\begin{equation}\label{Eq:AxTrans}
    \sigma_g\ket{a_x} = \sum_{\ell\ni x}\ket{g\cdot \ell} = \sum_{(g^{-1}\cdot\ell)\ni x}\ket{\ell} = \sum_{\ell \ni g\cdot x}\ket{\ell} = \ket{a_{g\cdot x}}.
\end{equation}
This translation from a unitary permutation action on $(\mathbb{C}^2)^{\otimes 8}$ into a $\GLx$ action on $\PG$ underlies our phantom code construction. 

As well as the point-star states $\ket{a_x}$, we define the uniform superposition of all line states as $\ket{t}\coloneqq \sum_{\ell\in L} \ket{\ell}$. This state has norm $\braket{t}{t} = |L| = 35$ and is invariant under all $g\in\GLx$, since $\GLx$ acts by permuting the lines in $L$. Furthermore, we have 
\begin{equation}\label{Eq:NormAxT}
    \braket{a_x}{t} = \sum_{\ell\ni x}\sum_{\ell'\in L} \braket{\ell}{\ell'} = |\{\ell \in L\,|\, x\in\ell\}| = 7
\end{equation}
since every point is contained in exactly $7$ lines. Finally, note that 
\begin{equation}\label{Eq:TaX}
    \sum_{x\in\PG}\ket{a_x} = 3\ket{t}
\end{equation}
as every line contains 3 points.

With these definitions in hand, we are ready to construct the phantom code.

\begin{theorem}
    The states 
    \begin{equation}
        \ket{\bar{x}} \coloneqq \frac{1}{\sqrt{6}}\ket{a_x}-\frac{1}{5}\left(\frac{1}{\sqrt{6}}-\frac{1}{\sqrt{21}}\right)\ket{t}
    \end{equation}
    where $x\in\PG$, and $\ket{\bar{0}} \coloneqq (\ket{00000000}+\ket{11111111})/\sqrt{2}$, together form an orthonormal basis of logical codewords of an $(\!(8, 2^4, 2)\!)$ phantom code, satisfying $\sigma_g \ket{\bar{x}} = \ket{\overline{g\cdot x}}$ for all $g\in\GLx$ and all $x\in\mathbb{F}_2^4 = \PG\cup\{0\}$. 
\end{theorem}

\begin{proof}
    For $x\in\PG$, the orthonormality of the states $\ket{\bar{x}}$ follows from \cref{Eq:NormAx,Eq:NormAxT}, along with $\braket{t}{t}=35$. Moreover, $\braket{\bar{x}}{\bar{0}} = 0$ as all the binary states appearing in $\ket{\bar{x}}$ have weight 4, while the components of $\ket{\bar{0}}$ have weight 0 and 8.

    Next, from the transformation of the point-star states under $\GLx$ given in \cref{Eq:AxTrans}, and the invariance of $\ket{t}$, we have
    \begin{equation}
        \sigma_g\ket{\bar{x}} = \frac{1}{\sqrt{6}}\ket{a_{g\cdot x}}+\frac{1}{5}\left(\frac{1}{\sqrt{21}} - \frac{1}{\sqrt{6}}\right)\ket{t} = \ket{\overbar{g\cdot x}}
    \end{equation} 
    for all $x\in\PG$. Moreover, both $\ket{0\cdots 0}$ and $\ket{1\cdots 1}$ are invariant under all permutations, so $\sigma_g\ket{\bar{0}} = \ket{\bar{0}} = \ket{\overbar{g\cdot 0}}$, meaning the transformation rule holds for all $x\in\mathbb{F}_2^4$ as claimed.

    Finally, we compute the distance of the code. To confirm that the code has distance $d\geq 2$, we evaluate $P \tau_i P$ for all on-site Pauli operators $\tau_i \in \{X_i, Y_i, Z_i\}$, $i=1, \ldots, 8$, where $P$ is the projector onto the code space. Firstly, a single $X_i$ or $Y_i$ operator changes the Hamming weights of the binary strings in the codewords by $\pm 1$, taking them out of the code space, hence $P X_i P = P Y_i P = 0$ for all $i$. For an on-site $Z$ operator, we have $Z_i\ket{\bar{0}} = (\ket{00000000}-\ket{11111111})/\sqrt{2}$, which is orthogonal to the code space. On the other hand, all logical states $\ket{\bar{x}}$ can be expressed as linear combinations of line states $\ket{\ell}$, which satisfy $Z_i\ket{\ell} = \pm(\ket{b(\ell)} - \ket{\bar{b}(\ell)})/\sqrt{2}$. These states are orthogonal to the code space, so we have that $\bra{\bar{x}}Z_i\ket{\bar{y}} = 0$ for all $x, y\in\PG$. Overall, it follows that $P \tau_i P = 0$ for all $i$ and all single-qubit Pauli operators $\tau$.

    To see that the code has distance $\leq 2$, it is sufficient to evaluate the expectation value of $Z_1 Z_2$ in $\ket{\bar{0}}$ and $\frac{1}{\sqrt{15}}\sum_{x\in\PG}\ket{\bar{x}} = \ket{t}/\sqrt{35}$. Clearly $\bra{\bar{0}}Z_1 Z_2\ket{\bar{0}} = 1$ as both $\ket{0\cdots 0}$ and $\ket{1\cdots 1}$ are +1 eigenstates of $Z_1 Z_2$. On the other hand, $\ket{t}/\sqrt{35}$ is the equal-amplitude superposition of all weight-4 bit strings of length 8, so the corresponding expectation value can be computed by counting the number of such strings for which the first and second bits are equal, and for which they are different. We have
    \begin{equation}
        \frac{1}{35}\bra{t}Z_1 Z_2\ket{t} = \frac{|\{\text{bits }1, 2\text{ equal}\}|-|\{\text{bits }1, 2\text{ different}\}|}{|\{\text{Weight-}4 \text{ strings in }\mathbb{F}_2^8\}|} = \frac{\left[\binom{6}{2}+\binom{6}{4}\right]-2\binom{6}{3}}{\binom{8}{4}} = -\frac{1}{7}.
    \end{equation}
    Hence $\bra{\bar{0}}Z_1 Z_2\ket{\bar{0}} \neq \bra{t}Z_1 Z_2\ket{t}/35$, meaning $Z_1 Z_2$ is not a scalar on the code, and the code has distance 2.
\end{proof}

Note that $d=2$ is the highest possible distance for a code of type $(\!(8, 2^4)\!)$. A prospective $(\!(8, 2^4, 3)\!)$ code would saturate the quantum singleton bound $n-k\geq 2(d-1)$, making it a quantum maximum distance separable (QMDS) code. However, in Ref.~\cite{Huber2020quantumcodesof} it was shown that any qudit $(\!(n, K, d)\!)_D$ QMDS code must satisfy $n\leq D^2+d-2$. In the binary case $D=2$, this implies that a QMDS code of distance $d=3$ can have length at most $n\leq 2^2+3-2 = 5$. Hence no $(\!(8, 2^4, 3)\!)$ code exists, even without the phantom condition. This provides an alternative proof that the distance of the code is $\leq 2$.

\subsection{\texorpdfstring{$S_8$}{S8} invariance}\label{Sec:PermutationInv}
So far we have established that the code space is invariant under the group $A_8\cong\GLx$ of even permutations of the 8 physical qubits. In fact, as we now show, it is also invariant under all odd permutations, meaning that the permutation automorphism group of the code is maximal and equal to $S_8$. Odd permutations realise the natural duality of $\PG$ corresponding to the orthogonal complementation operation $\perp$ defined in \cref{Eq:Dual}. 

Recall that, for a point $x\in\PG$, the point-star state $\ket{a_x}$ is defined as the equal-weight superposition of all the line states containing $x$. The analogue of these states for the planes of $\PG$ are the \textit{plane} states, defined as
\begin{equation}\label{Eq:PlaneState}
    \ket{p_\Pi} \coloneqq \sum_{\ell\subset \Pi} \ket{\ell}
\end{equation}
where $\Pi$ is a plane in $\PG$, and a line $\ell$ is contained in $\Pi$ whenever every point of $\ell$ is in $\Pi$. By relating the planes and points of $\PG$ using the duality operation $\perp$, we can express the plane states in terms of the point-star states.

\begin{lemma}\label{Lem:Pspan}
    For a plane $\Pi$ in $\PG$, the corresponding plane state is given by 
    \begin{equation}\label{Eq:PlaneToX}
        \ket{p_\Pi} = \frac{1}{2}\left(\sum_{x\in\Pi}\ket{a_x} - \ket{t}\right).
    \end{equation}
    where $\ket{a_x}$ are the point-star states defined in \cref{Eq:PointStar}. 
\end{lemma}

\begin{proof}
    We evaluate the sum $\sum_{x\in\Pi} \ket{a_x} = \sum_{x\in\Pi}\sum_{\ell\ni x}\ket{\ell}$ by counting the number of times each state $\ket{\ell}$ appears on the right-hand side. If $\ell\subset\Pi$, then $\ell$ is counted once for each of its three points, since these all lie in $\Pi$. If $\ell\not\subset\Pi$, then $\ell$ intersects with $\Pi$ exactly once. Hence
    \begin{equation}\label{Eq:AxSum}
        \sum_{x\in\Pi} \ket{a_x} = \sum_{\ell\not\subset\Pi}\ket{\ell}+3\sum_{\ell\subset\Pi}\ket{\ell} = \ket{t} + 2\ket{p_\Pi},
    \end{equation}
    where we have used the definition \cref{Eq:PlaneState} and recognised that, for any plane $\Pi$, the uniform state $\ket{t}$ can be expressed as $\ket{t} = \sum_{\ell\not\subset\Pi}\ket{\ell}+\sum_{\ell\subset\Pi}\ket{\ell}$. The result then follows by rearranging \cref{Eq:AxSum}.
\end{proof}

Recall (cf. \cref{Eq:TauC}) that the duality operation $\perp$ on $\PG$ is realised at the level of bit strings by the action of a fixed odd permutation $\tau_{\text{c}} \coloneqq (2\;4)(3\;7)(5\;6)$. More precisely, for all $\ell\in L$, the representative strings $b(\ell)$ from \cref{Tab:Lines} satisfy $b(\ell^\perp) = \tau_{\text{c}} \cdot b(\ell)$. We define $\sigma_{\text{c}}$ to be the unitary permutation operator implementing $\tau_{\text{c}}$, so that $\sigma_{\text{c}}\ket{s} \coloneqq \ket{\tau_{\text{c}}\cdot s}$ for all $s\in\mathbb{F}_2^8$. Then $\sigma_{\text{c}}\ket{b(\ell)} = \ket{\tau_{\text{c}}\cdot b(\ell)} = \ket{b(\ell^\perp)}$ and similarly $\sigma_{\text{c}} \big|\overline{b(\ell)}\big\rangle = \big|\tau_{\text{c}}\cdot\overline{b(\ell)}\big\rangle = \big|\overline{\tau_{\text{c}}\cdot b(\ell)}\big\rangle = \big|\overline{b(\ell^\perp)}\big\rangle$, hence $\sigma_{\text{c}}\ket{\ell} = \ket{\ell^\perp}$. It follows that the action of $\sigma_{\text{c}}$ on the point-star states is
\begin{equation}\label{Eq:PerpAction}
    \sigma_{\text{c}} \ket{a_x} = \sum_{\ell\ni x} \sigma_{\text{c}}\ket{\ell} = \sum_{\ell\ni x}\ket{\ell^\perp} = \sum_{\ell^\perp\ni x}\ket{\ell} = \sum_{\ell\subset x^\perp}\ket{\ell} = \ket{p_{x^\perp}},
\end{equation}
where $x^\perp$ is the plane dual to the point $x$. Here we have used the fact that $x\in\ell$ if and only if $\ell^\perp\subset x^\perp$ for all $x\in\PG$, $\ell\in L$. Moreover, because the duality permutes the set of lines, the uniform line-sum $\ket{t}$ is invariant under $\sigma_{\text{c}}$, that is, $\sigma_{\text{c}}\ket{t} = \ket{t}$. By combining these results, we arrive at the following theorem.

\begin{theorem}\label{Thm:S8code}
    The duality permutation operator $\sigma_{\text{c}}$ acts as the non-Clifford unitary operator on the logical space. Defining a $2^4\times 2^4$ dimensional matrix to have entries $(U_{\text{c}})_{xy} = \bra{\bar{x}}\sigma_{\text{c}}\ket{\bar{y}}$, we have
    \begin{equation}\label{Eq:Uc}
        U_{\text{c}}\ket{x} = \frac{1}{3}\sum_{y\in x^\perp}\ket{y} - \frac{1}{6}\sum_{y\not\in x^\perp}\ket{y}.
    \end{equation}
    for all $x\in \PG$, while the action on the logical zero state is given by $U_{\text{c}}\ket{0000}=\ket{0000}$. Here $x^\perp$ is the plane in $\PG$ consisting of the 7 points orthogonal to $x$; the relevant incidence data are recorded in \cref{Tab:xinL}. Every unitary permutation operator acting on the eight physical qubits can be expressed as either $\sigma_g$ or $\sigma_{\text{c}}\sigma_g$ for some $g\in\GLx$, so the code has the maximal permutation automorphism group $\PAut Q = S_8$.
\end{theorem}

\begin{proof}
    From \cref{Eq:PlaneToX,Eq:PerpAction} we can directly compute the action of $\sigma_{\text{c}}$ on $\ket{\bar{x}}$ for any $x\in\PG$:
    \begin{equation}
        \sigma_{\text{c}}\ket{\bar{x}} = \frac{1}{\sqrt{6}}\ket{p_{x^\perp}}+\frac{1}{5}\left(\frac{1}{\sqrt{21}} - \frac{1}{\sqrt{6}}\right)\ket{t} = \frac{1}{2\sqrt{6}}\left(\sum_{y\in x^\perp}\ket{a_y} - \ket{t}\right)+\frac{1}{5}\left(\frac{1}{\sqrt{21}} - \frac{1}{\sqrt{6}}\right)\ket{t}.
    \end{equation}
    \cref{Eq:TaX} tells us that the uniform state $\ket{t}$ can be written as $\ket{t} = \frac{1}{3}\sum_{x\in\PG}\ket{a_x}$, so for any plane $\Pi$ we have $\ket{t} = \frac{1}{3}\big(\sum_{x\in\Pi}\ket{a_x} + \sum_{x\not\in\Pi}\ket{a_x}\big)$. Rearranging, we find
    \begin{equation}
        \sigma_{\text{c}}\ket{\bar{x}} = \frac{1}{3}\sum_{y\in x^\perp}\ket{\bar{y}} - \frac{1}{6}\sum_{y\not\in x^\perp}\ket{\bar{y}}.
    \end{equation}
    for all $x\in\PG$. Finally, $\ket{\bar{0}}$ is invariant under all permutations, so $\sigma_{\text{c}}\ket{\bar{0}} = \ket{\bar{0}}$ and the result holds for all $x\in\mathbb{F}_2^4$. That this operator is unitary follows from the fact that it preserves the logical subspace, while the fact that it is non-Clifford can be seen by noting that the amplitudes of the images of the logical basis states do not have uniform magnitudes.
\end{proof}

Since $\sigma_{\text{c}}$ is a product of transpositions, it squares to the identity, which implies that $U_{\text{c}}^2=\mathbbm{1}$. As a consequence, the unitary matrix $U_{\text{c}}$ can only have eigenvalues $\pm 1$. The structure of the $\pm1$-eigenspaces of $U_{\text{c}}$ within the code space can be understood by inspecting the geometrical meaning of the duality operation. Firstly, since $(S^\perp)^\perp = S$ for all subspaces $S\leq\PG$, the (unnormalised) states $\ket{\psi_x^\pm} \coloneqq \ket{a_x}\pm\ket{p_{x^\perp}}$ satisfy
\begin{equation}
    \sigma_{\text{c}}\ket{\psi_x^\pm} = \ket{p_{x^\perp}}\pm\ket{a_x} = \pm \ket{\psi_x^\pm},
\end{equation}
i.e. they are $\pm 1$ eigenstates of $\sigma_{\text{c}}$. The states $\ket{\psi^\pm_x}$ may be explicitly constructed using the data from \cref{Tab:xinL}. Because $\ket{\psi_x^+}+\ket{\psi_x^-}=2\ket{a_x}$ and $\ket{\psi_x^+}-\ket{\psi_x^-}=2\ket{p_{x^\perp}}$, it follows from \cref{Lem:Pspan}, and the fact that the states $\ket{a_x}$ span the space of non-zero logical codewords $Q_4 = Q\cap\Span\{\ket{s}\,|\,s\in\mathbb{F}_2^8, |s|=4\} = \Span\{\ket{\bar{x}}\,|\,x\in\PG\}$, that this space decomposes as $Q_4 = Q_4^+\oplus Q_4^-$, where $Q_4^\pm \coloneqq \Span\left\{\ket{\psi^\pm_x}\,\middle|\, x\in\PG \right\}$. Hence $\dim Q_4^+ + \dim Q_4^- = \dim Q_4 = 15$, meaning that the states $\ket{\psi_x^\pm}$ are not linearly independent.

To determine the dimensions of the subspaces $Q_4^\pm$, we compute the trace of $U_{\text{c}}$ over $Q_4$. A point $x\in\PG$ is said to be \textit{isotropic} if $x\in x^\perp$, otherwise $x$ is \textit{anisotropic}. $x\in\PG$ is isotropic if and only if its binary representation contains an even number of 1s, so exactly 7 points in $\PG$ are isotropic, and the remaining 8 are anisotropic. Now, from \cref{Eq:Uc} we have
\begin{equation}
    \bra{x}U_{\text{c}}\ket{x} = \begin{cases}1/3 & x\text{ is isotropic}\\ -1/6 & x\text{ is anisotropic}\end{cases}
\end{equation}
so that 
\begin{equation}\label{Eq:UcTrace}
    \dim Q_4^+ - \dim Q_4^- = \tr(U_{\text{c}}|_{Q_4}) = 7\cdot \frac{1}{3} + 8\cdot\left(-\frac{1}{6}\right) = 1,
\end{equation}
which implies $\dim Q_4^+ = 8$ and $\dim Q_4^- =7$. Combined with the fact that $U_{\text{c}}\ket{0000} = \ket{0000}$, we see that the +1 eigenvalue of this operator has multiplicity 9, while $-1$ has multiplicity 7.

\subsection{Transversal non-Clifford gate}\label{Sec:NonClifford}
\cref{Thm:S8code} establishes that the code has a non-Clifford gate implemented by a specific odd permutation $\tau_{\text{c}}$. This is not the only $\SWAP$-transversal logical gate, as the following result shows.

\begin{theorem}\label{Thm:TransversalT8}
    The operator $T^{\otimes 8}$ is a transversal gate for the $(\!(8, 2^4, 2)\!)$ phantom code, and implements the operation $2\ketbra{\bar{0}}-\mathbbm{1}$ on the logical subspace.
\end{theorem}

\begin{proof}
    Define the single-qubit phase gate $R(\theta)\coloneqq\exp(\ii\theta[\mathbbm{1}- Z]/2)=\diag(1, \ee^{\ii\theta})$, and consider the transversal action of $R(\theta)^{\otimes 8}$ on the different logical codewords. Firstly, we have
    \begin{equation}
        R(\theta)^{\otimes 8}\ket{\bar{0}} = \frac{\ket{00000000}+\ee^{8\ii\theta}\ket{11111111}}{\sqrt{2}},
    \end{equation}
    which remains in the logical space if and only if $8\theta \equiv 0 \mod 2\pi$. All the other logical codewords are superpositions of weight-4 binary strings, so $R(\theta)^{\otimes 8}\ket{\bar{x}} = \ee^{4\ii\theta}\ket{\bar{x}}$ for all $x\in\PG$. For the allowed values of $\theta$, the logical operation is non-trivial only if $4\theta \equiv \pi \mod 2\pi$, and any such choice yields the same unitary $\diag(1, -1, \ldots, -1) = 2\ketbra{0}-\mathbbm{1}$ on the logical space. In particular, the choice $\theta = \pi/4$ gives $R(\pi/4)^{\otimes 8} = T^{\otimes 8}$.
\end{proof}

\subsection{Stabilisers and permutation representation decomposition}\label{Sec:PermutationModule}
We now analyse the representation-theoretic structure of the code space under the natural action of the full permutation group $S_8$. The weight-4 part of the $(\!(8, 2^4, 2)\!)$ code $Q$ is most naturally viewed in two complementary ways. On the one hand, it is a 15-dimensional permutation module arising from the action of $S_8$ on the points of $\PG$. On the other hand, because $Q$ is $S_8$-invariant, Schur-Weyl duality identifies it as a direct sum of irreducible $S_8$ representations inside $(\mathbb{C}^2)^{\otimes 8}$. 

Let $J_\alpha = \frac{1}{2}\sum_{i=1}^n\sigma^{(i)}_\alpha$, $\alpha\in\{x, y, z\}$, be the collective spin operators, and write $J^2 = J_x^2+J_y^2+J_z^2$. By Schur-Weyl duality, the Hilbert space of physical qubits decomposes as
\begin{equation}
    (\mathbb{C}^2)^{\otimes 8} \cong \bigoplus_{\lambda\vdash 8} S^{(\lambda)}\otimes D^{(\lambda)} = \bigoplus_{j=0}^4 S^j\otimes D_j,
\end{equation}
where $\lambda=(\lambda_1\geq\lambda_2)$ with $\lambda_1+\lambda_2 = 8$ is a partition of $n=8$, and $S^{(\lambda)}$ and $D^{(\lambda)}$ respectively denote the irreducible representations of $S_8$ and $\mathsf{U}(2)$ corresponding to $\lambda$. In the second equality we have expressed this partition as $\lambda = (4+j, 4-j)$ for $j=0, \ldots, 4$, and noted that $D_j\coloneqq D^{(4+j, 4-j)}$ is the spin-$j$ representation of $\mathsf{U}(2)$, upon which the Casimir operator $J^2$ acts as the scalar $j(j+1)$. The dimensions of the corresponding $S_8$ irreps $S^j\coloneqq S^{(4+j, 4-j)}$ may be computed with the hook length formula, and are given by $\dim S^j = \binom{8}{4-j}-\binom{8}{3-j}$; this yields $\dim S^0 = 1$, $\dim S^1=7, \dim S^2=20, \dim S^3=28$, and $\dim S^4 = 14$.

Using this decomposition, we now identify the $S_8$-representation structure of the code $Q$.

\begin{lemma}
    As an $S_8$ representation, the code space decomposes as $Q \cong S^0\oplus S^0\oplus S^4$. Equivalently,     the space of weight-4 logical codewords $Q_4 \coloneqq Q\cap\Span\{\ket{s}\,|\,s\in\mathbb{F}_2^8, |s|=4\} = \Span\{\ket{\bar{x}}\,|\,x\in\PG\}$ is given by $Q_4\cong \mathbb{C}\ket{t}\oplus S^4$.
\end{lemma}

\begin{proof}
    The vector $\ket{\bar{0}}$ is invariant under all physical permutations, so $\mathbb{C}\ket{\bar{0}}$ forms a trivial irrep $S^0$ of $S_8$. We now focus on the space $Q_4$, which by \cref{Thm:S8code} is likewise invariant under $S_8$. Every state in $Q_4$ is a linear combination of weight-4 computational basis vectors, so lies in the $J_z = 0$ sector. Since each spin-$j$ irrep contains a unique vector with $J_z=0$, the weight-4 subspace decomposes as
    \begin{equation}
        \Span\{\ket{s}\;|\; s\in\mathbb{F}_2^8, |s|=4\}\cong S^0\oplus S^1\oplus S^2\oplus S^3\oplus S^4,
    \end{equation}
    and has dimension $\binom{8}{4}=70 = 1+7+20+28+14$. Because $\dim Q_4 = 15$, the only possible decomposition is $Q_4\cong S^0\oplus S^4$. The trivial irrep $S^0$ is readily identified as $\mathbb{C}\ket{t}$, as the state $\ket{t}$ is invariant under all permutations. 
\end{proof}

This decomposition yields a particularly simple description of the code space in terms of a small number of non-Pauli stabilisers. Firstly, the above result shows that $J^2(J^2-20)$ vanishes on the code space. Secondly, from \cref{Thm:TransversalT8}, we know that $T^{\otimes 8}$ implements an involutive logical gate, so $S^{\otimes 8} = (T^{\otimes 8})^2$ acts as the logical identity on the code. Finally, since both $\ket{\bar{0}}$ and all the line states $\ket{\ell}$ from which the logical states are constructed are invariant under a simultaneous bit flip of all physical qubits, the code space is invariant under $X^{\otimes 8}$. In fact, it can be verified that the mutual intersection of the relevant eigenspaces of these operators has dimension 16, so it must be the code space:
\begin{equation}
    Q = (V_{j=0}\oplus V_{j=4})\cap V_{S^{\otimes 8} = +1} \cap V_{X^{\otimes 8} = +1}.
\end{equation}
In other words, the non-Pauli operators $X^{\otimes 8}$, $S^{\otimes 8}$ and $J^2(J^2-20)$ form a complete set of stabilisers for the code.

\subsection{No Pauli stabiliser \texorpdfstring{$[\![8, 4]\!]$}{[[8, 4]]} phantom code exists}\label{Sec:NoPauli}
We now show that the exceptional $(\!(8, 2^4, 2)\!)$ code constructed above is necessarily non-Pauli. More precisely, no Pauli stabiliser code of type $[\![8, 4]\!]$ can admit the required $A_8\cong\GLx$ symmetry.

Recall that an $[\![n, k]\!]$ stabiliser code is specified by an Abelian subgroup $S$ of the $n$-qubit Pauli group of size $2^{n-k}$. By associating the operator $X^{a_1} Z^{b_1}\otimes\cdots\otimes X^{a_n} Z^{b_n}$ with the binary label $(\vb{a}|\vb{b})\in V = \mathbb{F}_2^n\oplus\mathbb{F}_2^n$, the group $S$ can be identified with an $(n-k)$-dimensional linear subspace $\Lambda$ of $V$ that is totally isotropic with respect to the standard symplectic form $\omega[(\vb{a}|\vb{b}), (\vb{a}'|\vb{b}')] = \vb{a}\cdot\vb{b}'+\vb{a}'\cdot\vb{b}$. This picture extends to subsystem stabiliser codes: the gauge group $\mathcal{G}$ of an $[\![n, k, r]\!]$ subsystem code, where $r\leq n-k$ is the number of gauge qubits, is associated with a non-isotropic subspace $\Gamma\leq V$ of dimension $n-k+r$, and the corresponding isotropic stabiliser subspace $\Lambda = \Gamma^\perp\cap\Gamma$ has dimension $n-k-r$ (this is a consequence of the relation $S = \mathcal{Z}(\mathcal{G})\cap\mathcal{G}$, where $\mathcal{Z}(\mathcal{G})$ is the centre of the gauge group~\cite{PhysRevLett.95.230504}).

A Clifford gate $U$ preserves the code space of a stabiliser code with stabiliser group $S$ if and only if $USU^\dagger = S$. In particular, if a qubit permutation $\sigma$ is a code automorphism, then the corresponding action on binary labels preserves $\Lambda$. Permutations act diagonally on the binary strings in $V$ by individually permuting the $X$- and $Z$-Pauli labels, i.e. $(\vb{a}|\vb{b})\mapsto(\sigma\vb{a}|\sigma\vb{b})$. Hence if a stabiliser code is invariant under a subgroup of permutations $G$, it follows that both the binary spaces $\Lambda_X\coloneqq\pi_X(\Lambda)$ and $\Lambda_Z\coloneqq\pi_Z(\Lambda)$, where $\pi_X:(\vb{a}|\vb{b})\mapsto\vb{a}$ and $\pi_Z:(\vb{a}|\vb{b})\mapsto\vb{b}$ are the coordinate projections, must be individually invariant under the same group $G$. $\Lambda_X$ and $\Lambda_Z$ are classical binary codes with automorphism groups containing $G$. When applied to subsystem codes, an identical argument shows that $\Gamma_X\coloneqq\pi_X(\Gamma)$ and $\Gamma_Z\coloneqq\pi_Z(\Gamma)$ must be classical codes with $\Aut\Gamma_X$ and $\Aut \Gamma_Z$ both containing $G$. In the particular case $G=A_n$, the following lemma places constraints on the possible forms of each of these codes.

\begin{proposition}[Proposition 6.1, Ref.~\cite{bienert-2009}]\label{Prop:A8Codes}
    Let $C$ be a binary linear code of length $n$. If $A_n\leq \Aut C$ and $n\neq 2$, then $C$ is either $\{0\}, C^{\text{rep}}_n, C^{\text{even}}_n$, or $\mathbb{F}_2^n$, where $C^{\text{rep}}_n$ is the repetition code and $C^{\text{even}}_n$ is the code consisting of all even-weight vectors.
\end{proposition}

From here, we can prove a general result concerning the construction of phantom stabiliser codes of minimal length for $k=4$.

\begin{theorem}
    No Pauli stabiliser subsystem phantom code of type $[\![8, 4, r, d\geq 2]\!]$ exists.
\end{theorem}

\begin{proof}
    Recall the subsystem Singleton bound~\cite{Klappenecker_2007}: every $[\![n, k, r, d]\!]$ stabiliser subsystem code satisfies $k+r\leq n-2d+2$. Hence, any (not necessarily phantom) $[\![8, 4, r, d\geq 2]\!]$ stabiliser subsystem code must have $r\leq 2$. From now on, we assume that this condition holds.
    
    The binary subspace $\Gamma$ corresponding to an $[\![8, 4, r, d]\!]$ stabiliser subsystem code has dimension $\dim \Gamma = n-k+r = 4+r$. If the code is phantom, then by \cref{Prop:A8Codes}, the projected codes $\Gamma_X$ and $\Gamma_Z$ must both be one of $\{0\}, C^{\text{rep}}_n, C^{\text{even}}_n$, or $\mathbb{F}_2^n$, corresponding to dimensions $0, 1, 7,$ and $8$ respectively. Because $\dim \Gamma_X, \dim\Gamma_Z\leq \dim\Gamma = 4+r\leq 6$, the spaces $\Gamma_X$ and $\Gamma_Z$ must be either $\{0\}$ or $C^{\text{rep}}_n$. This implies that
    \begin{equation}
        4+r = \dim \Gamma \leq \dim\Gamma_X +\dim\Gamma_Z \leq 2,
    \end{equation}
    a contradiction, so no such codes exist.
\end{proof}

\section{Qudit phantom codes}\label{Sec:Qudit}
In the main text, we proved a number of results concerning qubit phantom and phantom-LU codes. It is natural to ask to what extent these results continue to hold for codes made up of qudits of local dimension greater than $2$. In this section we discuss the extension of \cref{Thm:GeneralCase,Thm:PhantomBound,Thm:STPbound} to phantom codes of Galois qudits, where each local Hilbert space has a prime power dimension $q=p^t$ and computational basis states labelled by the elements of the finite field $\mathbb{F}_q$. It is possible that these results may be further generalised to phantom codes over qudits of arbitrary integer dimensions; however, we do not pursue this direction here.

The natural analogues of the $\CNOT$ gate for Galois qudits are the generalised $\SUM$ gates, defined for a system of two qudits as 
\begin{equation}
    \SUM_{12}^{(\alpha)}\ket{x, y} = \ket{x, x+\alpha y}, 
\end{equation}
where $x, y, \alpha \in\mathbb{F}_q$. In the symplectic representation of the qudit Clifford group, the $\SUM_{12}^{(\alpha)}$ operator is described by a shear matrix $I+\alpha E_{12}$, where $E_{ab}$ is the matrix with entries $(E_{ab})_{cd} = \delta_{ac}\delta_{bd}\in\mathbb{F}_q$. A general circuit composed of logical $\SUM$ operators corresponds to a product of such shear matrices, which together generate $\SLkFq$. Therefore, the group of logical $\SUM$ circuits on $k$ qudits is isomorphic to $\SLkFq$. For an element $g\in\SLkFq$, we define $U_g$ to be the unitary $\SUM$ circuit corresponding to $g$; this provides a faithful unitary representation of $\SLkFq$~\cite{ConradSimple}. 

In the qudit setting, the group of $\SWAP$-transversal gates is the semidirect product $M_n^q \coloneqq \mathsf{U}(q)^n\rtimes S_n$, consisting of products of on-site qudit unitary operators $\mathsf{U}(q)$, and permutations in $S_n$. This suggests a natural way to generalise the definition of the phantom-LU property to qudit codes.

\begin{definition}
    A $q$-ary subsystem QEC code $Q$ defined by an isometry $V: \mathcal{A}\otimes \mathcal{B}\to\mathcal{H}$ encoding $k$ logical and $r$ gauge qudits into $n$ physical qudits is \emph{phantom-LU} if, for some logical basis in $\mathcal{A}$, there exists a $\SWAP$-transversal operator $m_{ab}\in M_n^q$ for each ordered pair of logical qudits $(a, b)\in [k]^2$ with $a\neq b$, such that $m_{ab} V = V(\overline{\SUM}_{ab}\otimes w_{ab})$, where $w_{ab}\in\mathsf{U}(\mathcal{B})$. Similarly, $Q$ is said to be \textit{phantom} if it is possible to choose all $m_{ab}$ to be permutations.
\end{definition}

The following example shows that qudit phantom codes exist for all $q$ and any number of encoded logical qudits $k$.

\begin{proposition}
    The $[\![q^k, k, 2]\!]_q$ quantum Reed-Muller code~\cite{1523494}, defined as the CSS code with classical constituents $C_X = \RM_q(0, k)^\perp$ and $C_Z = \RM_q(1, k)$, where $\RM_q(r, m)$ is a $q$-ary classical Reed-Muller code, is phantom. 
\end{proposition}

\begin{proof}
    Ref.~\cite{berger-1993} proves that the permutation automorphism group of $\RM_q(r, m)$ is equal to the affine general linear group $\mathsf{AGL}_m(\mathbb{F}_q)$ for $0\leq r\leq m(q-1)-2$, and that an affine transformation $T:x\mapsto Ax+b$ in this group acts on a polynomial $f(x)\in\RM_q(r, m)$ by sending $f(x)\mapsto f(T^{-1}x)$. The result then follows as a direct generalisation of Proposition 10 of Ref.~\cite{koh2026entanglinglogicalqubitsphysical}.
\end{proof}

Having established the existence of qudit phantom codes, we now proceed to find general constraints on their parameters by analysing the structure of their automorphism groups. The following theorem is the qudit version of \cref{Thm:GeneralCase}.

\begin{theorem}\label{Thm:QuditGeneralCase}
    Let $Q$ be a $q$-ary subspace or subsystem QEC code of length $n$. If there exist subgroups $N\triangleleft\, G\leq\PAut Q$ such that $S\coloneqq G/N$ is non-Abelian and simple, then $n\geq \mu(S)$.
\end{theorem}

\begin{proof}
    Because $G$ is faithfully embedded in the permutation automorphism group of $Q$, it follows that $n\geq \mu(G)$. Since $S=G/N$ is simple and non-Abelian, we can apply \cref{Thm:muGN} to obtain $n\geq \mu(G) \geq \mu(G/N) = \mu(S)$.
\end{proof}

In contrast to \cref{Thm:GeneralCase}, which holds for compact subgroups $\mathcal{N}\triangleleft\mathcal{G}\leq \Aut Q$ of the full (i.e. $\SWAP$-transversal) automorphism group of $Q$, \cref{Thm:QuditGeneralCase} holds only for subgroups of the permutation automorphism group $\PAut Q$ of $Q$. This means that compactness conditions on $G$ and $N$, along with the assumption that $S = G/N$ is finite and not equal to $A_5$, are no longer required. The reason for this difference is that \cref{Thm:GeneralCase} relies upon \cref{Lem:U2A5}, which in turn makes use of the classification of the finite subgroups of $\mathsf{U}(2)$. In order to generalise this lemma to qudits, we would require a classification of the finite, simple, non-Abelian sections of $\mathsf{U}(q)$ for each prime power $q$, a significantly more challenging problem than that posed by the binary case alone. In fact, because every finite group has a faithful finite-dimensional unitary representation --- equivalently, every finite group appears as a subgroup of $\mathsf{U}(q)$ for some $q$ --- it is not possible to formulate a condition such as this that holds for all prime powers $q$ simultaneously.  

Despite the slightly more limited applicability of \cref{Thm:QuditGeneralCase}, we are still able to prove a general bound that applies to all subspace and subsystem qudit phantom codes. 

\begin{theorem}\label{Thm:Qudit}
    A subspace or subsystem qudit phantom code over $\mathbb{F}_q$ encoding $k$ logical qudits, with $(k, q)\neq(2, 2)$ or $(2, 3)$, requires at least $n\geq \mu[\PSLkFq]$ physical qudits, where~\cite{mazurov_minimal_1993}
    \begin{equation}\label{Eq:PSL}
        \mu[\PSLkFq] = \begin{cases}
            5 & (k, q) = (2, 5)\\
            7 & (k, q) = (2, 7)\\
            6 & (k, q) = (2, 9)\\
            11 & (k, q) = (2, 11)\\
            8 & (k, q) = (4, 2)\\
            \frac{q^k-1}{q-1} & \text{otherwise.}
        \end{cases}
    \end{equation} 
\end{theorem}

\begin{proof}
The proof of \cref{Thm:Qudit} follows the same outline as that of \cref{Thm:PhantomBound}. Let $V:\mathcal{A}\otimes\mathcal{B}\to\mathcal{H}$ be the isometry defining $Q$, and let
\begin{equation}
    \Gq \coloneqq \left\{\sigma\in S_n\;\middle|\; \sigma V = V(U_g\otimes w)\text{ for some }g\in\SLkFq, w\in\mathsf{U}(\mathcal{B}) \right\},
\end{equation}
be the set of all permutations implementing a logical $\SUM$ circuit modulo a gauge transformation in $\mathcal{B}$. Define $\Pi_q:\Gq\to\PSLkFq$ as the map sending $\sigma\in S_n$ to the equivalence class $[g]\in\PSLkFq$, where $g\in\SLkFq$ satisfies $\sigma V = V(U_g\otimes w)$ for some $w\in\mathsf{U}(\mathcal{B})$. (We define this map into $\PSLkFq$ rather than $\SLkFq$ as the latter group is not always simple, so generally does not permit the application of \cref{Thm:QuditGeneralCase}.) This map is well defined, for if $U_g\otimes w = U_{g'}\otimes w'$ then $U_g^\dagger U_{g'}\otimes \mathbbm{1} = \mathbbm{1}\otimes w (w')^\dagger.$ The left-hand side lies in $\End(\mathcal{A})\otimes \mathbbm{1}$, and the right-hand side in $\mathbbm{1}\otimes\End(\mathcal{B})$, so both must be scalar. Hence $U_{g'} = \lambda U_g$ for some $\lambda\in\mathsf{U}(1)$. As the $\SUM$ circuit representation of $\SLkFq$ is faithful and real in the computational basis, this is possible only if $g=g'$, which obviously implies $[g]=[g']$. For any qudit phantom code, the map $\Pi_q$ is surjective, so with $\Nq\coloneqq\ker\Pi_q$ we have $\Gq/\Nq\cong\PSLkFq$. Since $\PSLkFq$ is simple and non-Abelian whenever $k\geq 2$ and $(k, q)\neq(2, 2)$ or $(2, 3)$~\cite{ConradSimple}, we can apply \cref{Thm:QuditGeneralCase} (equivalently, \cref{Thm:muGN}) to obtain
\begin{equation}
    n \geq \mu(\Gq) \geq \mu(\Gq/\Nq) = \mu[\PSLkFq]
\end{equation}
The value of $\mu[\PSLkFq]$ is given in Theorem 1 of Ref.~\cite{mazurov_minimal_1993}; the anomalous values of $\mu[\PSLkFq]$ listed in \cref{Eq:PSL} arise from exceptional small-index subgroups of $\PSLkFq$.
\end{proof}

Finally, we comment on the possible extension of \cref{Thm:Minimal} to the qudit setting. As mentioned at the end of \cref{Sec:RMCodes}, Ref.~\cite{bardoe} describes the full submodule lattice decomposition of $\overline{\mathbb{F}}_q^{\raisebox{-.3em}{\scalebox{.7}{$P$}}}$, where $P = (\mathbb{F}_q^k\backslash\{0\})/\mathbb{F}_q^*$, under the natural action of $\GLkFq$ for any integer $k$ and any prime power $q=p^t$. Two issues make the direct application of this result to qudit codes challenging. Firstly, while Ref.~\cite{bardoe} gives an exhaustive classification of qudit codes of length $n=(q^k-1)/(q-1)$ that are invariant under the action of $\GLkFq$, it is not immediately obvious that every $\SLkFq$-invariant code of this length is captured by this result. For this reason, additional work would be required either to extend the classification to this setting, or to show that every $\SLkFq$-invariant code of length $(q^k-1)/(q-1)$ is invariant under $\GLkFq$. Secondly, when $t>1$ the submodule lattice structure described in Theorem 1 of Ref.~\cite{bardoe} is larger and more complicated than the simple inclusion chain of the classical Reed-Muller codes, meaning there is a wider variety of $\GLkFq$-invariant CSS codes than in the case $q=2$. Furthermore, it is unclear whether a CSS code constructed from some combination of these codes is capable of saturating the bound given by \cref{Thm:Qudit}. Hence, establishing the (non-)uniqueness of these candidate minimal qudit phantom codes may require significantly more effort than in the binary case.

\section{Reducing phantom-LU codes}\label{Sec:Equivalence}
The following proposition gives a necessary and sufficient condition for a phantom-LU code to be equivalent to a phantom code via a local unitary transformation.

\begin{proposition}
    A phantom-LU code with encoding isometry $V$ is LU equivalent to a phantom code if and only if there exists a local unitary $W \in \mathsf{U}(2)^n$ and a set of logical operators $m_g\in M_n$, where $g\in \GLkF$, such that (\textit{i}) $m_g V = \ee^{\ii\theta_g}V U_g$ for some $\theta_g\in\mathbb{R}$, and (\textit{ii}) $m_g = W^\dagger \sigma_g W $, where each $\sigma_g\in S_n$ is a permutation.
\end{proposition}

\begin{proof}
    Let $V'$ be a phantom code with permutation operators $\sigma_g$ obeying $\sigma_g V' = \ee^{\ii\theta_g}V' U_g$ for all $g\in \GLkF$ and some $\theta_g\in\mathbb{R}$. Any code related to $V'$ by a LU transformation has the form $V = W^\dagger V'$ for some $W\in \mathsf{U}(2)^n$. Each such $V$ is a phantom-LU code, since the operators $m_g\coloneqq W^\dagger \sigma_g W$ obey $m_g V = W^\dagger \sigma_g V' = \ee^{\ii\theta_g}W^\dagger V' U_g = \ee^{\ii\theta_g}VU_g$. Conversely, if conditions (\textit{i}) and (\textit{ii}) hold, then $V$ is LU equivalent to a phantom code.
\end{proof}

\end{document}